\documentclass[11pt]{article}

\usepackage{palatino}
\usepackage{pgfplots}
\usepackage{tikz}
\usetikzlibrary{shapes.multipart}
\usetikzlibrary{fit}
\usetikzlibrary{shapes.geometric,positioning}
\usetikzlibrary{calc}
\usetikzlibrary{patterns.meta}
\usepackage{caption}
\usepackage{subcaption}
\usepackage{url}
\usepackage{stfloats}
\usepackage{booktabs}
\usepackage{siunitx}
\usepackage[utf8]{inputenc}
\usepackage{mathtools}
\usepackage{graphicx}
\usepackage{algorithm}
\usepackage{cite}
\usepackage{amsmath,amssymb,amsfonts, amsthm}
\usepackage{algorithmic}
\usepackage{textcomp}
\usepackage{xcolor}
\usepackage{enumitem}
\usepackage{multirow}
\newtheorem{lemma}{Lemma}
\newtheorem{theorem}{Theorem}
\newtheorem{example}{Example}

\newtheorem{corollary}{Corollary}

\newcommand\rk{\normalfont{\mbox{rk}}}

\newcommand\bU{{\mbox{\bf U}}}
\newcommand\bV{{\mbox{\bf V}}}

\newcommand\bB{{\mbox{\bf B}}}

\title{The Capacity of $3$ User\\ Linear Computation Broadcast}
\author{Yuhang Yao, Syed A. Jafar\\
{\small Center for Pervasive Communications and Computing (CPCC)}\\
{\small University of California Irvine, Irvine, CA 92697}\\
{\small \it Email: \{yuhangy5, syed\}@uci.edu}
}
\date{}      

\begin{document}
\maketitle

\begin{abstract}
The  $K$ User Linear Computation Broadcast (LCBC) problem is comprised of $d$ dimensional data (from $\mathbb{F}_q$), that is fully available to a central server, and $K$ users, who require various linear computations of the data, and have prior knowledge of various linear functions of the data as side-information. The optimal broadcast cost is the minimum number of $q$-ary symbols to be broadcast by the server per computation instance, for every user to retrieve its desired computation. The reciprocal of the optimal broadcast cost is called the capacity. The main contribution of this paper is the exact capacity characterization for the $K=3$ user LCBC for all cases, i.e., for arbitrary finite fields $\mathbb{F}_q$, arbitrary data dimension $d$, and arbitrary linear side-informations and  demands at each user. A remarkable aspect of the converse (impossibility result) is that unlike the $2$ user LCBC whose capacity was determined previously, the entropic formulation (where the entropies of demands and side-informations are specified, but not their functional forms) is insufficient to obtain a tight converse for the $3$ user LCBC. Instead, the converse exploits functional submodularity. Notable aspects of achievability include sufficiency of vector linear coding schemes,  subspace decompositions that parallel those found previously by Yao Wang  in degrees of freedom (DoF) studies of wireless broadcast networks, and efficiency tradeoffs that lead to a constrained waterfilling solution. Random coding arguments are invoked to resolve compatibility issues that arise as each user  has a different view of the subspace decomposition, conditioned on its own side-information. 
\end{abstract}
{\let\thefootnote\relax\footnote{{Presented in part at the 2022 IEEE International Symposium on Information Theory  (ISIT 2022) \cite{Yao_Jafar_ISIT22}.}}\addtocounter{footnote}{-1}}
 
\allowdisplaybreaks

\section{Introduction}
Recent years have seen explosive growth both in the number of devices connected to communication networks, as well as in the amount of data generated, shared, and collaboratively processed by these devices. With machine communication expected to dominate human communication, future communication networks will increasingly be used in the service of computation tasks \cite{Samsung_6G}. Along with the processing power of connected devices, a key determining factor of the potential of these `\emph{computation networks}' will be the fundamental limit of their communication-efficiency. Despite a multitude of advances spanning several decades \cite{Witsenhausen, Korner_Marton_sum, Ahlswede_Cai_CC, Orlitsky_Roche, Doshi_Shah_Medard_Effros, Rai_Dey, Giridhar_Kumar, Kowshik_Kumar, Ramamoorthy_Langberg, Feizi_Medard, Appuswamy1, Appuswamy2, Appuswamy3, Guang_Yeung_Yang_Li}, the capacity limits of computation networks remain largely unknown. Remarkably, this is the case even in the most basic of scenarios such as computational multiple access and broadcast, the presumptive starting points for developing a cohesive theory of computation networks. It is also noteworthy that many applications of recent interest, such as coded caching \cite{Maddah_Ali_Niesen, Yu_Maddah_Avestimehr_exact, Caching_D2D}, private information retrieval \cite{PIRfirstjournal,  PIR_tutorial}, coded MapReduce \cite{CompCommTradeoff}, distributed storage exact repair \cite{Dimakis_survey, DSS_overview, toni}, index coding \cite{Birk_Kol_Trans, Young_Han_FnT}, coded computing \cite{CodedComputing_Survey, Coding_Fog_Computing}, data shuffling \cite{Lee_Lam_Pedarsani}, federated learning \cite{FLsurvey1} and secure aggregation \cite{FastSecAgg}, are essentially linear computation multiple access (LCMAC) or broadcast (LCBC) settings with additional application-specific constraints. Future developments, say in networked VR/AR technology \cite{Networked_VR, Saad_VR1,Medard_VR, Zhao_Cui_Yang_Shamai, Saad_6G,Samsung_6G}, will similarly need linear broadcast and multiple access computational networks for coordination and synchronization \cite{HUYGENS, Fei_Chen_Wang_Jafar, Cuff_coordination} of users' perspectives across space, typically computed as linear projections of real-world coordinates. Evidently, beyond their significance as building blocks, LCMAC/LCBC networks are  important in and of themselves.

The collaborative, task-oriented, and interactive character of computation networks manifests in \emph{data dependencies}, and an abundance of \emph{side-information} accumulated at each node from prior computations on overlapping datasets. Both data dependencies and side-information significantly impact the capacity of computation networks. Furthermore,  because of the inherently algorithmic character of machine communication, the underlying structures of data dependencies and side-information are often predictable, and may be exploited in principled ways to improve communication efficiency. Indeed, both of these aspects are central to the computation broadcast (CBC) problem, an elemental one-to-many computation network  studied recently in \cite{Sun_Jafar_CBC}. The CBC setting is comprised of data stored at a central server, and multiple users, each of whom is given some function of the data as side-information and wishes to retrieve some other function of the data. The goal  is to find $\Delta^*$,  which is the least amount of information per computation that the server must broadcast such that all the users are able to compute their respective desired functions.  The capacity of CBC is defined as $C=1/\Delta^*$.

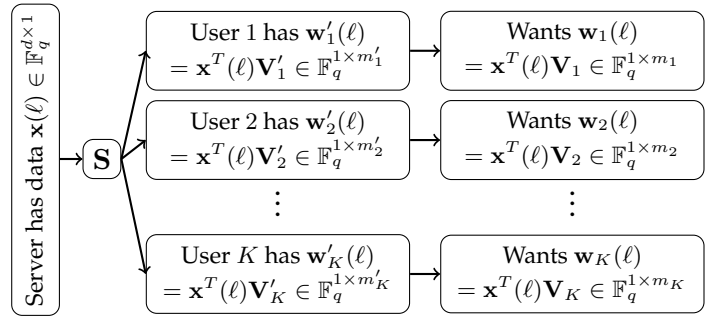
\begin{figure}[!h]
\center
\def\colh{white}
\def\colw{white}
\begin{tikzpicture}[yscale=0.9, xscale=0.9]
\def\h{1.1}
\node at (-1.6,-1.6)  [draw=black, rounded corners, fill=\colw] (S)  { ${\bf S}$};
\node[draw=black, fill=\colh,  rounded corners,minimum width=3cm, rotate=90, above left=1.8cm and 0.3cm of S] (T)  [align=center]  {\footnotesize Server has data  ${\bf x}(\ell)\in\mathbb{F}_q^{d\times 1}$};
\node[draw=black, fill=\colh, rounded corners, align=center, minimum width=3.5cm] (U1) at (1,0*\h) [align=center] {\footnotesize User $1$ has ${\bf w}_1'(\ell)$\\[-0.1ex] \footnotesize $={\bf x}^T(\ell){\bf V}_1'\in\mathbb{F}_q^{1\times m_1'}$};
\node[draw=black, fill=\colh, rounded corners, align=center, minimum width=3.5cm] (U2) at (1,-1.2*\h)  [align=center] {\footnotesize User $2$ has ${\bf w}_2'(\ell)$\\[-0.1ex] \footnotesize $={\bf x}^T(\ell){\bf V}_2'\in\mathbb{F}_q^{1\times m_2'}$};
\node[draw=black, fill=\colh, rounded corners, align=center, minimum width=3.5cm] (U3) at (1,-3*\h)  [align=center] {\footnotesize User $K$ has ${\bf w}_K'(\ell)$\\[-0.1ex] \footnotesize $={\bf x}^T(\ell){\bf V}_K'\in\mathbb{F}_q^{1\times m_K'}$};
\node[draw=black, right=0.4cm of U1, align=center, minimum width=3.5cm,  rounded corners, fill=\colw](W1)  {\footnotesize Wants ${\bf w}_1(\ell)$\\[-0.1ex] \footnotesize $={\bf x}^T(\ell){\bf V}_1\in\mathbb{F}_q^{1\times m_1}$};
\node[draw=black, right=0.4cm of U2, align=center,  minimum width=3.5cm, rounded corners, fill=\colw](W2) {\footnotesize Wants ${\bf w}_2(\ell)$\\[-0.1ex] \footnotesize $={\bf x}^T(\ell){\bf V}_2\in\mathbb{F}_q^{1\times m_2}$};
\node[draw=black, right=0.4cm of U3, align=center, minimum width=3.5cm, rounded corners, fill=\colw](W3)  {\footnotesize Wants ${\bf w}_K(\ell)$\\[-0.1ex] \footnotesize $={\bf x}^T(\ell){\bf V}_K\in\mathbb{F}_q^{1\times m_K}$};
\node[below=-0.2cm of W2] {\large $\vdots$};
\node[below=-0.2cm of U2] {\large $\vdots$};
\draw[thick, ->] (S.west-|T.south) -- (S.west);
\draw[thick, ->] (U1) -- (W1);
\draw[thick, ->] (U2) -- (W2);
\draw[thick, ->] (U3) -- (W3);
\draw[thick, ->] (S.east) -- (U1.west);
\draw[thick, ->] (S.east) -- (U2.west);
\draw[thick, ->] (S.east) -- (U3.west);
\end{tikzpicture}

\caption{\small The $\ell^{th}$  instance of the LCBC$\Big(\mathbb{F}_q, K,d,(m_k,m_k')_{k\in[K]},({\bf V}_k,{\bf V}_k')_{k\in[K]}\Big)$, $\ell\in\mathbb{N}$.}\label{fig:LCBC}
\end{figure}

The main result of \cite{Sun_Jafar_CBC} is an exact  capacity characterization for $K=2$ user \emph{linear} computation broadcast (LCBC), where each user's demand and side-information are linear functions of the data. The $K$ user LCBC problem, illustrated in Figure \ref{fig:LCBC}, is specified by the parameters $\big(\mathbb{F}_q, K,d,(m_k,m_k')_{k\in[K]},({\bf V}_k,{\bf V}_k')_{k\in[K]}\big)$, namely the finite field $\mathbb{F}_q$, the number of users $K$, the data dimension $d$, matrices ${\bf V}_k\in\mathbb{F}_q^{d\times m_k}$ that identify the $m_k$ dimensions of User $k$'s demand, and matrices ${\bf V}_k'\in\mathbb{F}_q^{d\times m_k'}$ that identify the $m_k'$ dimensions of User $k$'s side-information, for all $k\in[K]$. The index $\ell\in\mathbb{N}$ in Figure \ref{fig:LCBC} identifies the $\ell^{th}$ computation instance, corresponding to the $\ell^{th}$ instance of the data vector, ${\bf x}(\ell)\in\mathbb{F}_q^{d\times 1}$, such that User $k$, $k\in[K]$, wants ${\bf w}_k(\ell)={\bf x}^T(\ell){\bf V}_k$ and has ${\bf w}_k'(\ell)={\bf x}^T(\ell){\bf V}_k'$ as side-information.  Following a typical information theoretic formulation, multiple (say $L$)  instances may be considered jointly by a coding scheme for potential gains in efficiency. $L$ is called the batch size and may be chosen freely by a coding scheme. A coding scheme that satisfies all the users' demands across $L$ computation instances by broadcasting a total of $N$ $q$-ary symbols, achieves rate $R=L/N$, and broadcast cost per computation $\Delta = N/L=1/R$. The goal is to find the supremum of achievable rates (capacity $C$), or equivalently, the infimum of achievable broadcast costs per computation ($\Delta^*=1/C$) across all feasible coding schemes.
 We refer the reader to Section \ref{sec:notation} to clarify notational aspects, and to Section \ref{sec:probstat} for details of the  problem formulation. For $K=2$ users, the optimal  broadcast cost  is found in \cite{Sun_Jafar_CBC} to be  $\Delta^*= \max\Big(\rk_q([{\bf V}_i, {\bf V}_i'])-\rk_q({\bf V}_i')+ \rk_q([{\bf V}_1, {\bf V}_2, {\bf V}_1', {\bf V}_2'])-\rk_q([{\bf V}_i, {\bf V}_i', {\bf V}_j'])\Big)$, where $\rk_q()$ is the matrix rank function  over $\mathbb{F}_q$, and the $\max$ is over { $(i,j)\in\{(1,2), (2,1)\}$}.

The scope of LCBC includes problems such as index coding \cite{Yossef_Birk_Jayram_Kol_Trans, Birk_Kol_Trans, Rouayheb_Sprintson_Georghiades} that have been extensively studied and yet remain open in general. While many instances of index coding have been solved from a variety of perspectives \cite{Young_Han_FnT, Berliner_Langberg, Ong_single_uniprior, Jafar_TIM, Agarwal_Flodin_Mazumdar,Blasiak_Kleinberg_Lubetzky_2010 }, little is known about the optimal broadcast cost for the general index coding problem. It is shown in \cite{Birk_Kol_Trans} that for \emph{scalar linear}  index coding, the optimal broadcast cost  can be found in general by solving a min-rank problem. The min-rank solution has been extended to index coding with coded side-information in \cite{Lee_Dimakis_Heath} and is not difficult to further generalize to LCBC. However, on top of the difficulty of matrix rank minimizations (known to be NP-hard \cite[Thm. 3.1]{Peeters_min_rank}, \cite{chawin2023improved, langberg2011hardness}), scalar linear coding is only one of many possible coding schemes, and it is well known that capacity achieving schemes need not be scalar or linear, even for index coding \cite{Lubetzky_Stav, Blasiak_Kleinberg_Lubetzky_2011, Maleki_Cadambe_Jafar}. Thus, finding the capacity of LCBC in general is at least as hard as solving the general index coding problem.

On the other hand,  index coding problems constitute only a small subset of all possible LCBC instances. The special cases of LCBC that yield index coding problems are precisely those where all the columns of ${\bf V}_k,{\bf V}_k'$ can be represented as standard basis vectors. Evidently, LCBC allows a  significantly richer research space for developing new insights. This is why for LCBC, even settings with only $2,3$ users are interesting and insightful, whereas such settings would be trivial for index coding. The richer space of LCBC problems is particularly valuable if it  is  amenable to information theoretic analysis. Intrigued by this possibility, in this work we explore what new technical challenges  might emerge in the LCBC setting when we go from $2$ to $3$ users.  

The main result of this work is the exact capacity of the $3$ user LCBC for all cases, i.e., for arbitrary $\mathbb{F}_q$, arbitrary data dimension $d$, and arbitrary demands and side-informations $\bV_k, \bV_k'$ for each user, $k\in\{1,2,3\}$. An explicit expression for the capacity, $C$, is presented in Theorem \ref{thm:main}, and depends on the dimensions (ranks) of various unions and intersections of subspaces corresponding to the users' desired computations and side-information. The intuition behind the explicit form becomes more transparent when it is viewed as the solution to a linear program, in an alternative formulation of the capacity result, presented in Theorem \ref{thm:LP}. The linear program sheds light on the key ideas behind the optimal coding scheme. One of these ideas is a decomposition of the collective signal spaces of the three users (column spans of the $[{\bf V}_k',\bV_k]$ matrices) into distinct subspaces that allow different levels of communication efficiency. Remarkably, this decomposition,  which is formalized in Lemma \ref{lemma:decomposition}, closely parallels (see Appendix \ref{app:comparison}) a corresponding decomposition previously obtained in degrees of freedom (DoF) studies of the $3$ user MIMO broadcast channel in \cite{wang2017degrees},  underscoring its fundamental significance. Facilitated by the subspace decomposition, the linear program formulation of Theorem \ref{thm:LP} reveals a non-trivial tradeoff between the number of dimensions of broadcast that are drawn from each subspace, and leads to a constrained waterfilling solution in Section \ref{sec:solveLP}. What makes the achievability especially challenging is that the users have different (seemingly incompatible) views of the useful information within each subspace depending on their respective side-informations. Random coding arguments are invoked to find broadcast dimensions  for the optimal scheme that are useful across the different perspectives. Another remarkable aspect of the capacity result is that non-linear schemes are not needed for the $3$ user LCBC. While our optimal schemes  make use of both field size extensions (Section \ref{subsec:field_size}) and  matrix extensions (Section \ref{sec:matex}), they are still vector linear schemes over $\mathbb{F}_q$. In contrast, \emph{scalar} linear codes were found to be sufficient for the $2$ user LCBC in \cite{Sun_Jafar_CBC}. 
In terms of the converse bound,\footnote{\label{fn:converse}A \emph{converse} bound refers to an impossibility result, i.e., a lower bound on broadcast cost per computation, or equivalently, an upper bound on capacity.} an interesting insight emerges from this work regarding the entropic formulation of the LCBC problem that was considered in \cite{Sun_Jafar_CBC}. In the entropic formulation of \cite{Sun_Jafar_CBC}, the data is assumed i.i.d. uniform, and the entropies of all subsets of demand and side-information random-variables are specified as constraints, but  their functional forms are not specified. It was shown in \cite{Sun_Jafar_CBC} that for the $K=2$ user LCBC, these entropic constraints combined with standard Shannon entropic inequalities  produce a tight converse bound on the download cost per computation. In contrast, in this work we show (see \ref{rem:funcsub} in Section \ref{sec:examples}) by a counterexample that the same approach cannot work for the $K=3$ user LCBC, even if \emph{all} Shannon and non-Shannon information  inequalities are applied. Instead, in this work a tight converse for the $3$ user LCBC is obtained  based on functional submodularity (cf. \cite{Tao_FS, Kontoyiannis_Madiman}, also Lemma \ref{lemma:tao}  in this work) that additionally takes into account the functional forms of the users' demands and side-informations.

\section{Notation}\label{sec:notation}
$\mathbb{F}_q$ is a finite field with $q=p^n$ a power of a prime. The elements of the prime field $\mathbb{F}_p$ are represented as $\mathbb{Z}/p\mathbb{Z}$, i.e.,  integers modulo $p$. The notation $\mathbb{F}_q^{n_1\times n_2}$ represents the set of $n_1\times n_2$ matrices with elements in $\mathbb{F}_q$. For a matrix $M$, let $\langle M\rangle_q$ denote the $\mathbb{F}_q$-linear vector space spanned by the columns of $M$. The subscript $q$ will often be suppressed to simplify notation when it is clear from the context. The notation $M_1\cap M_2$ represents a matrix whose columns form a basis of $\langle M_1\rangle\cap\langle M_2\rangle$. $[M_1, M_2]$ represents a concatenated matrix which can be partitioned column-wise into $M_1$ and $M_2$.  The rank of $M$ over $\mathbb{F}_q$ is denoted by $\rk_{q}(M)$, and when written as $\rk(M)$ for simplicity, the subscript $q$ is assumed by default. If $\rk(M)$ is equal to the number of columns of $M$, i.e., $M$ has full column rank, then we say that $M$ is a basis of $\langle M\rangle$. Define a `conditional-rank' notation as $\rk(M_1|M_2) \triangleq \rk([M_1,M_2])-\rk(M_2)$. The notation $[n]$ represents the set $\{1,2,\cdots, n\}$.   $\mathbb{N}$ denotes the set of positive integers. $\mathbb{R}_+$ denotes the set of non-negative real numbers. $\mathbb{C}$ denotes the set of complex numbers.

\section{Problem Statement} \label{sec:probstat}
\subsection{The General $K$ User LCBC$(\mathbb{F}_q, K,d, (m_k, m_k')_{k\in[K]},$ $ ({\bf V}_k, {\bf V}_k')_{k\in[K]})$} \label{sec:genstat}
While our focus in this work is exclusively on the $K=3$ user case, in this section let us define the LCBC problem for the general $K$-user setting. As noted previously, the general LCBC problem is specified by the parameters $\big(\mathbb{F}_q, K,d, (m_k, m_k')_{k\in[K]}, ({\bf V}_k, {\bf V}_k')_{k\in[K]}\big)$. 
There is a stream of data vectors ${\bf x}(1), {\bf x}(2), \cdots$ that is available at a central server. For  each $\ell\in \mathbb{N}$, ${\bf x}(\ell) = (x_1(\ell),\cdots, x_d(\ell))^T\in \mathbb{F}_q^{d\times 1}$ is a $d$-dimensional vector with elements in ${\Bbb F}_q$.
There are $K$ users.  For all $\ell\in\mathbb{N}$, the $k^{th}$ user has side-information ${\bf w}_k'(\ell)={\bf x}^T(\ell){\bf V}_k'\in\mathbb{F}_q^{1\times m_k'}$ and wants ${\bf w}_k(\ell)={\bf x}^T(\ell){\bf V}_k\in\mathbb{F}_q^{1\times m_k}$.

\subsubsection{Coding Scheme}\label{sec:codingscheme}
A coding scheme for the LCBC is represented by a choice of parameters in the form of a tuple $\big(L,N,\Phi,(\Psi_k)_{k\in[K]}\big)$. The coding scheme aggregates $L$ instances of data, collectively denoted as ${\bf X} \triangleq ({\bf x}(1),  \cdots, {\bf x}(L))\in\mathbb{F}_q^{d\times L}$, and specifies an encoding function (encoder) $\Phi: {\Bbb F}_q^{d\times L}\rightarrow \mathbb{F}_q^{N}$, as well as $K$ decoding functions (decoders) $\Psi_{k}: \mathbb{F}_q^{N} \times \mathbb{F}_q^{L\times m_k'}\rightarrow\mathbb{F}_q^{L\times m_k}$, $k\in[K]$.

For compact notation, let us define, 
\begin{align}
	{\bf W}_k&\triangleq{\bf X}^T{\bf V}_k=({\bf w}_k(1),\cdots,{\bf w}_k(L))^T\in\mathbb{F}_q^{L\times m_k},\\
{\bf W}_k'&\triangleq{\bf X}^T{\bf V}_k'=({\bf w}_k'(1),\cdots,{\bf w}_k'(L))^T\in\mathbb{F}_q^{L\times m_k'}.
\end{align}
The encoder $\Phi$  maps the data ${\bf X}$ to the broadcast information comprised of $N$ symbols in $\mathbb{F}_q$, represented compactly as ${\bf S}\in\mathbb{F}_q^N$, i.e., 
\begin{align}
\Phi({\bf X})={\bf S}\in\mathbb{F}_q^N.\label{eq:encoderdef}
\end{align}
The $k^{th}$ decoder, $\Psi_{k}$ allows the $k^{th}$ user to retrieve  ${\bf W}_k$ from the broadcast information ${\bf S}$ and the side-information ${\bf W}_k'$, i.e.,
\begin{align}
\Psi_k({\bf S}, {\bf W}_k') = {\bf W}_k,&& \forall k\in[K],\label{eq:dec}
\end{align}
for all realizations of ${\bf X}$.

Let us denote the set of all feasible coding schemes as $\mathfrak{C}$. We  refer to coding schemes with batch size $L=1$ as \emph{scalar} (coding) schemes, and those with $L>1$ as \emph{vector} (coding) schemes. 

\subsubsection{Capacity $(C)$ and Optimal Download Cost per Computation $(\Delta^*)$}\label{sec:capacitydef}
The rate of a coding scheme $\big(L,N,\Phi,(\Psi_k)_{k\in[K]}\big)\in\mathfrak{C}$, is defined as $R=L/N$ representing the number of computation instances satisfied by the coding scheme per broadcast symbol.\footnote{Viewing each $q$-ary broadcast symbol as one channel-use, the rate can be equivalently viewed as the number of computation instances satisfied by the coding scheme per channel-use.} The supremum of rates across all feasible coding schemes in $\mathfrak{C}$ is called the capacity of LCBC, i.e.,
\begin{align}
C\triangleq \sup_{\big(L,N,\Phi,(\Psi_k)_{k\in[K]}\big)\in\mathfrak{C}}L/N.\label{eq:capacitydef}
\end{align}
Instead of rate $R$, it is often more convenient to consider its reciprocal value, the broadcast cost per computation, $\Delta=1/R=N/L$. The optimal broadcast cost per computation, $\Delta^*$ is defined as,
\begin{align}
\Delta^*&\triangleq \inf_{\big(L,N,\Phi,(\Psi_k)_{k\in[K]}\big)\in\mathfrak{C}}N/L\label{eq:deltadef1}\\
&=1/C.\label{eq:deltadef2}
\end{align}
Since $\Delta^*=1/C$, the problem of characterizing the capacity $C$ is equivalent to the problem of characterizing the optimal download cost per computation $\Delta^*$. We will find it more convenient to state and prove results in terms of $\Delta^*$ in this work.

\subsubsection{Data Distribution, Entropy} \label{sec:entropy}
Note that the LCBC problem does not specify any distribution for the data ${\bf X}$. This is because the capacity $C$ and optimal broadcast cost per computation $\Delta^*$ do not depend on the data distribution. By definition, any coding scheme $\big(L,N,\Phi,(\Psi_k)_{k\in[K]}\big)\in\mathfrak{C}$, while broadcasting no more than $N$ $q$-ary symbols, must guarantee successful decoding as in\eqref{eq:dec} for \emph{every} realization of the data, i.e., for all $q^{dL}$ realizations of ${\bf X}\in\mathbb{F}_q^{d\times L}$, regardless of what distribution ${\bf X}$ follows, and even if ${\bf X}$  follows no distribution. This  is significant  for computation tasks. Recall that conventional communication scenarios are comprised of independent messages that can be compressed prior to communication to reduce the size of the task from the outset and subsequently uncompressed upon successful reception. In principle optimal compression produces uniformly distributed data (otherwise further compression would be possible), thus justifying the common assumption that messages are uniformly distributed. For the LCBC, however, while the desired computation is a linear function of the original uncompressed data, it may  no longer be linear after compression. Thus, compression to uniformly distributed data cannot be taken for granted. Furthermore, it is  often the  case that the data distribution is either unknown, or the data is truly arbitrary. Therefore, assuming that data follows a particular distribution may be overly restrictive for computation problems. Such considerations motivate the conservative formulation presented above, which requires strong (maximum rather than average) communication cost guarantees, i.e., any achievable coding scheme must guarantee that a broadcast of $N$ symbols suffices for every data realization, regardless of the distribution of ${\bf X}$.

On the other hand, it will be occasionally useful, primarily as a thought-experiment, to consider hypothetically what might happen if the data followed an i.i.d. uniform distribution. Similar to genie-aided proofs, such thought-experiments are useful to construct converse bounds (impossibility results) by the following reasoning. Given any coding scheme $\big(L,N,\Phi,(\Psi_k)_{k\in[K]}\big)\in\mathfrak{C}$, we wish to find lower bounds on the broadcast cost $N$. As a thought-experiment,  suppose the data ${\bf X}$ follows a distribution $P_{{\bf X}}({\bf x})$ and this coding scheme is used. This imparts a corresponding distribution to the broadcast symbol ${\bf S}$, say $P_{\bf S}({\bf s})$. However, since ${\bf S}\in\mathbb{F}_q^{N}$ by the definition of the coding scheme, the entropy $H({\bf S})\leq \log_q|\mathbb{F}_q^N|=N$ in $q$-ary units,  which produces a lower bound on the broadcast cost, i.e., $N\geq H({\bf S})$. Thus any choice of $P_{{\bf X}}({\bf x})$ facilitates entropic analysis and leads to a  lower bound on $N$, by calculating the  entropy of ${\bf S}$ produced by the coding scheme. The quality of the bound depend on the choice of $P_{{\bf X}}({\bf x})$. For example, if we assume the data is deterministic, then so is ${\bf S}$, i.e., $H({\bf S})=0$, leading to the bound $N\geq 0$, which is not very useful. Uniform distributions are particularly interesting because they tend to produce good converse bounds. In preparation for the converse arguments in the sequel, it is useful to recall the following facts.

\begin{enumerate}
\item For a random variable ${Z}$, that takes values in a set $\mathcal{Z}$ according to the probability mass function $p_Z(z)$, the entropy $H(Z)$ in $q$-ary units is defined as,
\begin{align}
H(Z)&\triangleq -\sum_{z\in\mathcal{Z}}p_Z(z)\log_qp_z(z).\label{eq:entdef}
\end{align}
\item If $Z$ is i.i.d. uniform over $\mathbb{F}_q^{\mu\times\nu}$ then $H(Z)=\log_q|\mathbb{F}_q^{\mu\times\nu}|=\log_q(q^{\mu\nu})=\mu\nu$ in $q$-ary units.
\item If $Z$ is i.i.d. uniform over $\mathbb{F}_q^{\mu\times\nu}$ and $M\in\mathbb{F}_q^{\mu\times \xi}$ is a deterministic matrix, then 
\begin{align}
H(Z^TM)=\nu\cdot\rk_q(M)\label{eq:entrank}
\end{align} in $q$-ary units. This is seen as follows. Let $Z_{*i}$ denote the $i^{th}$ column of $Z$. Then $H(Z^TM)=H(Z_{*1}^TM, Z_{*2}^TM,\cdots, Z_{*\nu}^TM)=\sum_{i=1}^\nu H(Z_{*i}^TM)\stackrel{(a)}{=}\sum_{i=1}^\nu\rk_q(M)=\nu \cdot \rk_q(M).$ The step labeled (a) is a direct application of \cite[Lemma 2]{Sun_Jafar_CBC}. 
\item If $Z$ is i.i.d. uniform over $\mathbb{F}_q^{\mu\times\nu}$ and $M_1\in\mathbb{F}_q^{\mu\times \xi_1}$, $M_2\in\mathbb{F}_q^{\mu\times \xi_2}$ are deterministic matrices, then 
\begin{align}
H(Z^TM_1\mid Z^TM_2)=\nu(\rk_q([M_1,M_2])-\rk_q(M_2))=\nu\cdot\rk_q(M_1\mid M_2)\label{eq:entcon}
\end{align}
in $q$-ary units, where we used the conditional-rank notation $\rk_q(M_1\mid M_2)$ as defined in Section \ref{sec:notation}. Using \eqref{eq:entrank}, this is seen as follows: $H(Z^T{M_1}\mid {Z}^T{M}_2)=H({Z}^T{ M}_1, {Z}^T{M}_2)-H(Z^T{M}_2) = H({Z}^T[{M}_1, { M}_2])-H({Z}^T{ M}_2)=\nu\cdot \rk_q([{M}_1,{M}_2])-\nu\cdot \rk_q({ M}_2)=\nu\cdot \rk_q({M}_1\mid {M}_2)$.
\end{enumerate}

\subsection{Signal Spaces $\bU_1, \bU_2, \bU_3$ and their Intersections}\label{sec:Uint}
Recall that in this work our focus is on the LCBC with  $K=3$ users, i.e., the most general setting that we consider in this work is LCBC$\left(\mathbb{F}_q, 3,d, (m_k, m_k')_{k\in[3]}, ({\bf V}_k, {\bf V}_k')_{k\in[3]}\right)$. 
Let us define the spaces ${\bf U}_1, {\bf U}_2, {\bf U}_3$, associated with the $3$ users, as follows,
\begin{align}
{\bf U}_1 &\triangleq [{\bf V}_1', {\bf V}_1], &&~~~~~~~~~~~~{\bf U}_2\triangleq [{\bf V}_2', {\bf V}_2],&&{\bf U}_3\triangleq [{\bf V}_3', {\bf V}_3],\label{eq:Udef}\\
\intertext{and also define the following intersections,}
 {\bf U}_{ij} &\triangleq{\bf U}_i \cap {\bf U}_j, &&\forall i,j\in [3], i\neq j,\\
{\bf U}_{123}&\triangleq{\bf U}_1 \cap {\bf U}_2\cap{\bf U}_3,\\
{\bf U}_{i(j,k)}&\triangleq {\bf U}_i \cap [{\bf U}_j, {\bf U}_k], && \forall (i,j,k) \in \{ \mbox{\footnotesize permutations of } (1,2,3)\}.
 \end{align}
Recall that the subspaces $\langle \bU_i \rangle$ refer to the column spans of the corresponding matrices. These subspaces will be essential to the understanding of the $3$ user LCBC.

\section{Preliminary Step: Subspace Decomposition}\label{sec:decom}
For problems involving a vector space, the choice of a suitable basis representation is often an important preliminary simplification step. When multiple vector spaces are involved, it is similarly useful to explicitly partition them into independent subspaces that fit the needs of the problem. For the $3$ user LCBC, there are three vector spaces of interest, namely $\langle{\bf U}_1\rangle, \langle{\bf U}_2\rangle, \langle{\bf U}_3\rangle$, as defined in Section \ref{sec:Uint}. A suitable decomposition of these spaces into independent subspaces corresponding to various intersections is an important preliminary simplification that is the focus of this section.
To put it concisely, we need the following two lemmas regarding linear subspaces $\langle{\bf U}_1\rangle, \langle{\bf U}_2\rangle, \langle{\bf U}_3\rangle$.
\begin{lemma}[2-space decomposition] \label{lem:2_decomp}
There exist $3$ matrices, ${\bf B}_{12}, {\bf B}_{1c}$ and ${\bf B}_{2c}$ such that ${\bf B}_{12}$ is a basis of $\langle {\bf U}_{12} \rangle$, $[{\bf B}_{12}, {\bf B}_{1c}]$ is a basis of $\langle {\bf U}_1 \rangle$, $[{\bf B}_{12}, {\bf B}_{2c}]$ is a basis of $\langle {\bf U}_2 \rangle$, and $[{\bf B}_{12}, {\bf B}_1, {\bf B}_2]$ is a basis of $\langle [{\bf U}_1, {\bf U}_2] \rangle$.
\end{lemma}
\noindent Note that Lemma 1 also implies the following dimension formula,
\begin{align}
	\rk({\bf U}_1)+\rk({\bf U}_2) = \rk({\bf B}_{12})+\rk({\bf B}_{1c})+\rk({\bf B}_{12})+\rk({\bf B}_{2c})= \rk({\bf U}_{1}\cap {\bf U}_2) + \rk([{\bf U}_1, {\bf U}_2]). \label{eq:dimension_formula}
\end{align}

\noindent  A  common proof of Lemma \ref{lem:2_decomp} from a constructive perspective (e.g. \cite[Thm. 3, Ch. 3]{janich2007linear}) is based on incrementally growing a basis representation, and is summarized as follows. First one finds ${\bf B}_{12} \in {\Bbb F}_q^{d\times {\footnotesize \rk({\bf U}_{1}\cap {\bf U}_2)}}$ as a basis of $\langle {\bf U}_1 \cap {\bf U}_2 \rangle$. Then, by the basis extension theorem, one can find a submatrix ${\bf B}_{1c} \in  {\Bbb F}_q^{d\times {\footnotesize \rk({\bf U}_{1})}}$ of ${\bf U}_1$ such that $[{\bf B}_{12}, {\bf B}_{1c}]$ spans $\langle {\bf U}_1 \rangle$, and similarly a submatrix ${\bf B}_{2c} \in  {\Bbb F}_q^{d\times {\footnotesize \rk({\bf U}_{2})}}$ of ${\bf U}_2$ such that $[{\bf B}_{12}, {\bf B}_{2c}]$ spans $\langle {\bf U}_2 \rangle$. Note that $\langle {\bf B}_{2c} \rangle$ only has trivial intersection with $\langle {\bf U}_1 \rangle$ because otherwise ${\bf B}_{2c}{\bf v}+{\bf U}_1{\bf v}' = 0 \implies {\bf B}_{2c}{\bf v} \in \langle {\bf B}_{12} \rangle$ where ${\bf v}, {\bf v}'$ are non-zero vectors, which contradicts that $[{\bf B}_{12}, {\bf B}_{2c}]$ form a basis. Therefore, $[{\bf B}_{12}, {\bf B}_{1c}, {\bf B}_{2c}]$ is a basis of $\langle [{\bf U}_1,{\bf U}_2] \rangle$ since it also spans $\langle [{\bf U}_1,{\bf U}_2] \rangle$. \hfill \qed

Figure \ref{fig:2dec} illustrates the decomposition of $\langle {\bf U}_1\rangle$ and $\langle {\bf U}_2 \rangle$ by identifying $3$ subspaces, each labeled by its basis representation. 

\begin{figure}[h]
\center
\scalebox{0.75}{
\begin{tikzpicture}
\node [draw,
    circle,
    minimum size =3cm, color = blue, thick] (U2) at (0,0){};
\node [draw,
    circle,
    minimum size =3cm, color = magenta, thick] (U3) at (1.8,0){};

\node [draw,
    circle,
    minimum size =1cm, color = blue, thick] (UU2) at (-3,0){$\langle {\bf U}_1 \rangle$};
\node [draw,
    circle,
    minimum size =1cm, color = magenta, thick] (UU3) at (5,0){$\langle {\bf U}_2 \rangle$};
    
\node at (0.9,0) {${\bB}_{12}$};
\node[] at (-0.5,0) {${\bB}_{1c}$};
\node[] at (2.1,0) {${\bB}_{2c}$};
\draw [->, magenta,  thick] (UU3) to [out=130, in=30](U3);
\draw [->, blue, thick] (UU2) to [out=50, in=150](U2);
\end{tikzpicture}}
\caption{Decomposition of $\langle{\bf U}_1\rangle, \langle{\bf U}_2\rangle$ into $3$ subspaces labeled by their respective bases.}\label{fig:2dec}
\end{figure}
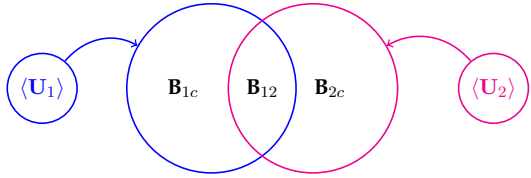

The following lemma non-trivially extends the argument to $3$ linear subspaces.
\begin{lemma}[$3$-space decomposition]\label{lemma:decomposition} There exist $10$ matrices, $\bB_{123}$, $\bB_{12}$, $\bB_{13}$, $\bB_{23}$, $\bB_{1(2,3)}$, $\bB_{2(1,3)}$, $\bB_{3(1,2)}$, $\bB_{1c}$, $\bB_{2c}$,$\bB_{3c}$, such that the following properties (P1)-(P20) are satisfied.

\begin{enumerate}[itemsep=+1ex, label=\it (P\arabic*)]
\item  ${\bB}_{123}$ is a basis of $\langle\bU_{123}\rangle$, 
\item  $[\bB_{123}, \bB_{12}]$ is a basis of $\langle\bU_{12}\rangle$,
\item $[\bB_{123}, \bB_{13}]$ is a basis of $\langle\bU_{13}\rangle$, 
\item $[\bB_{123}, \bB_{23}]$ is a basis of $\langle\bU_{23}\rangle$,
\item $[\bB_{123}, \bB_{12}, \bB_{13}]$ is a basis of $\langle[\bU_{12},\bU_{13}]\rangle$,
\item $[\bB_{123}, \bB_{12}, \bB_{23}]$ is a basis of $\langle[\bU_{12},\bU_{23}]\rangle$,
\item $[\bB_{123}, \bB_{13}, \bB_{23}]$ is a basis of $\langle[\bU_{13},\bU_{23}]\rangle$,
\item $[\bB_{123}, \bB_{12}, \bB_{13}, \bB_{1(2,3)}]$ is a basis of $\langle\bU_{1(2,3)}\rangle$,
\item $[\bB_{123}, \bB_{12}, \bB_{23}, \bB_{2(1,3)}]$ is a basis of $\langle\bU_{2(1,3)}\rangle$,
\item $[\bB_{123}, \bB_{13}, \bB_{23}, \bB_{3(1,2)}]$ is a basis of $\langle\bU_{3(1,2)}\rangle$,
\item $[\bB_{123}, \bB_{12}, \bB_{13}, \bB_{1(2,3)}, \bB_{1c}]$ is a basis of $\langle\bU_1\rangle$,
\item $[\bB_{123}, \bB_{12}, \bB_{23}, \bB_{2(1,3)}, \bB_{2c}]$ is a basis of $\langle\bU_2\rangle$, 
\item $[\bB_{123}, \bB_{13}, \bB_{23}, \bB_{3(1,2)}, \bB_{3c}]$ is a basis of $\langle\bU_3\rangle$, 
\item $[\bB_{123}, \bB_{12}, \bB_{13},\bB_{23},  \bB_{1(2,3)}, \bB_{2(1,3)},\bB_{1c}, \bB_{2c}]$ is a basis of $\langle[\bU_1,\bU_2]\rangle$,
\item $[\bB_{123}, \bB_{12}, \bB_{13},\bB_{23},  \bB_{1(2,3)}, \bB_{3(1,2)},\bB_{1c}, \bB_{3c}]$ is a basis of $\langle[\bU_1,\bU_3]\rangle$,
\item $[\bB_{123}, \bB_{12}, \bB_{13},\bB_{23},  \bB_{2(1,3)}, \bB_{3(1,2)},\bB_{2c}, \bB_{3c}]$ is a basis of $\langle[\bU_2,\bU_3]\rangle$,
\item $[\bB_{123}, \bB_{12}, \bB_{23}, \bB_{13}, \bB_{1(2,3)}, \bB_{2(1,3)}, \bB_{1c},\bB_{2c}, \bB_{3c}]$ is a basis of $\langle[\bU_1, \bU_2,\bU_3]\rangle$,
\item $[\bB_{123}, \bB_{12}, \bB_{23}, \bB_{13}, \bB_{1(2,3)}, \bB_{3(1,2)}, \bB_{1c},\bB_{2c}, \bB_{3c}]$ is a basis of $\langle[\bU_1, \bU_2,\bU_3]\rangle$,
\item $[\bB_{123}, \bB_{12}, \bB_{23}, \bB_{13}, \bB_{2(1,3)}, \bB_{3(1,2)}, \bB_{1c},\bB_{2c}, \bB_{3c}]$ is a basis of $\langle[\bU_1, \bU_2,\bU_3]\rangle$, and
\item ${\bf B}_{1(2,3)},{\bf B}_{2(1,3)},{\bf B}_{3(1,2)}$ have identical size and ${\bf B}_{1(2,3)}+{\bf B}_{2(1,3)} = {\bf B}_{3(1,2)}$.
\end{enumerate}
\end{lemma}
\noindent We leave the proof to Appendix \ref{proof:decomposition}. Figure \ref{fig:venn1} illustrates the decomposition of $\langle{\bf U}_1\rangle, \langle{\bf U}_2\rangle, \langle{\bf U}_3\rangle$ by identifying $10$ subspaces, each labeled by its  basis representation.

\begin{figure*}[hbt]
\center
\scalebox{0.85}{
\begin{tikzpicture}
\node [draw,
    circle,
    minimum size =3cm, color = olive, thick] (U1) at (0.9,1.5){};
\node [draw,
    circle,
    minimum size =3cm, color = magenta, thick] (U2) at (0,0){};
\node [draw,
    circle,
    minimum size =3cm, color = teal, thick] (U3) at (1.8,0){};
\node [draw, dashed,
    circle,
    minimum size =3.45cm, fill = yellow!40, thick] (D) at (0.9,0.52){};
\node [draw,
    circle,
    minimum size =1cm, color = blue, thick] (UU1) at (-2,2.4){$\langle {\bf U}_1 \rangle$};
\node [draw,
    circle,
    minimum size =1cm, color = magenta, thick] (UU2) at (-3,0){$\langle {\bf U}_2 \rangle$};
\node [draw,
    circle,
    minimum size =1cm, color = teal, thick] (UU3) at (5,0){$\langle {\bf U}_3 \rangle$};
    
\begin{scope}
    \clip (0,0) circle(1.5cm);
    \clip (0.9,1.5) circle(1.5cm);
    \fill[white](0,0) circle(1.5cm);
\end{scope}

\begin{scope}
    \clip (0,0) circle(1.5cm);
    \clip (1.8,0) circle(1.5cm);
    \fill[white](0,0) circle(1.5cm);
\end{scope}

\begin{scope}
   \clip (0.9,1.5) circle(1.5cm);
    \clip (1.8,0) circle(1.5cm);
    \fill[white](0.9,1.5) circle(1.5cm);
\end{scope}

\node at (0.9,0.4) {${\bB}_{123}$};
\node at (0,1) {\footnotesize ${\bB}_{12}$};
\node at (1.8,1) {\footnotesize ${\bB}_{13}$};
\node at (0.9,-0.5) {\footnotesize ${\bB}_{23}$};
\node at (0.9,1.8) {\footnotesize {${\bB}_{1(2,3)}$}};
\node[rotate=-60] at (-0.3,-0.1) {\footnotesize ${\bB}_{2(1,3)}$};
\node[rotate=45] at (2,-0.2) {\footnotesize ${\bB}_{3(1,2)}$};
\node at (0.9,2.6) {${\bB}_{1c}$};
\node[rotate=-45] at (-0.8,-0.6) {${\bB}_{2c}$};
\node[rotate=45] at (2.6,-0.6) {${\bB}_{3c}$};

\node [draw,
    circle,
    minimum size =3cm, color = blue, ultra thick] (U1) at (0.9,1.5){};
\node [draw,
    circle,
    minimum size =3cm, color = magenta, ultra thick] (U2) at (0,0){};
\node [draw,
    circle,
    minimum size =3cm, color = teal, ultra thick] (U3) at (1.8,0){};
\draw [->, blue, thick] (UU1) to [out=40, in=130](U1);
\draw [->, teal,  thick] (UU3) to [out=130, in=30](U3);
\draw [->, magenta, thick] (UU2) to [out=50, in=150](U2);

\begin{scope}[shift={(-7,-5)}]
\node at (4,2.5)[blue]{};         
\draw[clip] (0.9,1.5) circle (1.51cm);
\draw [blue, fill=black!10, ultra thick](0.9,1.5) circle (1.49cm);
\clip (1.8,0) circle (1.62cm);
\clip (0,0) circle (1.62cm);

\node [draw,
    circle,
    minimum size =3.2cm, color = magenta, ultra thick, fill=white] at (0,0){};
\node [draw,
    circle,
    minimum size =3.2cm, color = teal, ultra thick, fill=white] (U3) at (1.8,0){};  
\node [draw,
    circle,
    minimum size =3.2cm, color = magenta, ultra thick] at (0,0){};
\node [draw,
    circle,
    minimum size =3.2cm, color = teal, ultra thick] (U3) at (1.8,0){};      
\node [scale=1.5] () at (0.9,0.3){\tiny $\langle\bU_{123}\rangle$};
\node [draw,
    circle,
    minimum size =0.7cm, fill = red!10] () at (0,1){\tiny $\lambda_{12}$};
\node [draw,
    circle,
    minimum size =0.7cm, fill = red!10] () at (1.8,1){\tiny $\lambda_{13}$};
\node [draw,
    circle,
    minimum size =0.7cm, fill = red!10] () at (0.9,1.8){\tiny $\lambda$};

\draw[blue, ultra thick] (0.9,1.5) circle (1.5cm); 
\end{scope}

\begin{scope}[shift={(-3.5,-5)}]
\node at (4,2.5)[blue]{};         
\draw[clip] (0.9,1.5) circle (1.51cm);
\draw [blue, fill=black!10, ultra thick](0.9,1.5) circle (1.49cm);
\clip (0,0) circle (1.55cm);

\node [draw,
    circle,
    minimum size =3cm, color = magenta, ultra thick, fill=white] at (0,0){};
\node [draw,
    circle,
    minimum size =3cm, color = magenta, ultra thick] at (0,0){};
\node [scale=1.5] () at (0.5,0.7){\tiny {$\langle\bU_{12}\rangle$}};
\draw[blue, ultra thick] (0.9,1.5) circle (1.5cm); 
\end{scope}

\begin{scope}[shift={(0,-5)}]
\node at (4,2.5)[blue]{};         
\draw[clip] (0.9,1.5) circle (1.51cm);
\draw [blue, fill=black!10, ultra thick](0.9,1.5) circle (1.49cm);
\clip (1.8,0) circle (1.55cm);
\node [draw,
    circle,
    minimum size =3cm, color = teal, ultra thick, fill=white] (U3) at (1.8,0){};  
\node [draw,
    circle,
    minimum size =3cm, color = teal, ultra thick] (U3) at (1.8,0){};      
\node [scale=1.5] () at (1.4,0.7){\tiny $\langle\bU_{13}\rangle$};

\draw[blue, ultra thick] (0.9,1.5) circle (1.5cm); 
\end{scope}

\begin{scope}[shift={(3.5,-5)}]
\node at (4,2.5)[blue]{};         
\draw[clip] (0.9,1.5) circle (1.51cm);
\draw [blue, fill=black!10, ultra thick](0.9,1.5) circle (1.49cm);
\clip (0.9,0.52) circle (1.725cm);

\node [draw,
    circle,
    minimum size =3cm, color = magenta, ultra thick, fill=white] at (0,0){};
\node [draw,
    circle,
    minimum size =3cm, color = teal, ultra thick, fill=white] (U3) at (1.8,0){};  
\node [draw,
    circle,
    minimum size =2.85cm, color = white, ultra thick, fill=white] at (0,0){};    
\node [scale=1.5] () at (0.9,0.7){\tiny $\langle[\bU_{12},\bU_{13}]\rangle$};
\draw[blue, ultra thick] (0.9,1.5) circle (1.5cm); 
\end{scope}

\begin{scope}[shift={(7,-5)}]
\node at (4,2.5)[blue]{};         
\draw[clip] (0.9,1.5) circle (1.51cm);
\draw [blue, fill=black!10, ultra thick](0.9,1.5) circle (1.49cm);

\node [draw, dashed,
    circle,
    minimum size =3.45cm, fill = white, thick] (D) at (0.9,0.52){};
\node [scale=1.5] () at (0.9,1){\tiny $\langle\bU_{1(2,3)}\rangle$};
\draw[blue, ultra thick] (0.9,1.5) circle (1.5cm); 
\end{scope}

\end{tikzpicture}}
\caption{The top of the figure shows the decomposition of $\langle {\bf U}_1\rangle, \langle {\bf U}_2\rangle, \langle {\bf U}_3\rangle$  into $10$ subspaces that are labeled by corresponding bases  as specified in Lemma \ref{lemma:decomposition}. The five blue circles in the bottom row each show $\langle \bU_1 \rangle$, and highlight the  subspaces $\langle\bU_{123}\rangle$, $\langle\bU_{12}\rangle$, $\langle\bU_{13}\rangle$, $\langle[\bU_{12},\bU_{13}]\rangle$ and $\langle\bU_{1(2,3)}\rangle$, respectively. The compact notations, ${\bf U}_{123}, {\bf U}_{12}, {\bf U}_{1(2,3)}$ etc., are  defined in Section \ref{sec:probstat}, for example ${\bf U}_{12}\triangleq{\bf U}_1\cap{\bf U}_2$ and ${\bf U}_{1(2,3)}\triangleq {\bf U}_1\cap[{\bf U}_2,{\bf U}_3]$. The three  subspaces highlighted as dashed yellow regions are not independent, the span of the union of any two of them contains the third.}
\label{fig:venn1}
\end{figure*}
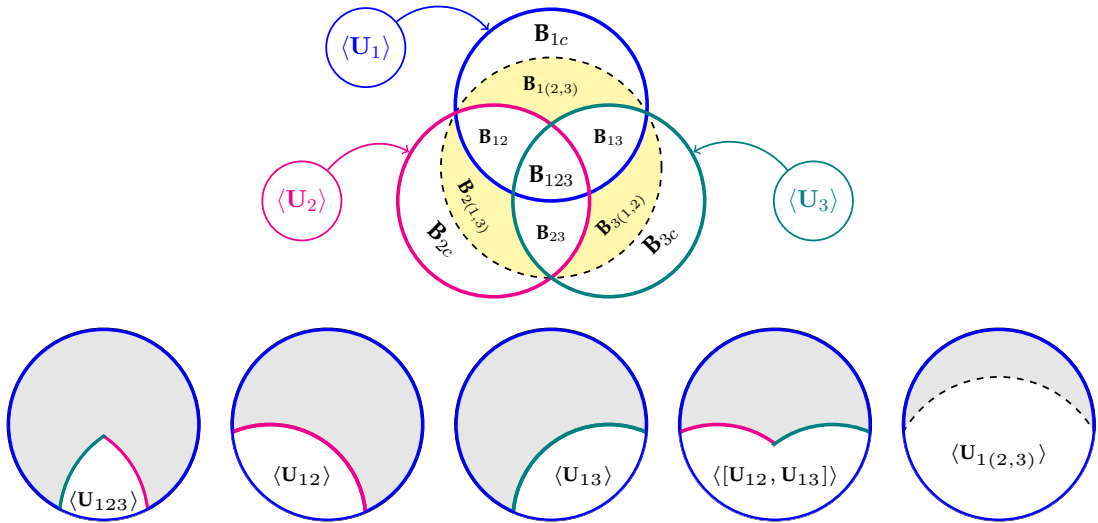

We conclude this section with the following observations.
\begin{enumerate}[wide, labelindent=1em ,labelwidth=!, labelsep*=1em, leftmargin =0em, style = sameline , label=\it Remark \it\arabic{section}.\arabic*]
\item The decomposition of $2$ linear subspaces in Lemma \ref{lem:2_decomp} resembles the decomposition of $2$ sets, e.g., the inclusion-exclusion principle and Venn's diagrams are reflected in the decompositions. However, the set-theoretic analogy is no longer true for $3$ linear subspaces, as in the decomposition the $3$ yellow spaces are not mutually independent. Appendix \ref{app:3space_discuss} provides more discussion regarding this property.
\item \label{rem:wang} Identifying suitable intersecting subspaces within vector spaces is also a recurrent theme in the degrees of freedom (DoF) studies of wireless networks, e.g., to simplify the design of interference alignment schemes in MIMO settings \cite{Jafar_FnT, Wang_Gou_Jafar_subspace}. In particular, the DoF study of a $3$ user wireless MIMO BC setting in the PhD thesis of Wang \cite[Ch. 3]{wang2017degrees} provides a subspace decomposition that very closely parallels Lemma \ref{lemma:decomposition}. The correspondence and the distinctions between the two are discussed in   Appendix  \ref{app:comparison}, as are the limitations that prevent the proof in \cite[Ch. 3]{wang2017degrees} from carrying over directly to our finite field setting. An independent proof of Lemma \ref{lemma:decomposition} for our setting is provided in Appendix \ref{proof:decomposition}. Notably the proof in Appendix \ref{proof:decomposition} only relies on arguments that hold both over finite fields as well as over the field of complex numbers, thereby unifying the two settings.

\item \label{rem:caveat0} Lemma \ref{lemma:decomposition} ignores  the details of how each ${\bf U}_i$ is composed of ${\bf V}_i$ and ${\bf V}_i'$. Depending on their own side-information and demand, each user will have a \emph{different} conditional view of these subspaces. This most essential aspect of the LCBC problem is not reflected in the decomposition. Thus, it is worthwhile to note that the decomposition is primarily a preparatory step, the main technical challenge from both achievability and converse perspectives remains focused on accounting for the distinct side-information and demand structures across users. See also \ref{rem:caveat}.
\end{enumerate}

\section{Results}\label{sec:result}
\subsection{A Closed Form Capacity Expression for the $3$ User LCBC}
 As our main result, the following theorem states the capacity of the $3$ user LCBC in closed form.
\begin{theorem} \label{thm:main}
For the $K=3$ user general LCBC, i.e., LCBC($\mathbb{F}_q,3,d,(m_k,m_k')_{k\in[3]},{({\bf V}_k,{\bf V}_k')_{k\in[3]}})$, the capacity  $C= 1/\max\{\Delta_1,\Delta_2\}$, equivalently, the optimal broadcast cost, $\Delta^*= \max\{\Delta_1,\Delta_2\}$ where,
\begin{align}
&	\Delta_n \triangleq \max_{(i,j,k)\in \{\mbox{\footnotesize permutations of } (1,2,3)\}} \{\Delta_n^{ijk}\}\label{eq:delta_n},~~~ n \in \{1,2\},\\
\intertext{and}
&\Delta_1^{ijk} \triangleq \rk({\bf V}_1\mid \bV_1')  + \rk({\bf V}_2\mid \bV_2')  +\rk({\bf V}_3\mid \bV_3')- \rk({\bf U}_{ij}\mid {\bf V}_j') - \rk({\bf U}_{k(i,j)}\mid {\bf V}_k'), \\
&\Delta_2^{ijk} \triangleq \rk({\bf V}_1\mid \bV_1') + \rk({\bf V}_2\mid \bV_2') + \rk({\bf V}_3\mid \bV_3') \notag\\
&\hspace{1cm}- \frac{1}{2} \Bigg(\min_{\ell\in [3]} \big(\rk({\bf U}_{123}\mid {\bf V}_\ell')\big) +\rk([{\bf U}_{ij},{\bf U}_{ik}]\mid {\bf V}_i') +\rk({\bf U}_{j(i,k)}\mid {\bf V}_j') + \rk({\bf U}_{k(i,j)}\mid {\bf V}_k') \Bigg),
\end{align}
\end{theorem} 
\noindent Recall the conditional-rank notation defined in Section \ref{sec:notation}, $\rk(X|Y) \triangleq \rk([X,Y])-\rk(Y)$ . The proof of Theorem \ref{thm:main} will be presented along with the proof of the upcoming Theorem \ref{thm:LP}, in Sections \ref{sec: Converse}, \ref{sec: Achievability},  and \ref{sec:linear_prog} according to the proof structure specified in Section \ref{sec:proofstructure}.

\begin{enumerate}[wide, labelindent=1em ,labelwidth=!, labelsep*=1em, leftmargin =0em, style = sameline , label=\it Remark \it\arabic{section}.\arabic{subsection}.\arabic*]
\item  \label{rem:follows}
The bound $\Delta^*\geq \Delta_1$  follows  from a  generalization of the converse bound of the $2$ user LCBC, and is similar to the genie-aided converse bound of coded caching (e.g., \cite[(71)-(75)]{yu2018characterizing}). However, unlike the $2$ user LCBC, this bound is not sufficient for  the $3$ user LCBC, which is why we also need the bound $\Delta^*\geq \Delta_2$. The  bound  $\Delta^*\geq \Delta_2$ encapsulates the new technical challenge in the $3$ user LCBC from the converse perspective (see Section \ref{sec: Converse}).
\item The capacity of the $3$ user LCBC can be expressed in various equivalent forms. The \emph{closed} form presented in Theorem \ref{thm:main} emerges naturally from the converse bounds. Indeed, the converse in Section \ref{sec: Converse} directly produces two bounds, one each for $\Delta_1^{ijk}, \Delta_2^{ijk}$. The achievability argument on the other hand, takes a different approach which involves auxiliary parameters (the $\lambda_\bullet$ parameters in Theorem \ref{thm:LP}) representing various design choices. Optimizing the design choices amounts to a linear program, the solution to which yields the same $\Delta^*$ as Theorem \ref{thm:main}. Even though the converse and achievability perspectives ultimately lead to the same $\Delta^*$, their different forms yield different insights. The achievability perspective in particular yields constructive insights into the tradeoffs involved in simultaneously satisfying all $3$ users' demands. This alternative (but equivalent) form of the capacity result is presented next.
\end{enumerate}

\subsection{An Alternative Expression for the Capacity of the $3$ User LCBC}
\begin{theorem}\label{thm:LP}
\begin{align}
\Delta^*=F^*
\end{align}
where $\Delta^*$ is the optimal broadcast cost for the $K=3$ user general LCBC and $F^*$ is the solution to the following linear program,
\begin{align}
F^*\triangleq \min_{\lambda_{123},\lambda_{12},\lambda_{13},\lambda_{23},\lambda\in\mathbb{R}_+}\rk(\bV_1|\bV_1')+\rk(\bV_2|\bV_2')+\rk(\bV_3|\bV_3') - 2\lambda_{123}-\lambda_{12}-\lambda_{13}-\lambda_{23}-\lambda,\label{eq:optimize}
\end{align}
such that\footnote{By definition the indices $ij$ and $ji$ are interchangeable.}
{\small
\begin{align}
\lambda_{123}&\leq \rk(\bU_{123}\mid \bV_i'), && \forall i \in \{1,2,3\},\label{eq:constraintfirst}\\
\lambda_{ij}+\lambda_{123}&\leq \min \left( \rk(\bU_{ij}\mid \bV_i'),  \rk(\bU_{ij}\mid \bV_j') \right),&& \forall (i,j) \in \{(1,2), (1,3), (2,3)\},\label{eq:constraintsecond}\\
\lambda_{ij}+\lambda_{ik}+\lambda_{123}&\leq \rk([\bU_{ij},\bU_{ik}]\mid \bV_i'),&& \forall (i,j,k) \in \{(1,2,3), (2,1,3), (3,1,2)\},\label{eq:constraintthird}\\
\lambda+\lambda_{ij}+\lambda_{ik}+\lambda_{123}&\leq \rk(\bU_{i(j,k)}\mid \bV_i'),&& \forall (i,j,k) \in \{(1,2,3), (2,1,3), (3,1,2)\}.\label{eq:constraintlast}
\end{align}
}
\end{theorem}

\begin{figure*}[htb]
\center
\scalebox{0.85}{
\begin{tikzpicture}
\node [draw,
    circle,
    minimum size =3cm, color = olive, thick] (U1) at (0.9,1.5){};
\node [draw,
    circle,
    minimum size =3cm, color = magenta, thick] (U2) at (0,0){};
\node [draw,
    circle,
    minimum size =3cm, color = teal, thick] (U3) at (1.8,0){};
\node [draw, dashed,
    circle,
    minimum size =3.45cm, fill = yellow!40, thick] (D) at (0.9,0.52){};
\node [draw,
    circle,
    minimum size =1cm, color = blue, thick] (UU1) at (-2,2.4){$\langle {\bf U}_1 \rangle$};
\node [draw,
    circle,
    minimum size =1cm, color = magenta, thick] (UU2) at (-3,0){$\langle {\bf U}_2 \rangle$};
\node [draw,
    circle,
    minimum size =1cm, color = teal, thick] (UU3) at (5,0){$\langle {\bf U}_3 \rangle$};
    
\begin{scope}
    \clip (0,0) circle(1.5cm);
    \clip (0.9,1.5) circle(1.5cm);
    \fill[white](0,0) circle(1.5cm);
\end{scope}

\begin{scope}
    \clip (0,0) circle(1.5cm);
    \clip (1.8,0) circle(1.5cm);
    \fill[white](0,0) circle(1.5cm);
\end{scope}

\begin{scope}
   \clip (0.9,1.5) circle(1.5cm);
    \clip (1.8,0) circle(1.5cm);
    \fill[white](0.9,1.5) circle(1.5cm);
\end{scope}

\node [draw,
    circle,
    minimum size =0.7cm, fill = red!10] () at (0.9,0.5){\tiny $\lambda_{123}$};
\node [draw,
    circle,
    minimum size =0.7cm, fill = red!10] () at (0,1){\tiny $\lambda_{12}$};
\node [draw,
    circle,
    minimum size =0.7cm, fill = red!10] () at (1.8,1){\tiny $\lambda_{13}$};
\node [draw,
    circle,
    minimum size =0.7cm, fill = red!10] () at (0.9,-0.5){\tiny $\lambda_{23}$};
\node [draw,
    circle,
    minimum size =0.7cm, fill = red!10] () at (0.9,1.8){\tiny $\lambda$};
\node [draw,
    circle,
    minimum size =0.7cm, fill = red!10] () at (-0.2,-0.2){\tiny $\lambda$};
\node [draw,
    circle,
    minimum size =0.7cm, fill = red!10] () at (2,-0.2){\tiny $\lambda$};

\node [draw,
    circle,
    minimum size =3cm, color = blue, ultra thick] (U1) at (0.9,1.5){};
\node [draw,
    circle,
    minimum size =3cm, color = magenta, ultra thick] (U2) at (0,0){};
\node [draw,
    circle,
    minimum size =3cm, color = teal, ultra thick] (U3) at (1.8,0){};
\draw [->, blue, thick] (UU1) to [out=40, in=130](U1);
\draw [->, teal,  thick] (UU3) to [out=130, in=30](U3);
\draw [->, magenta, thick] (UU2) to [out=50, in=150](U2);

\begin{scope}[shift={(-7,-5)}]
\node [align= center] at (0.9, -0.5) {\footnotesize $\begin{array}{c}\lambda_{123}\\ \leq \rk(\bU_{123}\mid \bV_1')\end{array}$};
\node at (4,2.5)[blue]{};         
\draw[clip] (0.9,1.5) circle (1.51cm);
\draw [blue, fill=black!10, ultra thick](0.9,1.5) circle (1.49cm);
\clip (1.8,0) circle (1.55cm);
\clip (0,0) circle (1.55cm);

\node [draw,
    circle,
    minimum size =3cm, color = magenta, ultra thick, fill=white] at (0,0){};
\node [draw,
    circle,
    minimum size =3cm, color = teal, ultra thick, fill=white] (U3) at (1.8,0){};  
\node [draw,
    circle,
    minimum size =3cm, color = magenta, ultra thick] at (0,0){};
\node [draw,
    circle,
    minimum size =3cm, color = teal, ultra thick] (U3) at (1.8,0){};      
\node [draw,
    circle,
    minimum size =0.7cm, fill = red!10] () at (0.9,0.5){\tiny $\lambda_{123}$};
\node [draw,
    circle,
    minimum size =0.7cm, fill = red!10] () at (0,1){\tiny $\lambda_{12}$};
\node [draw,
    circle,
    minimum size =0.7cm, fill = red!10] () at (1.8,1){\tiny $\lambda_{13}$};
\node [draw,
    circle,
    minimum size =0.7cm, fill = red!10] () at (0.9,1.8){\tiny $\lambda$};

\draw[blue, ultra thick] (0.9,1.5) circle (1.5cm); 
\end{scope}

\begin{scope}[shift={(-3.5,-5)}]
\node [align= center] at (0.9, -0.5) {\footnotesize $\begin{array}{c}\lambda_{12}+\lambda_{123}\\ \leq \rk(\bU_{12}\mid \bV_1')\end{array}$};
\node at (4,2.5)[blue]{};         
\draw[clip] (0.9,1.5) circle (1.51cm);
\draw [blue, fill=black!10, ultra thick](0.9,1.5) circle (1.49cm);
\clip (0,0) circle (1.55cm);

\node [draw,
    circle,
    minimum size =3cm, color = magenta, ultra thick, fill=white] at (0,0){};
\node [draw,
    circle,
    minimum size =3cm, color = magenta, ultra thick] at (0,0){};
\node [draw,
    circle,
    minimum size =0.7cm, fill = red!10] () at (0.9,0.5){\tiny $\lambda_{123}$};
\node [draw,
    circle,
    minimum size =0.7cm, fill = red!10] () at (0,1){\tiny $\lambda_{12}$};
\node [draw,
    circle,
    minimum size =0.7cm, fill = red!10] () at (1.8,1){\tiny $\lambda_{13}$};
\node [draw,
    circle,
    minimum size =0.7cm, fill = red!10] () at (0.9,1.8){\tiny $\lambda$};

\draw[blue, ultra thick] (0.9,1.5) circle (1.5cm); 
\end{scope}

\begin{scope}[shift={(0,-5)}]
\node [align= center] at (0.9, -0.5) {\footnotesize $\begin{array}{c}\lambda_{13}+\lambda_{123}\\ \leq \rk(\bU_{13}\mid \bV_1')\end{array}$};
\node at (4,2.5)[blue]{};         
\draw[clip] (0.9,1.5) circle (1.51cm);
\draw [blue, fill=black!10, ultra thick](0.9,1.5) circle (1.49cm);
\clip (1.8,0) circle (1.55cm);
\node [draw,
    circle,
    minimum size =3cm, color = teal, ultra thick, fill=white] (U3) at (1.8,0){};  
\node [draw,
    circle,
    minimum size =3cm, color = teal, ultra thick] (U3) at (1.8,0){};      
\node [draw,
    circle,
    minimum size =0.7cm, fill = red!10] () at (0.9,0.5){\tiny $\lambda_{123}$};
\node [draw,
    circle,
    minimum size =0.7cm, fill = red!10] () at (0,1){\tiny $\lambda_{12}$};
\node [draw,
    circle,
    minimum size =0.7cm, fill = red!10] () at (1.8,1){\tiny $\lambda_{13}$};
\node [draw,
    circle,
    minimum size =0.7cm, fill = red!10] () at (0.9,1.8){\tiny $\lambda$};

\draw[blue, ultra thick] (0.9,1.5) circle (1.5cm); 
\end{scope}

\begin{scope}[shift={(3.5,-5)}]
\node [align= center] at (0.9, -0.5) {\footnotesize $\begin{array}{c}\lambda_{12}+\lambda_{13}+\lambda_{123}\\ \leq \rk([\bU_{12}, \bU_{13}]\mid \bV_1')\end{array}$};
\node at (4,2.5)[blue]{};         
\draw[clip] (0.9,1.5) circle (1.51cm);
\draw [blue, fill=black!10, ultra thick](0.9,1.5) circle (1.49cm);
\clip (0.9,0.52) circle (1.725cm);

\node [draw,
    circle,
    minimum size =3cm, color = magenta, ultra thick, fill=white] at (0,0){};
\node [draw,
    circle,
    minimum size =3cm, color = teal, ultra thick, fill=white] (U3) at (1.8,0){};  
\node [draw,
    circle,
    minimum size =2.85cm, color = white, ultra thick, fill=white] at (0,0){};    
\node [draw,
    circle,
    minimum size =0.7cm, fill = red!10] () at (0.9,0.5){\tiny $\lambda_{123}$};
\node [draw,
    circle,
    minimum size =0.7cm, fill = red!10] () at (0,1){\tiny $\lambda_{12}$};
\node [draw,
    circle,
    minimum size =0.7cm, fill = red!10] () at (1.8,1){\tiny $\lambda_{13}$};

\draw[blue, ultra thick] (0.9,1.5) circle (1.5cm); 
\end{scope}

\begin{scope}[shift={(7,-5)}]
\node [align= center] at (0.9, -0.5) {\footnotesize $\begin{array}{c}\lambda+\lambda_{12}+\lambda_{13}+\lambda_{123}\\ \leq \rk(\bU_{1(2,3)}\mid \bV_1')\end{array}$};
\node at (4,2.5)[blue]{};         
\draw[clip] (0.9,1.5) circle (1.51cm);
\draw [blue, fill=black!10, ultra thick](0.9,1.5) circle (1.49cm);

\node [draw, dashed,
    circle,
    minimum size =3.45cm, fill = white, thick] (D) at (0.9,0.52){};
\node [draw,
    circle,
    minimum size =0.7cm, fill = red!10] () at (0.9,0.5){\tiny $\lambda_{123}$};
\node [draw,
    circle,
    minimum size =0.7cm, fill = red!10] () at (0,1){\tiny $\lambda_{12}$};
\node [draw,
    circle,
    minimum size =0.7cm, fill = red!10] () at (1.8,1){\tiny $\lambda_{13}$};
\node [draw,
    circle,
    minimum size =0.7cm, fill = red!10] () at (0.9,1.8){\tiny $\lambda$};

\draw[blue, ultra thick] (0.9,1.5) circle (1.5cm); 
\end{scope}

\end{tikzpicture}}
\caption{Intuitive understanding of the constraints \eqref{eq:constraintfirst}-\eqref{eq:constraintlast}. Note that the  blue circles in the bottom row of the figure do not show $\langle \bU_1\rangle$ per se, rather they show the space $\langle \bU_1\rangle$ conditioned on User $1$'s side-information, so that the corresponding sizes of all subspaces are represented with conditional ranks after conditioning on $\bV_1'$. Each user will generally have a different perspective due to different impact of conditioning on their respective side-informations, giving rise to different constraints \eqref{eq:constraintfirst}-\eqref{eq:constraintlast}.}
\label{fig:venn2}
\end{figure*}
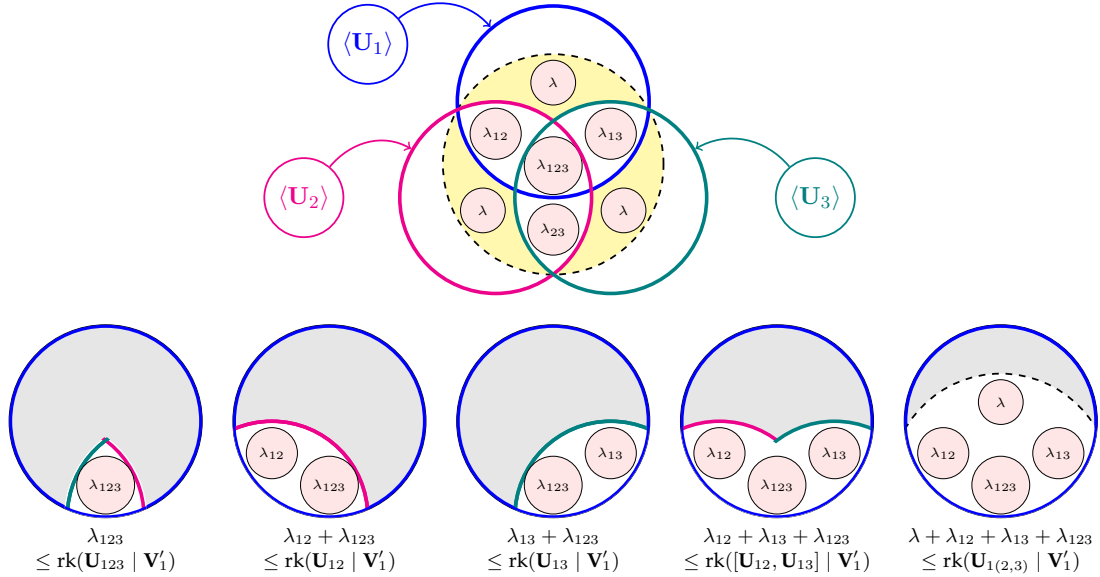

\begin{enumerate}[wide, labelindent=1em ,labelwidth=!, labelsep*=1em, leftmargin =0em, style = sameline , label=\it Remark \it\arabic{section}.\arabic{subsection}.\arabic*]
\item For the sake of high level intuition, Figure \ref{fig:venn2} conveys a somewhat oversimplified (the caveat is noted in  \ref{rem:caveat}) understanding of the conditions \eqref{eq:constraintfirst}-\eqref{eq:constraintlast} in Theorem \ref{thm:LP}. The $\lambda_\bullet$ parameters represent the size (dimension) of signals in various subspaces to be broadcast by the coding scheme. Depending upon the region they fall in, the subspaces have different communication efficiencies. For instance, note that $\lambda_{123}$ falls in $\bU_{123}$, and carries information that is simultaneously useful for all $3$ users. Thus, $\lambda_{123}$ transmitted dimensions satisfy a total of $3\lambda_{123}$ dimensions of demand ($\lambda_{123}$ per user). Borrowing the classical metaphor, we refer to the efficiency of such transmissions as $3$ birds, $1$ stone. Transmissions corresponding to $\lambda_{ij}$ fall in subspaces $\bU_{ij}$ and are simultaneously useful for Users $i$ and $j$, so the efficiency of such transmissions is similarly referred to as $2$ birds, $1$ stone. In other words, $\lambda_{ij}$ transmitted dimensions satisfy $2\lambda_{ij}$ dimensions of demand. Transmissions corresponding to $\lambda$ fall in the three subspaces highlighted in yellow in Figure \ref{fig:venn1} where we previously noted that any two subspaces are disjoint but contain the third. What this means is that the coding scheme needs to send any $2$ of the $3$ subspaces marked with $\lambda$, and the third can be automatically inferred from them. Thus, a transmission of $2\lambda$ dimensions, satisfies a total of $3\lambda$ dimensions of demand ($\lambda$ per user), yielding an efficiency of $3$ birds, $2$ stones.

\item In light of the previous remark, now consider the objective to be minimized in \eqref{eq:optimize},
$\Delta^*=\rk(\bV_1|\bV_1')+\rk(\bV_2|\bV_2')+\rk(\bV_3|\bV_3') - 2\lambda_{123}-\lambda_{12}-\lambda_{13}-\lambda_{23}-\lambda$. We recognize the sum of the first three terms as the broadcast cost if the users were to be served separately and no gain in efficiency was possible by jointly satisfying multiple demands. Let this be our baseline. Now note that because $3\lambda_{123}$ dimensions of demand were satisfied with $\lambda_{123}$ dimensions of broadcast, the cost-saving incurred relative to the baseline is $2\lambda_{123}$, which explains the fourth term that appears as a negative term in the objective. The next three  negative terms are similarly justified because each $\lambda_{ij}$ dimensions of transmission satisfies $2\lambda_{ij}$ dimensions of demand, thus saving $\lambda_{ij}$ relative to the baseline. Finally, for the $\lambda$ term, we recall that a total transmission cost of $2\lambda$ dimensions is able to satisfy $3\lambda$ dimensions of demand, thus saving another $\lambda$ in broadcast cost, which explains the last negative term in the objective function.

\item \label{rem:caveat} As a caveat,  note that the intuitive explanation above ignores a critical aspect of the problem that remains challenging --- namely, each user's view of useful dimensions depends on their own side-information, and is in general different from other users. This is indicated in Figure \ref{fig:venn2} by noting that the relevant signal spaces for User $1$ are not simply the $\bU_\bullet$ spaces that appear in the decomposition at the top of Figure \ref{fig:venn1} and Figure \ref{fig:venn2}. Rather, each user's view of useful subspaces is conditional on his side-information. For example, the same signal space $\bU_{123}$ when seen by the Users $1,2,3$, contains $\rk(\bU_{123}\mid \bV_1'), \rk(\bU_{123}\mid \bV_2'), \rk(\bU_{123}\mid \bV_3')$ useful dimensions, respectively. Thus, the total number of  dimensions useful to all three users, i.e., the size of $\lambda_{123}$ is limited by the bound in \eqref{eq:constraintfirst}. Even with the size of $\lambda_{123}$ constrained in this manner, finding the broadcast dimension is not trivial because each user may find a different $\lambda_{123}$ portion of $\bU_{123}$ useful to them. Similar challenges arise in identifying $\lambda_{ij}$ dimensions that are useful to Users $i$ and $j$, when each user's perspective is different, conditioned on their own side-information. Even greater care has to be taken in identifying the $\lambda$ sections of the broadcast signal, to ensure that $2$ transmissions span the third, while facing the challenge that the projections of $\lambda$ into each user's perspective are distinguished by their different side-informations.

\item \label{rem:polymat} Since linear optimizations over polymatroidal constraints allow  greedy solutions \cite{Edmonds} that can simplify dimensional analysis (see e.g., the DoF study in \cite[Chapter 5]{romerothesis}), it is worth noting that the constraints \eqref{eq:constraintfirst}-\eqref{eq:constraintlast} do not specify a polymatroidal structure. To verify this with a toy example, suppose ${\bf V}_1={\bf V}_2'={\bf V}_3'=[1,1]^T$ and ${\bf V}_1'={\bf V}_2={\bf V}_3=[0,0]^T$. Then we have the constraints, $\lambda_{123}\leq 0$, $\lambda_{12}+\lambda_{123}\leq 0$, $\lambda_{13}+\lambda_{123}\leq  0$ and $\lambda_{12}+\lambda_{13}+\lambda_{123}\leq 1$, which violate the polymatroidal structure.
\end{enumerate}

\subsection{Structure of Proofs}\label{sec:proofstructure}
Theorem \ref{thm:main} and Theorem \ref{thm:LP} are  equivalent alternative forms of the same capacity result. We organize the proofs of these two theorems as follows. In Section \ref{sec: Converse} we prove the converse (lower) bound for the optimal broadcast cost, i.e., $\Delta^*\geq \max\{\Delta_1, \Delta_2\}$.  Then in Section \ref{sec: Achievability} we prove the achievability (upper) bound $\Delta^*\leq F^*$. Finally in Section \ref{sec:linear_prog} we prove that $F^*\leq \max\{\Delta_1, \Delta_2\}$. The three proofs together imply that $\Delta^*=F^*=\max\{\Delta_1,\Delta_2\}$, thus proving both Theorem \ref{thm:main} and Theorem \ref{thm:LP}.

\section{Toy Examples} \label{sec:examples}
In this section, we present simple toy examples that illustrate some of the ideas discussed previously,  such as subspace decompositions and linear-programming  tradeoffs between schemes with different communication efficiency (birds vs stones), some ideas that will be important later on in the construction of the general coding scheme, such as field extensions, vector coding, and mixing of dimensions, and some new insights, such as the insufficiency of entropic structure, and the need for functional submodularity. For these examples we use specialized notation for simplicity: $(({\bf W}_i'\rightarrow{\bf W}_i), i=1,2,3)$ to specify the setting, $A,B,C,D,E$ instead of $x_1,x_2,x_3,x_4,x_5$, and $A_\ell$ instead of $A(\ell)$.

\begin{example}[$3$ birds, $1$ stone]
	\label{ex:3x_efficiency} Consider $d=3$ dimensional data ${\bf X}^T=(A,B,C)$ over $\mathbb{F}_3$, and $((A\rightarrow B+C), (B\rightarrow A+C),(C\rightarrow A+B))$. In other words, User $1$ has $A$ and wants $B+C$, User $2$ has $B$ and wants $A+C$, and User $3$ has $C$ and wants $A+B$. A signal space decomposition as in Figure \ref{fig:venn1}  yields for this example,
\begin{align}
\begin{array}{|c|c|c|c|c|c|c|c|c|c|}\hline
\bB_{123}& \bB_{12}& \bB_{13}& \bB_{23}&\bB_{1(2,3)}&\bB_{2(1,3)}&\bB_{3(1,2)}&\bB_{1c}&\bB_{2c}&\bB_{3c}\\\hline
A+B+C& - & - & - & A & B & A+B & - & - & -\\\hline
\end{array}	
\end{align}
Note that for simplicity in these examples we indicate $\bB_{123}$ as $A+B+C$ instead of the formal representation as the vector $[1,1,1]^T$ in the $3$ dimensional data universe. The optimal broadcast cost is $\Delta^*=1$, achieved with $L=1$, $N=1$, $\lambda_{12}=\lambda_{23}=\lambda_{13}=\lambda=0$, $\lambda_{123}=1$, by broadcasting ${\bf S}=(A+B+C)$.
\end{example}

\begin{example}[$2$ birds, $1$ stone, vector coding, insufficiency of entropic structure]\label{ex:2x_efficiency} Consider $d=3$ dimensional data ${\bf X}^T=(A,B,C)$ over $\mathbb{F}_3$, and $((A\rightarrow B+C), (B\rightarrow A+C),(C\rightarrow A+2B))$. 
A signal space decomposition yields,
\begin{align*}
\begin{array}{|c|c|c|c|c|c|c|c|c|c|}\hline
\bB_{123}& \bB_{12}& \bB_{13}& \bB_{23}&\bB_{1(2,3)}&\bB_{2(1,3)}&\bB_{3(1,2)}&\bB_{1c}&\bB_{2c}&\bB_{3c}\\\hline
-&{\small A+B+C} & {\small A+2B+2C} & {\small A+2B+C} & - & - & - & - & - & -\\\hline
\end{array}	
\end{align*}
The optimal broadcast cost is $\Delta^*=1.5$, achieved with $L=2$, $N=3$, $\lambda_{123}=\lambda=0, \lambda_{12}=\lambda_{13}=\lambda_{23}=0.5$, by broadcasting ${\bf S}=(A_1+B_1+C_1, A_2+2B_2+2C_2, (A_1+2B_1+C_1)+(A_2+2B_2+C_2))$. 
\end{example}
Evidently, $1.5$ dimensions of broadcast satisfy a total of $3$ dimensions of demand, as expected from a $2$ birds, $1$ stone setting. Also note that this example requires vector coding, i.e., we need $L>1$. Most importantly, however, this example illustrates that unlike the $2$ user LCBC, the entropic formulation of \cite{Sun_Jafar_CBC} is not enough for the $3$ user LCBC. The following remark elaborates upon this  observation.

\begin{enumerate}[wide, labelindent=1em ,labelwidth=!, labelsep*=1em, leftmargin =0em, style = sameline , label=\it Remark \it\arabic{section}.\arabic*]
\item \label{rem:funcsub} Reference \cite{Sun_Jafar_CBC} considers an entropic formulation of the LCBC that is summarized as follows. The data ${\bf X}$ is assumed to be i.i.d. uniform, $\mathcal{W}^*\triangleq \{{\bf W}_1, {\bf W}_1', \cdots, {\bf W}_K, {\bf W}_K'\}$ denotes the set of all $2K$ demand and side-information random variables,  the entropies $H(\mathcal{W})$ are specified for all $2^{2K}-1$ non-empty subsets of random variables $\mathcal{W}\subset\mathcal{W}^*$, the encoding constraint is represented as $H({\bf S}\mid \mathcal{W}^*)=0$, and the decoding constraints are represented as $H({\bf W}_k\mid {\bf S}, {\bf W}_k')=0$ for all $k\in[K]$. Subject to these entropy specifications, as well as standard (Shannon and non-Shannon) information inequalities, the goal is to minimize the entropy $H({\bf S})$. As discussed in Section \ref{sec:entropy}, such a formulation produces a  lower bound on the download cost, as $N\geq H({\bf S})$, which in turn yields a lower bound on $\Delta^*$. For the $K=2$ user LCBC, this bound turns out to be tight. Remarkably, however, the same approach does not work for the $K=3$ user LCBC, as  we argue based on Example \ref{ex:3x_efficiency} and Example \ref{ex:2x_efficiency}. Although a bit tedious, it is not difficult to verify that all $2^6-1=63$ entropies $H(\mathcal{W})$ match for Example \ref{ex:3x_efficiency} and Example \ref{ex:2x_efficiency}. For example, consider $\mathcal{W}=\{{\bf W}_1', {\bf W}_3\}$. Note that $H(\mathcal{W})=H({\bf W}_1', {\bf W}_3)=H(A, A+B)=H(A,B)=2L$ in Example \ref{ex:3x_efficiency}, and $H(\mathcal{W})=H({\bf W}_1', {\bf W}_3)=H(A, A+2B)=H(A,B)=2L$ in Example \ref{ex:2x_efficiency}, so both examples have the same entropy for this $\mathcal{W}$. One can similarly compute $H(\mathcal{W})$ for all $63$ non-empty subsets $\mathcal{W}\subset\mathcal{W}^*$ for both Example \ref{ex:3x_efficiency} and Example \ref{ex:2x_efficiency} and verify that in each case both examples produce matching entropies. Therefore, since all the entropic constraints for both examples are identical, and all Shannon and non-Shannon information inequalities apply to both examples, the two examples can only produce the same entropic lower bound on $H({\bf S})$. However, we know that the  two examples have different capacities. Example \ref{ex:3x_efficiency} has  $\Delta^*=1, C=1$ while Example \ref{ex:2x_efficiency} has  $\Delta^*=1.5, C=1/1.5=2/3$. Since Example \ref{ex:2x_efficiency} requires a strictly stronger  bound (impossibility result) than Example \ref{ex:3x_efficiency} for a tight converse, it follows that the entropic formulation cannot yield a tight converse for  Example \ref{ex:2x_efficiency}.  Indeed, the key to the converse bound $\Delta^*\geq 1.5$ for Example \ref{ex:2x_efficiency} is the \emph{functional} submodularity property \cite{Tao_FS, Kontoyiannis_Madiman} that takes into account the functional forms of the users' side-information and demands. A converse for Example \ref{ex:2x_efficiency} is explicitly provided in Section \ref{sec:conex2}.
\end{enumerate}

\begin{example}[$3$ birds, $2$ stones, the user's perspective]
	\label{ex:ic} Consider $d=3$ dimensional data ${\bf X}^T=(A,B,C)$ over $\mathbb{F}_2$, and $((A\rightarrow B), (B\rightarrow C),(C\rightarrow A))$. A signal space decomposition as in Figure \ref{fig:venn1}  yields,
\begin{align*}
\begin{array}{|c|c|c|c|c|c|c|c|c|c|}\hline
\bB_{123}& \bB_{12}& \bB_{13}& \bB_{23}&\bB_{1(2,3)}&\bB_{2(1,3)}&\bB_{3(1,2)}&\bB_{1c}&\bB_{2c}&\bB_{3c}\\\hline
-& B & A & C & - & - & - & - & - & -\\\hline
\end{array}	
\end{align*}
	This coincides with an index coding problem,  the optimal broadcast cost is $\Delta^*=2$, achieved with $L=1$, $N=2$, $\lambda_{123}=\lambda_{12}=\lambda_{13}=\lambda_{23}=0,\lambda=1$, by broadcasting ${\bf S}=(A+B, B+C)$.
\end{example}
This example also highlights the importance of the users' individual perspectives conditioned on their side-information. Without accounting for side-information, the signal space decomposition of Figure \ref{fig:venn1} suggests that all the signals reside in $\bU_{12}, \bU_{13}, \bU_{23}$, which might suggest $2$ birds, $1$ stone schemes with $\lambda_{12}=\lambda_{13}=\lambda_{23}=0.5$ and a download cost of $\Delta^*=1.5$. However, this is not achievable, as we note the optimal download cost is $\Delta^*=2$. To see this, consider individual users' perspectives. For example, User $1$ requires $\lambda_{13}+\lambda_{123}\leq \rk({\bf U}_{13}\mid \bV_1')$. Now since both $\bU_{13}$ and $\bV_1'$ correspond to the data dimension $A$, this conditional rank is $0$. In other words, even though the subspace $\bU_{13}$ has one dimension that may suggest the opportunity to simultaneously satisfy users $1$ and $3$, this dimension happens to be already available to User $1$. Thus, upon taking into account User $1$'s side-information, there is no such opportunity. We end up with $\lambda_{123}=\lambda_{12}=\lambda_{13}=\lambda_{23}=0$, and $\lambda=1$. Out of the $3$ dimensions, say $A+B, B+C, C+A$, any two yield the third by summation (over $\mathbb{F}_2$), and it suffices to send any $2$ to satisfy all $3$ users. Notice the need to mix up the dimensions, appealing to mixed dimensions (similar to random coding arguments) will be a key idea to develop the general coding scheme.

As noted, Example \ref{ex:2x_efficiency} used vector coding $(L>1)$ to achieve the optimal download cost $\Delta^*=1.5$. Vector coding may be strictly necessary even in cases where the optimal download cost $\Delta^*$ is an integer value, as illustrated by the next example. The necessity of vector coding for the $3$ user LCBC is remarkable because scalar coding was found to be sufficient for the 2 user LCBC in \cite{Sun_Jafar_CBC}).

\begin{example}[Field size extension] \label{ex:field_ext} Consider $d=2$ dimensional data ${\bf X}^T=(A,B)$ over $\mathbb{F}_2$, and $((A\rightarrow B), (B\rightarrow A+B),(A+B\rightarrow A))$. 
A signal space decomposition as in Figure \ref{fig:venn1}  yields,
\begin{align*}
\begin{array}{|c|c|c|c|c|c|c|c|c|c|}\hline
\bB_{123}& \bB_{12}& \bB_{13}& \bB_{23}&\bB_{1(2,3)}&\bB_{2(1,3)}&\bB_{3(1,2)}&\bB_{1c}&\bB_{2c}&\bB_{3c}\\\hline
A,B& - & - & - & - & - & - & - & - & -\\\hline
\end{array}	
\end{align*}
We have $\Delta^*=1$, achieved with $L=2$, $N=2$, $\lambda_{123}=1, \lambda_{12}=\lambda_{13}=\lambda_{23}=\lambda=0$, ${\bf S}=(A_1+A_2+B_2,A_1+B_1+B_2)$.
\end{example}
Appendix \ref{app:field_extension} shows that $\Delta^*=1$ is not achievable with scalar coding, i.e., neither scalar linear nor scalar non-linear  coding scheme can achieve $\Delta=1$ for $L=1$ computation for this example. However, $\Delta^* = 1$ and can be achieved for $L=2$ computations with $N=2$. In this case, because $\lambda_{123}=1$,  we would like to broadcast one dimension. In the scalar code setting $L=1$, this one dimension can be found for each pair  of users but it cannot be the same for the three users simultaneously. To see this, note that $A+B$ helps User $1$ and User $2$ but not User $3$; $B$ helps User $1$ and User $3$ but not User $2$; $A$ helps User $2$ and User $3$ but not User $1$. Aside from the time-sharing type vector coding solution shown for  Example \ref{ex:2x_efficiency}, another approach is to consider $L>1$ (which implies a vector code) and use a scalar code in a larger extended field ${\Bbb F}_{2^z}$ (in general ${\Bbb F}_{q^z}$). For this example, with $L=2$, we can use a scalar code over $\mathbb{F}_4=\mathbb{F}_2[x]/(x^2+x+1)$,  which results in $N=2$ in ${\Bbb F}_2$. Representing $A=A_1+A_2x\in\mathbb{F}_4$, $B=B_1+B_2x\in\mathbb{F}_4$, the transmitted symbol is simply $(1+x)A+xB \mod (x^2+x+1)= (A_1+A_2+B_2)+x(A_2+B_1+B_2)$ which corresponds to the transmitted symbol ${\bf S}=(A_1+A_2+B_2,A_1+B_1+B_2)$. Additional discussion can be found in Appendix \ref{app:field_extension} as well. Indeed, field extensions are a key element of the general coding scheme.
 
\begin{example}[Inseparability] \label{ex:trade_off}
	Consider $d=5$ dimensional data ${\bf X}^T=(A,B,C,D,E)$ over $\mathbb{F}_3$, and $((A\rightarrow [B+C,D]), (B\rightarrow [A+C,E]),([C,D+E]\rightarrow A+2B))$. A signal space decomposition as in Figure \ref{fig:venn1}  yields,
\begin{align*}
\begin{array}{|c|c|c|c|c|c|c|c|c|c|}\hline
\bB_{123}& \bB_{12}& \bB_{13}& \bB_{23}&\bB_{1(2,3)}&\bB_{2(1,3)}&\bB_{3(1,2)}&\bB_{1c}&\bB_{2c}&\bB_{3c}\\\hline
-& A+B+C & A+2B+2C & A+2B+C & D & E & D+E & - & - & -\\\hline
\end{array}	
\end{align*}
	We have $\Delta^*=3$, achieved with $L=1$, $N=3$, $\lambda_{123}=\lambda_{13}=\lambda_{23}=0, \lambda_{12}=\lambda=1$, by broadcasting ${\bf S}=(A+B+C, A+D, 2B+E)$.
\end{example}
Note that this problem combines Example \ref{ex:2x_efficiency} for data $(A,B,C)$ and another LCBC instance with data $(D,E)$ where User $1$ wants $D$, User 2 wants $E$ and User 3 knows $D+E$. Separately, these problems have download costs of $1.5$ and $2$, respectively. Since the two problems deal with independent data, one might expect the solution to be separable, however a separate solution would have a total broadcast cost of $1.5+2=3.5$. The optimal $\Delta^*=3$, which is better than $3.5$, thus showing that even though an LCBC problem may be a composition of instances with separate datasets, in general a separate solution would be suboptimal. This observation also underscores why the tradeoffs in LCBC, that we see represented in the linear program, are non-trivial.

\section{Proof of Converse: $\Delta^* \geq \max\{\Delta_1,\Delta_2\}$} \label{sec: Converse}
The converse is comprised of the two bounds, $\Delta^*\geq \Delta_1$, and $\Delta^*\geq \Delta_2$. The first bound, $\Delta^*\geq \Delta_1$ is a straightforward generalization of the corresponding bound for the $2$ user LCBC found in \cite{Sun_Jafar_CBC} to the 3 user setting. The second bound, $\Delta^*\geq \Delta_2$ is novel, and requires functional submodularity. For the sake of completeness in this section we present the proof of both bounds. Let us begin by recalling the functional submodularity property.
\begin{lemma}[Functional submodularity of Shannon entropy (Lemma A.2 of \cite{Tao_FS})]\label{lemma:tao} If $X_0,X_1,X_2,X_{12}$  are random variables such that $X_1$ and $X_2$ each determine $X_0$ and $(X_1,X_2)$ determine $X_{12}$, then:
\begin{align}
H(X_1)+H(X_2) \geq H(X_{12})+H(X_0)\label{eq:tao}
\end{align}
\end{lemma}
Note that  `$A$ determines $B$' as used in Lemma \ref{lemma:tao} is equivalent to the statement that $H(B\mid A)=0$, i.e., $B$ is a function of $A$. Thus, the lemma assumes that $H(X_0\mid X_1)=H(X_0\mid X_2)=H(X_{12}\mid X_1,X_2)=0$.

As an immediate corollary, let us note the following form in which we will apply the functional submodularity.
\begin{corollary} \label{cor:submod}
For arbitrary matrices ${\bf M}_1\in\mathbb{F}_q^{d\times \mu_1}, {\bf M}_2\in\mathbb{F}_q^{d\times \mu_2}$, any random matrix ${\bf X}\in\mathbb{F}_q^{d\times L}$, and any random variable ${\bf Z}$,
\begin{align}
H({\bf Z}, {\bf X}^T{\bf M}_1)+H({\bf Z}, {\bf X}^T{\bf M}_2)&\geq H({\bf Z}, {\bf X}^T[{\bf M}_{1}\cap {\bf M}_2])+H({\bf Z}, {\bf X}^T[{\bf M}_1, {\bf M}_2])\label{eq:cortao}
\end{align}
\end{corollary}
\proof The corollary follows from Lemma \ref{lemma:tao} by setting $X_1=({\bf Z},{\bf X}^T{\bf M}_1)$, $X_2=({\bf Z}, {\bf X}^T{\bf M}_2)$, and noting that $X_0=({\bf Z}, {\bf X}^T[{\bf M}_1\cap{\bf M}_2])$ can be obtained as a function of both $X_1$ and $X_2$ individually, while $X_{12}=({\bf Z}, {\bf X}^T[{\bf M}_1, {\bf M}_2])$ is a function of $(X_1, X_2)$.
$\hfill\square$

\subsection{Proof of the bound: $\Delta^*\geq \Delta_1$}
As noted, the proof of this bound is straightforward. It follows along the same lines as the proof for the $2$ user LCBC in \cite{Sun_Jafar_CBC}, also similar to  the genie-aided bound in coded caching (e.g., \cite[(71)-(75)]{yu2018characterizing}) and is provided here for the sake of completeness. In particular, it does not require functional submodularity. As explained in Section \ref{sec:entropy}, recall that  the converse bound is based on a thought-experiment that supposes  that the data ${\bf X}$ is i.i.d. uniform, which leads to a lower bound $N\geq H({\bf S})$. 
Let ${\bf W}_k^* \triangleq ({\bf W}_k, {\bf W}_k'), \forall k\in [3]$. 
The bound follows essentially by iteratively using the argument 
\begin{align}
	H({\bf S} \mid {\bf W}_k',  {\bf W}^*_{[k-1]})  \geq \underbrace{H({\bf W}_k\mid {\bf W}_k', {\bf W}^*_{[k-1]})}_{\footnotesize \mbox{The genie-aided bound for the $k^{th}$ user}} + ~~~~H({\bf S} \mid {\bf W}_{k+1}', {\bf W}_{[k]}^*), \label{eq:genie_aided}
\end{align}
since $H({\bf S} \mid {\bf W}_k', {\bf W}_{[k-1]}^*) = H({\bf S}, {\bf W}_k \mid {\bf W}_k', {\bf W}_{[k-1]}^*) \geq H({\bf W}_k \mid {\bf W}_k', {\bf W}_{[k-1]}^*) + H({\bf S} \mid {\bf W}_{k+1}', {\bf W}_{[k]}^*)$, where the first step uses the decoder definition \eqref{eq:dec} and the second step applies the chain rule of entropy and the fact that conditioning reduces entropy.
It then follows that for any coding scheme $\big(L,N,\Phi,(\Psi_{k})_{k\in[3]}\big)\in\mathfrak{C}$,

\begin{align}
	&N \geq H({\bf S}) \geq H({\bf S} \mid {\bf W}_1') \\
    & \geq H({\bf W}_1 \mid {\bf W}_1') + H({\bf S} \mid {\bf W}_2', {\bf W}_1^*) \label{eq:use_genie_aided1}  \\
    & \geq H({\bf W}_1 \mid {\bf W}_1') + H({\bf W}_2 \mid {\bf W}_2', {\bf W}_1^*) + H({\bf S} \mid {\bf W}_3', {\bf W}_1^*, {\bf W}_2^*) \label{eq:use_genie_aided2} \\
    & \geq H({\bf W}_1 \mid {\bf W}_1')+ H({\bf W}_2 \mid {\bf W}_2', {\bf W}_1^*) + H({\bf W}_3 \mid {\bf W}_3', {\bf W}_1^*, {\bf W}_2^*)\label{eq:use_genie_aided3} \\
    & = L \cdot \left( \rk({\bf V}_1 \mid {\bf V}_1') + \rk({\bf V}_2 \mid [{\bf U}_1, {\bf V}_2']) +  \rk({\bf V}_3 \mid [{\bf U}_1, {\bf U}_2, {\bf V}_3']) \right)  \label{eq:entropytorank}\\
    & = L \cdot \left( \rk(\bV_1 \mid \bV_1')  +\rk(\bV_2 \mid \bV_2')-\rk(\bU_{12} \mid \bV_2')+ \rk(\bV_3 \mid \bV_3')-\rk(\bU_{3(1,2)} \mid \bV_3') \right) \label{eq:converting_ranks} \\
    & \implies \Delta = N/L \geq \Delta_1^{123}
\end{align} 
Steps \eqref{eq:use_genie_aided1} -- \eqref{eq:use_genie_aided3} follow from \eqref{eq:genie_aided}. Step \eqref{eq:entropytorank} uses the fact that for i.i.d. uniform data ${\bf X}^T\in\mathbb{F}_q^{L\times d}$ and an arbitrary matrix ${\bf M}\in\mathbb{F}_q^{d\times \mu}$, we have $H({\bf X}^T{\bf M})=L\cdot\rk({\bf M})$ in $q$-ary units, and applies the conditional-rank notation as defined in Section \ref{sec:notation}. Step \eqref{eq:converting_ranks} follows from the observation that $\rk({\bf V}_k \mid [{\bf Z}, {\bf V}_k']) = \rk([{\bf U}_k, {\bf Z}]) - \rk([{\bf V}_k',{\bf Z}]) = \rk({\bf U}_k)- \rk({\bf V}_k')  - \Big( \rk({\bf U}_k \cap {\bf Z}) - \rk({\bf V}_k' \cap {\bf Z}) \Big)   = \rk({\bf V}_k \mid {\bf V}_k')-  \rk({\bf U}_k \cap{\bf Z} \mid {\bf V}_k')$. Similarly, $\Delta \geq \Delta_1^{ijk}$ for all $(i,j,k)$ that are permutations of $(1,2,3)$. Since this holds for every coding scheme $\big(L,N,\Phi,(\Psi_{k})_{k\in[3]}\big)\in\mathfrak{C}$, it follows that $\Delta^*\geq \Delta_1$. $\hfill \qed$

\subsection{Proof of the bound: $\Delta^*\geq \Delta_2$}
The main idea of this proof is to successfully identify and introduce the entropies of certain (linear) \emph{functions} of users' demands and side-information that are critical in determining the capacity, with the application of Lemma \ref{lemma:tao}. To build intuition, let us start with the converse proof for a toy example, specifically Example 2 of Section \ref{sec:examples}.

\subsubsection{Converse Proof for a Toy Example}\label{sec:conex2}
Consider any coding scheme $\big(L,N,\Phi,(\Psi_{k})_{k\in[3]}\big)\in\mathfrak{C}$ for Example 2 of Section \ref{sec:examples}. Recall that User $1$ has $A_\ell$ and wants $B_\ell+C_\ell$; User $2$ has $B_\ell$ and wants $A_\ell+C_\ell$; User $3$ has $C_\ell$ and wants $A_\ell+2B_\ell$ for all $\ell \in [L]$. We want to prove the converse bound $\Delta^*\geq 1.5$. Let us denote $A$ as $(A_1,\cdots, A_L)\in {\Bbb F}_q^{1\times L}$, $B$ as $(B_1,\cdots, B_L)\in {\Bbb F}_q^{1\times L}$ and $C$ as $(C_1,\cdots, C_L)\in {\Bbb F}_q^{1\times L}$. As mentioned in Section \ref{sec:entropy} let us start the converse proof with the thought-experiment  that  $A,B,C$ are i.i.d. uniform in ${\Bbb F}_q$, which allows the following entropic arguments.
\begin{align}
	& 2H({\bf S}) + H(A)+ H(B) - I({\bf S};A) - I({\bf S};B) \notag \\
	& = H({\bf S},A) + H({\bf S},B) \label{eq:ex_converse_mi} \\
	& \geq H({\bf S},A, B+C) + H({\bf S},B, A+C) \label{eq:ex_converse_decoding} \\
	& \geq H({\bf S},A+B+C) + H({\bf S}, A,B,C)  \label{eq:ex_converse_fs} \\
	& \geq H({\bf S}) + H(A+B+C \mid {\bf S})  +  H(A,B,C).
\end{align}
Similarly,
\begin{align}
	&  2H({\bf S}) + H(A)+ H(C) - I({\bf S};A) - I({\bf S};C) \notag \\
	&  = H({\bf S},A) + H({\bf S},C) \label{eq:ex_converse_mi2}  \\
	& \geq H({\bf S},A, B+C) + H({\bf S},C, A+2B) \label{eq:ex_converse_decoding2} \\
	& \geq H({\bf S},A+2B+2C) + H({\bf S}, A,B,C)  \label{eq:ex_converse_fs2} \\
	& \geq H({\bf S}) + H(A+2B+2C\mid {\bf S}) + H(A,B,C).
\end{align}
Steps \eqref{eq:ex_converse_mi} and \eqref{eq:ex_converse_mi2} use the definition of mutual information $I(A;B) = H(A)+H(B)-H(A,B)$. Steps \eqref{eq:ex_converse_decoding} and \eqref{eq:ex_converse_decoding2} use the decoder definition \eqref{eq:dec}. Step \eqref{eq:ex_converse_fs} uses functional submodularity (Lemma \ref{lemma:tao}) by recognizing that $(A, B+C)$ and $(B,A+C)$ each determine $A+B+C$, and $(A,B+C,B,A+C)$ determines $(A,B,C)$. Step \eqref{eq:ex_converse_fs2} uses functional submodularity by recognizing that $(A, B+C)$ and $(C,A+2B)$ each determine $A+2B+2C$, and $(A,B+C,C,A+2B)$ determines $(A,B,C)$.

Adding the above two inequalities, we have
\begin{align}
	&4H({\bf S}) + 2H(A) + H(B) + H(C)- 2I({\bf S};A) - I({\bf S};B) - I({\bf S};C) \notag \\
	& \geq 2H({\bf S}) + H(A+B+C \mid {\bf S}) + H(A+2B+2C\mid {\bf S}) +  2H(A,B,C).
\end{align}
It follows that
\begin{align}
	2H({\bf S}) &\geq H(A+B+C \mid {\bf S}) + H(A+2B+2C\mid {\bf S})+I({\bf S};A) + 2L \label{eq:ex_converse_evaluate_entropy} \\
	&\geq H(A+B+C,A+2B+2C\mid {\bf S}) + I({\bf S};A) + 2L \label{eq:ex_converse_joint_entroy} \\
	& \geq H(A \mid {\bf S}) + I({\bf S};A) + 2L \label{eq:ex_converse_A_is_a_function} \\
	& = H(A) + 2L \label{eq:ex_converse_mi3}  \\
	& = 3L \label{eq:ex_converse_evaluate_entropy2} .
\end{align}
Step \eqref{eq:ex_converse_evaluate_entropy} and \eqref{eq:ex_converse_evaluate_entropy2} apply the assumption that $A,B,C$ are i.i.d. uniform in ${\Bbb F}_q$. Step \eqref{eq:ex_converse_joint_entroy} uses the general property of joint entropy that $H(X\mid Z)+H(Y\mid Z)\geq H(X,Y\mid Z)$ for any random variables $X,Y,Z$. Step \eqref{eq:ex_converse_A_is_a_function} is obtained by recognizing that $A$ is a function of $(A+B+C, A+2B+2C)$. Step \eqref{eq:ex_converse_mi3} uses the information equality $I(A;B)=H(A)-H(A|B)$.
Therefore, we have the desired converse bound, $\Delta=N/L\geq H({\bf S})/L\geq 1.5$ for the coding scheme. Since this is true for every feasible coding scheme, we have the bound $\Delta^*\geq 1.5$.
\subsubsection{General Proof of Converse Bound $\Delta^*\geq \Delta_2$}\label{sec:congen}
As mentioned in Section \ref{sec:entropy} let us start the converse proof based on the thought-experiment that supposes the elements of the data ${\bf X}$ are i.i.d. uniform in $\mathbb{F}_q$.
\begin{align}
	&2H(\mathbf{S})+  2\sum_{k=1}^3 H(\mathbf{W}_k')\notag\\
	&=2\sum_{k=1}^3 H(\mathbf{S}, \mathbf{W}_k')   + 2\sum_{k=1}^3 I(\mathbf{S}; \mathbf{W}_k') -4 H({\bf S})\label{eq:defMIb} \\
	&=2\sum_{k=1}^3 H(\mathbf{S}, \mathbf{W}_k', \mathbf{W}_k)   + 2\sum_{k=1}^3 I(\mathbf{S}; \mathbf{W}_k') -4 H({\bf S})\label{eq:defdecodeb} \\	
	&=\sum_{(i,j)\in\{(1,2),(2,3),(1,3)\}}\left(H({\bf S},{\bf X}^T\bU_i)+H({\bf S},{\bf X}^T\bU_j)\right)+ 2\sum_{k=1}^3 I(\mathbf{S}; \mathbf{W}_k') -4 H({\bf S})\label{eq:useUdefb} \\
	&\geq \left[H({\bf S}, {\bf X}^T\bU_{12})+H({\bf S}, {\bf X}^T\bU_{13})\right]+\left[H({\bf S},{\bf X}^T[{\bf U}_1,{\bf U}_2])+H({\bf S},{\bf X}^T{\bf U}_3)\right] \notag\\
	&\hspace{1cm}+\left[ H({\bf S},{\bf X}^T[{\bf U}_1,{\bf U}_3])+H({\bf S},{\bf X}^T{\bf U}_2)\right]+ 2\sum_{k=1}^3 I(\mathbf{S}; \mathbf{W}_k') -4 H({\bf S})\label{eq:usefuncsub1b}\\
	&\geq H({\bf S},{\bf X}^T{\bf U}_{123})+H({\bf S},{\bf X}^T[{\bf U}_{12},{\bf U}_{13}])+H({\bf S},{\bf X}^T{\bf U}_{3(1,2)})+H({\bf S},{\bf X}^T{\bf U}_{2(1,3)}) \notag \\
		&\hspace{1cm}+2H({\bf S},{\bf X}^T[{\bf U}_1,{\bf U}_2,{\bf U}_3])+ 2\sum_{k=1}^3 I(\mathbf{S}; \mathbf{W}_k') -4 H({\bf S})\label{eq:usefuncsub2b}\\
		& \geq H({\bf X}^T{\bf U}_{123}|{\bf S})+H({\bf X}^T[{\bf U}_{12},{\bf U}_{13}]|{\bf S})+H({\bf X}^T{\bf U}_{3(1,2)}\mid{\bf S})+H({\bf X}^T{\bf U}_{2(1,3)}|{\bf S}) \notag \\
		&\hspace{1cm}+2H({\bf X}^T[{\bf U}_1,{\bf U}_2,{\bf U}_3]) + 2\sum_{k=1}^3 I(\mathbf{S}; \mathbf{W}_k')\label{eq:conv_2_RHSb}\\
		& \geq \max_{\ell\in\{1,2,3\}}H({\bf X}^T({\bf U}_{123}\cap \bV_\ell')\mid{\bf S})+H({\bf X}^T[{\bf U}_{12},{\bf U}_{13}]\mid{\bf S})+H({\bf X}^T{\bf U}_{3(1,2)}\mid{\bf S})+H({\bf X}^T{\bf U}_{2(1,3)}\mid {\bf S}) \notag \\
		&\hspace{1cm}+2H({\bf X}^T[{\bf U}_1,{\bf U}_2,{\bf U}_3]) + 2\sum_{k=1}^3 I(\mathbf{S}; \mathbf{W}_k')\label{eq:fb}\\		
	& \geq \max_{\ell\in \{1,2,3\}}H({\bf X}^T({\bf U}_{123}\cap {\bf V}_\ell')) +H({\bf X}^T[{\bf U}_{12},{\bf U}_{13}]\mid{\bf S})+H({\bf X}^T{\bf U}_{3(1,2)}\mid{\bf S})+H({\bf X}^T{\bf U}_{2(1,3)}\mid{\bf S}) \notag \\
		&\hspace{1cm}+2H({\bf X}^T[{\bf U}_1,{\bf U}_2,{\bf U}_3]) + \sum_{k=1}^3 I(\mathbf{S}; \mathbf{W}_k') \label{eq:useDPI2b} \\
	&\geq \max_{\ell\in \{1,2,3\}}H({\bf X}^T({\bf U}_{123}\cap {\bf V}_\ell')) +H({\bf X}^T([{\bf U}_{12},{\bf U}_{13}]\cap {\bf V}_1'))+H({\bf X}^T({\bf U}_{3(1,2)}\cap {\bf V}_3')) \notag \\
		&\hspace{1cm}+H({\bf X}^T({\bf U}_{2(1,3)}\cap {\bf V}_2'))+2H({\bf X}^T[{\bf U}_1,{\bf U}_2,{\bf U}_3])  \label{eq:useDPI3b}
\end{align}
In the deduction, the most critical steps are Step \eqref{eq:usefuncsub1b} and Step \eqref{eq:usefuncsub2b}. Specifically, Step \eqref{eq:usefuncsub1b} uses functional submodularity property from Corollary \ref{cor:submod} twice, once for $(i,j)=(1,2)$ and once for $(i,j)=(1,3)$. Step \eqref{eq:usefuncsub2b} uses functional submodularity from Corollary \ref{cor:submod} three times, once for each of the collections of terms inside the three square parantheses in \eqref{eq:usefuncsub1b}, making use of the fact that $\bU_{12}\cap\bU_{13}=\bU_{123}$, $[\bU_1,\bU_2]\cap\bU_3=\bU_{3(1,2)}$, and $[\bU_1,\bU_3]\cap\bU_2=\bU_{2(1,3)}$ by definition. The other steps follow from conventional entropic inequalities. Specifically, Step \eqref{eq:defMIb} uses the definition of mutual information $I(A; B)=H(A)+H(B)-H(A,B)$. Step \eqref{eq:defdecodeb} uses the decoder definition \eqref{eq:dec}, i.e., ${\bf W}_k$ is a function of $({\bf S}, {\bf W}_k')$. Step \eqref{eq:useUdefb} uses the definition of $\bU_k=[\bV_k',\bV_k]$  from \eqref{eq:Udef} to recognize ${\bf X}^T{\bf U}_k=[{\bf W}_k', {\bf W}_k]$.  Step \eqref{eq:conv_2_RHSb} uses the chain rule of entropy to extract $H({\bf S})$ from the first four terms, and the property that $H(A,B)\geq H(B)$ to drop ${\bf S}$ from the fifth term.  Step \eqref{eq:fb} uses the property that $H(A\mid B)\geq H(f(A)\mid B)$, and the fact that ${\bf X}^T(\bU_{123}\cap\bV_\ell')$ is a function of ${\bf X}^T{\bf U}_{123}$. Step \eqref{eq:useDPI2b} uses the fact that $H({\bf X}^T(\bU_{123}\cap \bV_\ell')\mid{\bf S})=H({\bf X}^T(\bU_{123}\cap \bV_\ell'))-I({\bf S}; {\bf X}^T(\bU_{123}\cap \bV_\ell'))$ by definition of mutual information, and $I({\bf S}; {\bf X}^T(\bU_{123}\cap \bV_\ell'))\leq I({\bf S}; {\bf X}^T \bV_\ell')=I({\bf S}; {\bf W}_\ell')\leq \sum_{k=1}^3I({\bf S};{\bf W}_k')$ by data-processing inequality, and the non-negativity of mutual information. Similar reasoning is applied to the third, fourth and fifth terms of \eqref{eq:useDPI2b} to obtain \eqref{eq:useDPI3b} by removing the conditioning on ${\bf S}$ and by absorbing one of the $I({\bf S}; {\bf W}_k')$ terms each. The reasoning can be summarized as $H(A\mid S)+I(B;S)\geq  H(C\mid S)+I(C;S)=H(C)$ if $C$ is a function of both $A$ and $B$ individually.
 
Evaluating the entropies in terms of the corresponding ranks, and normalizing by $L$, we obtain,
\begin{align} 
	& \Delta = N/L \geq H({\bf S})/L\notag\\
	& \geq \frac{1}{2}\Bigg(\max_{l\in \{1,2,3\}}\rk(({\bf U}_{123}\cap {\bf V}_l')) +\rk([{\bf U}_{12},{\bf U}_{13}]\cap {\bf V}_1')+ \rk({\bf U}_{3(1,2)}\cap {\bf V}_3')+\rk({\bf U}_{2(1,3)}\cap {\bf V}_2')\Bigg) \notag \\
		&\hspace{1cm}+\rk([{\bf U}_1,{\bf U}_2,{\bf U}_3]) -\sum_{k=1}^3\rk({\bf V}_k') \label{eq:rankdefb}\\
	&= \frac{1}{2}\Bigg(\rk({\bf U}_{123})- \min_{\ell\in \{1,2,3\}}\rk({\bf U}_{123}\mid{\bf V}_\ell')+\rk([{\bf U}_{12},{\bf U}_{13}])   -\rk([{\bf U}_{12},{\bf U}_{13}]\mid \bV_1')+ \rk({\bf U}_{3(1,2)})\notag\\
		&\hspace{1cm}-\rk({\bf U}_{3(1,2)}\mid\bV_3')+\rk({\bf U}_{2(1,3)})-\rk({\bf U}_{2(1,3)}\mid \bV_2')\Bigg) +\rk([{\bf U}_1,{\bf U}_2,{\bf U}_3]) -\sum_{k=1}^3\rk({\bf V}_k')\label{eq:usedefconrank1b} \\
		&= \frac{1}{2}\Bigg(\rk({\bf U}_{12})+\rk({\bf U}_{13}) - \min_{\ell\in \{1,2,3\}}\rk({\bf U}_{123}\mid{\bf V}_\ell')  -\rk([{\bf U}_{12},{\bf U}_{13}]\mid \bV_1')+\rk(\bU_3)+\rk([\bU_1,\bU_2]) \notag\\
		&\hspace{1cm}-\rk({\bf U}_{3(1,2)}\mid\bV_3')+\rk(\bU_2)+\rk([\bU_1,\bU_3])-\rk({\bf U}_{2(1,3)}\mid \bV_2')\Bigg)  -\sum_{k=1}^3\rk({\bf V}_k') \label{eq:usedefunionrank1b}\\
		&= \rk(\bU_1)+\rk(\bU_2)+\rk(\bU_3)-\sum_{k=1}^3\rk({\bf V}_k')\notag\\
		&\hspace{0.5cm}-\frac{1}{2}\Bigg(\min_{\ell\in \{1,2,3\}}\rk({\bf U}_{123}\mid{\bf V}_\ell')  +\rk([{\bf U}_{12},{\bf U}_{13}]\mid \bV_1') +\rk({\bf U}_{3(1,2)}\mid\bV_3')+\rk({\bf U}_{2(1,3)}\mid \bV_2')\Bigg) \label{eq:splitrank1b}  \\
		&= \rk(\bV_1\mid \bV_1')+\rk(\bV_2\mid\bV_2')+\rk(\bV_3\mid\bV_3')\notag\\
		&\hspace{0.5cm}-\frac{1}{2}\Bigg(\min_{\ell\in \{1,2,3\}}\rk({\bf U}_{123}\mid{\bf V}_\ell')  +\rk([{\bf U}_{12},{\bf U}_{13}]\mid \bV_1') +\rk({\bf U}_{3(1,2)}\mid\bV_3')+\rk({\bf U}_{2(1,3)}\mid \bV_2')\Bigg)   \label{eq:UgivenV}
\end{align}
Step \eqref{eq:rankdefb} uses the fact that for i.i.d. uniform data ${\bf X}^T\in\mathbb{F}_q^{L\times d}$ and an arbitrary matrix ${\bf M}\in\mathbb{F}_q^{d\times \mu}$, $H({\bf X}^T{\bf M})=L\cdot\rk({\bf M})$ in $q$-ary units. Step \eqref{eq:usedefconrank1b} applies the conditional-rank notation, $\rk(A\mid B)=\rk([A,B])-\rk(B)$ as defined in Section \ref{sec:notation}. Step \eqref{eq:usedefunionrank1b} uses the fact that $\rk(\bU_{123})=\rk(\bU_{12}\cap\bU_{13})=\rk(\bU_{12})+\rk(\bU_{13})-\rk([\bU_{12},\bU_{13}])$,  similarly $\rk(\bU_{3(1,2)})=\rk(\bU_3\cap [\bU_1,\bU_2])=\rk(\bU_3)+\rk([\bU_1,\bU_2])-\rk([\bU_1,\bU_2,\bU_3])$,  and by the same token $\rk(\bU_{2(1,3)})=\rk(\bU_2\cap [\bU_1,\bU_3])=\rk(\bU_2)+\rk([\bU_1,\bU_3])-\rk([\bU_1,\bU_2,\bU_3])$. Step \eqref{eq:splitrank1b} uses the fact that $\rk(\bU_{ij})=\rk(\bU_i)+\rk(\bU_j)-\rk([\bU_i,\bU_j])$. Finally, step \eqref{eq:UgivenV} uses the fact that $\rk(\bU_i)-\rk(\bV_i')=\rk([\bV_i', \bV_i])-\rk(\bV_i')=\rk(\bV_i\mid \bV_i')$.

Since this holds for every coding scheme $\big(L,N,\Phi,(\Psi_{k})_{k\in[3]}\big)\in\mathfrak{C}$, it follows that $\Delta^*\geq \Delta_2^{123}$. Similarly, $\Delta^*\geq \Delta_2^{ijk}$,  $\forall(i,j,k)$ that are permutations of $(1,2,3)$, which implies that $\Delta^*\geq \Delta_2$.$\hfill \qed$

\section{Proof of Achievability: $ \Delta^*\leq F^*$} \label{sec: Achievability}
In this section, we  will construct a general scheme for the $3$ user LCBC that achieves broadcast cost per computation equal to $F^*$ as specified in the form of a linear program in Theorem \ref{thm:LP}, thus establishing an upper bound on the optimal broadcast cost per computation, $\Delta^*\leq F^*$.  Finding an explicit solution to the linear program in closed form will be left for Section \ref{sec:linear_prog}. We start this proof with some preliminary steps.

\subsection{Eliminating Redundancies}
As a first step let us eliminate redundancies, if any, that exist in the users' side-information and demands by removing redundant columns in  ${\bf V}_k, {\bf V}_k'$ such that ${\bf U}_k=[{\bf V}_k', {\bf V}_k]$ has full column rank for each $k\in[3]$. Essentially, we retain only linearly independent columns because the remaining columns either represent information desired by a user that is already available to the user (overlap between $\langle {\bf V}_k\rangle$ and $  \langle {\bf V}_k'\rangle$), or information that is already accounted for by the independent columns (redundancies within ${\bf V}_k$ or within ${\bf V}_k'$). Thus, henceforth  let us assume, without loss of generality, that 
\begin{align} \label{eq:full_rank_Uk}
\rk({\bf U}_k) = \rk([\bV_k', \bV_k]) = \rk(\bV_k')+\rk(\bV_k) = m_k'+m_k.
\end{align}

\subsection{Field Size Extension} \label{subsec:field_size}
Recall that the problem formulation specifies a field $\mathbb{F}_q$, but allows us to choose the number of computations $L$ to be encoded together as a free parameter in the achievable scheme. The freedom in the choice of $L$ in fact allows  field extensions that translate the specified field of operations from $\mathbb{F}_q$ to $\mathbb{F}_{q^z}$ for arbitrary $z\in \mathbb{N}$.  
Specifically, consider $L=z$ computations, and denote $\bar{{\bf V}}_k' = {\bf V}_k' \otimes {\bf I}^{z\times z}$, $\bar{{\bf V}}_k = {\bf V}_k \otimes {\bf I}^{z\times z}$ and $\bar{{\bf U}}_k = {\bf U}_k \otimes {\bf I}^{z\times z}$ as the {\it $z$-extension} of the coefficient matrices, where $\otimes$ denotes the Kronecker product.   Denote $\bar{{\bf X}} = \mbox{vec}({\bf X}^T)$, where $\mbox{vec}(\cdot)$ is the vectorization function. By this notation, we can restate the problem such that User $k$ has side-information $\bar{{\bf X}}^T\bar{{\bf V}}_k'$ and wants to compute $\bar{{\bf X}}^T\bar{{\bf V}}_k$ for $k=[1:3]$, where $\bar{{\bf X}} \in {\Bbb F}_q^{dz\times 1}$, $\bar{{\bf V}}_k' \in {\Bbb F}_q^{dz\times m'z}$ and $\bar{{\bf V}}_k \in {\Bbb F}_q^{dz\times mz}$. Now, since $\mathbb{F}_q$ is a subfield of $\mathbb{F}_{q^z}$, this problem is equivalent to the problem where ${\bf X} \in {\Bbb F}_{q^z}^{d \times 1}$, ${\bf V}_k'\in{\Bbb F}_{q^z}^{d\times m'}$ and ${\bf V}_k\in{\Bbb F}_{q^z}^{d\times m}$ for $L=1$ computation.  By considering the elements in ${\Bbb F}_{q^z}$ instead of ${\Bbb F}_{q}$, we have more flexibility in designing schemes by choosing symbols in the extension field to jointly code over $z$ computations. Since the  achievable scheme allows joint coding over any $L$ computations, considering $L = L'z$ computations in the original problem with field $\mathbb{F}_{q}$ is equivalent to considering $L'$ computations in the extended field with $\mathbb{F}_{q^z}$. Appendix \ref{app:field_extension} illustrates the idea of field size extension with an example.

\subsection{Useful Lemma}
Next let us introduce a useful lemma.
\begin{lemma} \label{lem:2_complement}
	Let $A\in \mathbb{F}_q^{d\times a}, B_1\in \mathbb{F}_q^{d\times b_1}$ and $B_2\in \mathbb{F}_q^{d\times b_2}$ be arbitrary matrices with full column rank (bases), i.e., $\rk(A)=a, \rk(B_1)=b_1, \rk(B_2)=b_2$. Denote $\rk(B_1\mid A) = r_{1|A}$, $\rk(B_2\mid A) = r_{2|A}$ and $\rk([B_1,B_2]\mid A) = r_{1,2|A}$. Then for any non-negative integers $n_1, n_2$ such that $n_1\leq r_{1|A}, n_2\leq r_{2|A}$ and $n_1+n_2\leq r_{1,2|A}$, there exist submatrices of $B_1, B_2$, namely $B_1'\in \mathbb{F}_q^{d\times n_1}$ and $B_2'\in \mathbb{F}_q^{d\times n_2}$, respectively, such that $[A,B_1',B_2']$ has full column rank $a+n_1+n_2$.
\end{lemma}
\begin{proof}
Consider first the case that $n_1+n_2=r_{1,2|A}$. By Steinitz Exchange lemma there exist submatrices $B_1^{(r_{1|A})}, B_2^{(r_{2|A})}$, comprised of $r_{1|A}, r_{2|A}$ columns of $B_1, B_2$,  respectively, such that $[A, B_1^{(r_{1|A})}],$ $[A,B_2^{(r_{2|A})}]$ have full column ranks (the superscripts within the parantheses indicate the number of columns). Now, we claim that if $Y^{(a+r_{1|A}+r_{2|A})}=[A, B_1^{(r_{1|A})}, B_2^{(r_{2|A})}]$ does not have full column rank, i.e., $a+r_{1|A}+r_{2|A}>\rk(Y^{(a+r_{1|A}+r_{2|A})})=a+r_{1,2|A}$,  then it is always possible to drop a column of  $B_1^{(r_{1|A})}$ to yield $Y^{(a+r_{1|A}+r_{2|A}-1)}=[A, B_1^{(r_{1|A}-1)}, B_2^{(r_{2|A})}]$ which has one less column but the same column rank as $Y^{(a+r_{1|A}+r_{2|A})}$. The claim is proved as follows. Since $Y^{(a+r_{1|A}+r_{2|A})}$ does not have full column rank, there exists a non-zero column vector $Z$, such that $Y^{(a+r_{1|A}+r_{2|A})}Z=0_{d\times 1}$. This non-zero  vector $Z$ must have more than one non-zero element (because $Y^{(a+r_{1|A}+r_{2|A})}$ has non-zero columns), and at least one of its non-zero elements must be in a row-index that maps to one of the columns of $B_1^{(r_{1|A})}$ (because $[A, B_2^{(r_{2|A})}]$ has full column rank). This column of $B_1^{(r_{1|A})}$ can be dropped because it is spanned by the remaining columns of $Y^{(a+r_{1|A}+r_{2|A})}$ that are selected by the support of $Z$, so that $Y^{(a+r_{1|A}+r_{2|A}-1)}$ has the same rank as $Y^{(a+r_{1|A}+r_{2|A})}$. The same claim holds for $B_2^{(r_{2|A})}$ as well. Repeating this argument we can drop columns of $B_1^{(r_{1|A})}, B_2^{(r_{2|A})}$, one-by-one, in any order we wish, until we meet the target values $n_1, n_2$ at which point the resulting matrix $[A, B_1', B_2']$ has full column rank, equal to $a+r_{1,2|A}$. Finally, if $n_1+n_2<r_{1,2|A}$, then we continue the process for an additional $r_{1,2|A}-(n_1+n_2)$ steps, but each additional column that is dropped now reduces both the rank and the number of columns by $1$, until $B_1', B_2'$ are left with only $n_1,n_2$ columns, respectively, and $\rk(Y^{(a+r_{1|A}+r_{2|A})})=a+r_{1,2|A}-(r_{1,2|A}-(n_1+n_2))=a+n_1+n_2$.\end{proof}

Let us also note the following direct corollary of Lemma \ref{lem:2_complement} which will be used multiple times in our construction of the coding scheme.
\begin{corollary} \label{cor:1_complement}
	Let $A\in \mathbb{F}_q^{d\times a}$ and $B\in \mathbb{F}_q^{d\times b}$ be arbitrary matrices with full column rank (bases), i.e., $\rk(A)=a, \rk(B)=b$.  Denote $\rk(B\mid A) = r$. Then for any non-negative integer $n$ such that $n\leq r$, there exists a submatrix of $B$, namely $B'\in \mathbb{F}_q^{d\times n}$ such that $[A,B']$ has full column rank $a+n$.
\end{corollary}
\begin{proof}
	Corollary \ref{cor:1_complement} is implied by Lemma \ref{lem:2_complement}, by mapping $A$, $B$ here to $A$, $B_1$ in Lemma \ref{lem:2_complement}, respectively, and setting $b_2$ = 0. 
\end{proof}

\subsection{Construction of the Optimal Broadcast Scheme} \label{sec:scheme_construction}
The construction of the optimal broadcast information follows the formulation of Theorem \ref{thm:LP} and the depiction in Figure \ref{fig:venn2}. At a high level, the goal is to construct a scheme that broadcasts $\lambda_{123}$ dimensions that are simultaneously useful to all $3$ users ($3$ birds, $1$ stone), $\lambda_{ij}$ dimensions that are simultaneously useful to Users $i,j$ ($2$ birds, $1$ stone) for $(i,j)\in\{(1,2),(1,3),(2,3)\}$, and $\lambda$ dimensions that are of the type ($3$ birds, $2$ stones), i.e., where transmission of $2$ dimensions collectively satisfies $1$ demand dimension for every user. For this construction, let us first consider non-negative integers $\lambda_{123}, \lambda_{12}, \lambda_{13}, \lambda_{23}, \lambda$ that satisfy the constraints \eqref{eq:constraintfirst}-\eqref{eq:constraintlast} specified in Theorem \ref{thm:LP}. Generalization of $\lambda_{123}, \lambda_{12}, \lambda_{13}, \lambda_{23}, \lambda$ to  rationals is handled in Section \ref{sec:matex} and \ref{sec:achievablef}.

Let us start with the ($3$ birds, $1$ stone) component of the construction, and for now let us focus on User $1$. Some adjustments will be necessary eventually to make the scheme work for all $3$ users. We wish to broadcast $\lambda_{123}$ dimensions for this ($3$ birds, $1$ stone) component of our scheme, but it remains to determine the actual information to be transmitted. For this, let us recall Corollary \ref{cor:1_complement}, which guarantees that there exists a submatrix of ${\bf U}_{123}$, namely ${\bf U}_{123}^{(\lambda_{123})}\in\mathbb{F}_{q^z}^{d\times \lambda_{123}}$,  such that the following matrix has full column rank,
\begin{align}
\rk([{\bf V}_1',{\bf U}_{123}^{(\lambda_{123})}])=m_1'+\lambda_{123}.
\end{align}
Broadcasting ${\bf X}^T{\bf U}_{123}^{(\lambda_{123})}$ would help User $1$ acquire $\lambda_{123}$ desired dimensions based on his side-information ${\bf X}^T\bV_1'$.

As a cautionary note, let us point out that this particular ${\bf U}_{123}^{(\lambda_{123})}$ which is useful for User $1$ may not be useful for User $2$ or User $3$, i.e., $[{\bf V}_k',{\bf U}_{123}^{(\lambda_{123})}]$ may not have full column rank for $k=2,3$. One can similarly find submatrices of ${\bf U}_{123}$ of size (number of columns) $\lambda_{123}$ that are useful for User $2$, or $3$ individually, but in general these will be different matrices. In the end the challenge will be to find the \emph{same} matrix that is useful for all three users. For now we ignore this challenge and proceed with only User $1$ as our focus.

Next, consider the ($2$ birds, $1$ stone) components, specifically let us find $\lambda_{12}$  dimensions within $\langle\bU_{12}\rangle$, and another $\lambda_{13}$ dimensions within $\langle \bU_{13}\rangle$ that will be useful to User $1$, conditioned on the user's side-information $\bV_1'$. Letting $A = [{\bf V}_1',{\bf U}_{123}^{(\lambda_{123})}]$, $B_1 = {\bf U}_{12}$ and $B_2 = {\bf U}_{13}$ in Lemma \ref{lem:2_complement}, we have $a = m_1'+\lambda_{123}$, $r_{1|A}=\rk(\bU_{12}\mid \bV_1')-\lambda_{123}$, $r_{2|A}= \rk(\bU_{13}\mid \bV_1')-\lambda_{123}$, and $r_{1,2|A} = \rk([\bU_{12},\bU_{13}]\mid \bV_1')-\lambda_{123}$. Then according to Lemma \ref{lem:2_complement}, there exists a submatrix of ${\bf U}_{12}$, namely, ${\bf U}_{12}^{(\lambda_{12})} \in {\Bbb F}_{q^z}^{d\times \lambda_{12}}$, and a submatrix of ${\bf U}_{13}$, namely, ${\bf U}_{13}^{(\lambda_{13})} \in {\Bbb F}_{q^z}^{d\times \lambda_{13}}$,
such that the following matrix has full column rank,
\begin{align}
	\rk([{{\bf V}_1'}, {\bf U}_{123}^{(\lambda_{123})}, {\bf U}_{12}^{(\lambda_{12})}, {\bf U}_{13}^{(\lambda_{13})}])&=m_1'+\lambda_{123}+\lambda_{12}+\lambda_{13}.
\end{align}
Once again, note that these choices may not work for Users $2, 3$, so that challenge remains to be overcome later.

Next, consider the ($2$ stones, $3$ birds) component of the scheme. Keeping our focus on User $1$, let us find $\lambda$ dimensions of broadcast information from the subspace $\langle\bU_{1(2,3)} \rangle$ that will be useful for User $1$. Since we only consider parameters that satisfy the conditions in Theorem \ref{thm:LP}, which include in particular \eqref{eq:constraintlast}, it follows that $\lambda\leq \rk(\bU_{1(2,3)}\mid \bV_1')-\lambda_{12}-\lambda_{13}-\lambda_{123}$ by definition. Letting $A = [{{\bf V}_1'}, {\bf U}_{123}^{(\lambda_{123})}, {\bf U}_{12}^{(\lambda_{12})}, {\bf U}_{13}^{(\lambda_{13})}]$, $B = {\bf U}_{1(2,3)}$ in Corollary \ref{cor:1_complement}, we have $a=m_1'+\lambda_{123}+\lambda_{12}+\lambda_{13}$ and $r= \rk(\bU_{1(2,3)}\mid \bV_1')-\lambda_{12}-\lambda_{13}-\lambda_{123}\geq \lambda$. Then Corollary \ref{cor:1_complement} implies that there exists a submatrix of ${\bf U}_{1(2,3)}$, namely, ${\bf U}_{1(2,3)}^{(\lambda)} \in {\Bbb F}_{q^z}^{d\times \lambda}$ such that the following matrix has full column rank,
\begin{align}
	&\rk([{{\bf V}_1'}, {\bf U}_{123}^{(\lambda_{123})}, {\bf U}_{12}^{(\lambda_{12})}, {\bf U}_{13}^{(\lambda_{13})}, {\bf U}_{1(2,3)}^{(\lambda)}]) =m_1'+\lambda_{123}+\lambda_{12}+\lambda_{13}+\lambda.
\end{align}
Next, by letting $A$ be the above matrix and $B={\bf U}_1$ in Corollary \ref{cor:1_complement}, we have $a = m_1'+\lambda_{123}+\lambda_{12}+\lambda_{13}+\lambda$, $r = \rk(\bU_1\mid \bV_1') - (\lambda_{123}+\lambda_{12}+\lambda_{13}+\lambda) = m_1-(\lambda_{123}+\lambda_{12}+\lambda_{13}+\lambda) \triangleq t_1$. Then by Corollary \ref{cor:1_complement}, there exists a submatrix of ${\bf U}_1$, namely, ${\bf U}_1^{(t_1)} \in {\Bbb F}_{q^z}^{d\times t_1}$ such that the following matrix has full column rank.
\begin{align} \label{eq:mat_user1}
	\rk([{{\bf V}_1'}, {\bf U}_{123}^{(\lambda_{123})}, {\bf U}_{12}^{(\lambda_{12})}, {\bf U}_{13}^{(\lambda_{13})}, {\bf U}_{1(2,3)}^{(\lambda)}, {\bf U}_1^{(t_1)}])&=m_1+m_1',
\end{align}
which implies that it is a basis of $\langle {\bf U}_1 \rangle$ since each column of the matrix is in $\langle {\bf U}_1 \rangle$.

Finally, in Corollary \ref{cor:1_complement} let $A$ be the matrix in \eqref{eq:mat_user1} and $B = {\bf I}^{d\times d}$ be the $d\times d$ identity matrix. We have $a = m_1+m_1'$ and $r = d-m_1-m_1'$. Then by Corollary \ref{cor:1_complement}, there exists a submatrix of ${\bf I}^{d\times d}$, namely, ${\bf Z}_1 \in {\Bbb F}_{q^z}^{d\times (d-m_1-m_1')}$ such that the following $d\times d$ matrix has full rank.
\begin{align}\label{eq:det_user1local}
	\rk([{{\bf V}_1'}, {\bf U}_{123}^{(\lambda_{123})}, {\bf U}_{12}^{(\lambda_{12})}, {\bf U}_{13}^{(\lambda_{13})}, {\bf U}_{1(2,3)}^{(\lambda)}, {\bf U}_1^{(t_1)}, {\bf Z}_1])&=d.
\end{align}
In particular, the determinant of the matrix is non-zero.

The following step allows a mixing of information, leading to a random-coding argument that will be important to reconcile the users' different perspectives. So consider the following determinant, which is a polynomial in the variables corresponding to the elements of the matrices ${\bf N}_{123},{\bf N}_{12},{\bf N}_{13},{\bf M}_{12},{\bf M}_{13},{\bf M}$, while the remaining matrices are fixed.
\begin{align}\label{eq:det_user1}
	&P_1 = \det \big([ \overbrace{{\bf V}_1'}^{m_1'}, \overbrace{{\bf U}_{123}{\bf N}_{123}}^{\lambda_{123}}, \overbrace{{\bf U}_{12}{\bf N}_{12}}^{\lambda_{12}}, \overbrace{{\bf U}_{13}{\bf N}_{13}}^{\lambda_{13}}, \overbrace{{\bf U}_{12}{\bf M}_{12}+{\bf U}_{13}{\bf M}_{13}+{\bf B}_{1(2,3)} {\bf M}}^{\lambda},{\bf U}_1^{(t_1)}, {\bf Z}_1]  \big)
\end{align}
The sizes of the variable matrices are specified below.
\begin{align}\label{eq:sizespec}
	&{\bf N}_{123}: \rk({\bf U}_{123}) \times \lambda_{123}, && {\bf N}_{12}: \rk({\bf U}_{12}) \times \lambda_{12} , && {\bf N}_{13}: \rk({\bf U}_{13}) \times \lambda_{13}, \\
	&{\bf M}_{12}: \rk({\bf U}_{12}) \times \lambda, && {\bf M}_{13}: \rk({\bf U}_{13}) \times \lambda,&&{\bf M}: \rk({\bf B}_{1(2,3)}) \times \lambda.
\end{align}
We claim that $P_1$ is not a zero polynomial. This is because we can assign values to the variables such that the matrix in \eqref{eq:det_user1} becomes identical to the constant matrix in \eqref{eq:det_user1local}, which has non-zero determinant. Note that by Lemma \ref{lemma:decomposition}, {\it (P5)} and {\it(P8)}, $\langle\bU_{1(2,3)}\rangle=\langle[\bB_{123}, \bB_{12} , \bB_{13} , \bB_{1(2,3)} ]\rangle=\langle[\bU_{12},\bU_{13}, \bB_{1(2,3)} ]\rangle$, therefore,  ${\bf U}_{1(2,3)}^{(\lambda)}={\bf U}_{12}{\bf M}_{12}+{\bf U}_{13}{\bf M}_{13}+{\bB}_{1(2,3)} {\bf M}$ for some realization of ${\bf M}_{12}, {\bf M}_{13}, {\bf M}$.  Since there exists a non-zero evaluation of $P_1$ it cannot be the zero polynomial.

So far our discussion focused on User $1$. Proceeding similarly for Users $2$ and $3$ we arrive at corresponding polynomials $P_2, P_3$ as shown below, 
\begin{align} \label{eq:det_user2}
	&P_2 = \det ([ \overbrace{{\bf V}_2'}^{m_2'}, \overbrace{{\bf U}_{123}{\bf N}_{123}}^{\lambda_{123}}, \overbrace{{\bf U}_{12}{\bf N}_{12}}^{\lambda_{12}}, \overbrace{{\bf U}_{23}{\bf N}_{23}}^{\lambda_{23}}, \overbrace{-{\bf U}_{12}{\bf M}_{12}+{\bf U}_{23}{\bf M}_{23}+{\bf B}_{2(1,3)} {\bf M}}^{\lambda},{\bf U}_2^{(t_2)}, {\bf Z}_2]  )\\
	&P_3 = \det ([ \overbrace{{\bf V}_3'}^{m_3'}, \overbrace{{\bf U}_{123}{\bf N}_{123}}^{\lambda_{123}}, \overbrace{{\bf U}_{13}{\bf N}_{13}}^{\lambda_{13}}, \overbrace{{\bf U}_{23}{\bf N}_{23}}^{\lambda_{23}}, \overbrace{{\bf U}_{13}{\bf M}_{13}+{\bf U}_{23}{\bf M}_{23}+{\bf B}_{3(1,2)} {\bf M}}^{\lambda},{\bf U}_3^{(t_3)}, {\bf Z}_3]  )
\end{align}
that are similarly shown to be non-zero polynomials, in the variables  corresponding to the elements of the matrices ${\bf N}_{123}$, ${\bf N}_{12}$, ${\bf N}_{13}$, ${\bf N}_{23}$, ${\bf M}_{12}$, ${\bf M}_{13}$, ${\bf M}_{23}$, ${\bf M}$, with the following  remaining specifications in addition to those in \eqref{eq:sizespec}.
\begin{align}
&{\bf M}_{23}: \rk({\bf U}_{23}) \times \lambda, && {\bf N}_{23}: \rk(\bU_{23})\times\lambda_{23},
\end{align}
and
\begin{align}
{\bf Z}_2 &\in {\Bbb F}_{q^z}^{d\times (d-m_2-m_2')}, \\
{\bf Z}_3 &\in {\Bbb F}_{q^z}^{d\times (d-m_3-m_3')},\\
t_2&\triangleq m_2-(\lambda_{123}+\lambda_{12}+\lambda_{23}+\lambda), \\
t_3&\triangleq m_3-(\lambda_{123}+\lambda_{13}+\lambda_{23}+\lambda).
\end{align}

Note that the minus sign before ${\bf U}_{12}{\bf M}_{12}$ in \eqref{eq:det_user2} still allows the entries of $-{\bf M}_{12}$  to be any element in ${\bf F}_{q^z}$, and thus we can still evaluate the determinants individually to non-zero by choosing appropriate elements in ${\bf F}_{q^z}$. Now since $P_1$, $P_2$ and $P_3$ are non-zero polynomials, their product $P \triangleq P_1P_2P_3$ is also a non-zero polynomial in the variables  corresponding to the elements of the matrices ${\bf N}_{123}$, ${\bf N}_{12}$, ${\bf N}_{13}$, ${\bf N}_{23}$, ${\bf M}_{12}$, ${\bf M}_{13}$, ${\bf M}_{23}$, and ${\bf M}$. Furthermore, the polynomial $P$ has a degree $D$ loosely (the loose bound suffices for our purpose) bounded above as,
\begin{align}
D\leq 3d.
\end{align}
By Schwartz-Zippel Lemma, if the elements of ${\bf N}_{123},{\bf N}_{12},{\bf N}_{13},{\bf N}_{23},{\bf M}_{12},{\bf M}_{13},{\bf M}_{23}, {\bf M}$ are chosen i.i.d uniformly from $\mathbb{F}_{q^z}$, then the probability of $P$ evaluating to $0$ is not more than $\frac{D}{q^z} \leq \frac{3d}{q^z}$. Thus, by choosing $z> \log_q(3d)$, we ensure that there exist such ${\bf N}_{123},{\bf N}_{12},{\bf N}_{13},{\bf N}_{23},{\bf M}_{12},{\bf M}_{13},{\bf M}_{23}, {\bf M}$ that produce a non-zero evaluation of $P$, which implies that $P_1$, $P_2$ and $P_3$ are evaluated to non-zero simultaneously. Recall that we previously found three constructions, by identifying submatrices of subspace matrices, and each such construction could only be guaranteed to work for one user. The  formulation based on ${\bf N}_{123},{\bf N}_{12},{\bf N}_{13},{\bf N}_{23},{\bf M}_{12},{\bf M}_{13},{\bf M}_{23}, {\bf M}$ represents essentially a generic solution for each user. Whereas the original solutions comprised of specific submatrices may not be compatible, the generic solutions turn out to be compatible, as evident in the argument that $P_1, P_2, P_3$ are simultaneously non-zero for appropriate choices of the variables. This is essentially a random coding argument, because it shows the existence of a good code among randomly chosen possibilities.

With any such choice, we are able to construct the broadcast symbol as follows.
\begin{align}
	&{\bf S} = {\bf X}^T [
		{\bf U}_{123}{\bf N}_{123}, {\bf U}_{12}{\bf N}_{12}, {\bf U}_{13}{\bf N}_{13}, {\bf U}_{23}{\bf N}_{23}, {\bf U}_{12}{\bf M}_{12}+{\bf U}_{13}{\bf M}_{13}+{\bf B}_{1(2,3)} {\bf M},...\\
		&~~~~~~~~~~~~~~~~-{\bf U}_{12}{\bf M}_{12}+{\bf U}_{23}{\bf M}_{23} + {\bf B}_{2(1,3)} {\bf M},{\bf U}_1^{(t_1)}, {\bf U}_2^{(t_2)}, {\bf U}_3^{(t_3)}].
\end{align}
With ${\bf S}$, User $1$ is able to obtain (using its side-information)
\begin{align}
	& {\bf X}^T [{\bf V}_1', {\bf U}_{123}{\bf N}_{123}, {\bf U}_{12}{\bf N}_{12},{\bf U}_{13}{\bf N}_{13},  {\bf U}_{12}{\bf M}_{12}+{\bf U}_{13}{\bf M}_{13}+{\bf B}_{1(2,3)} {\bf M}, {\bf U}_1^{(t_1)}],\label{eq:comma}
\end{align}
and thus compute ${\bf X}^T{\bf U}_1$, since the columns of the matrix to the right of ${\bf X}^T$ form a basis of $\langle {\bf U}_1 \rangle$, guaranteed by the fact that $P_1$ has a non-zero evaluation. Similarly, User $2$ is able to obtain (with its side-information) 
\begin{align}
	&{\bf X}^T [{\bf V}_2', {\bf U}_{123}{\bf N}_{123}, {\bf U}_{12}{\bf N}_{12},{\bf U}_{23}{\bf N}_{23},  -{\bf U}_{12}{\bf M}_{12}+{\bf U}_{23}{\bf M}_{23}+{\bf B}_{2(1,3)} {\bf M}, {\bf U}_2^{(t_2)}],
\end{align}
and thus compute ${\bf X}^T{\bf U}_2$, since the columns of the matrix on the right of ${\bf X}^T$ form a basis of  $\langle {\bf U}_2 \rangle$, guaranteed by the fact that $P_2$ has a non-zero evaluation. 

User $3$ first computes
\begin{align}
&{\bf X}^T({\bf U}_{12}{\bf M}_{12}+{\bf U}_{13}{\bf M}_{13}+{\bf B}_{1(2,3)} {\bf M}) \\
&+{\bf X}^T(-{\bf U}_{12}{\bf M}_{12}+{\bf U}_{23}{\bf M}_{23}+{\bf B}_{2(1,3)} {\bf M})\\
&= {\bf X}^T ({\bf U}_{13}{\bf M}_{13}+{\bf U}_{23}{\bf M}_{23}+ {\bf B}_{3(1,2)} {\bf M})
\end{align}
where we used (P20) from Lemma \ref{lemma:decomposition}, i.e., ${\bf B}_{1(2,3)} +{\bf B}_{2(1,3)}  = {\bf B}_{3(1,2)} $. Using its side-information, User $3$ is then able to obtain,
\begin{align}
	&{\bf X}^T [{\bf V}_3', {\bf U}_{123}{\bf N}_{123}, {\bf U}_{13}{\bf N}_{13},{\bf U}_{23}{\bf N}_{23},  {\bf U}_{13}{\bf M}_{13}+{\bf U}_{23}{\bf M}_{23}+{\bf B}_{3(1,2)} {\bf M}, {\bf U}_3^{(t_3)}].
\end{align}
Thus, it can compute ${\bf X}^T{\bf U}_3$, since the matrix on the right of ${\bf X}^T$ is a basis of $\langle {\bf U}_3 \rangle$, guaranteed by the fact that $P_3$  evaluates to a non-zero value.

The cost of this broadcast ${\bf S}$, as noted in Theorem \ref{thm:LP}, is found as,
\begin{align}
	\Delta &= N/L \\
	& = N/z \\
	& = \lambda_{123}+\lambda_{12}+\lambda_{13}+\lambda_{23}+2\lambda+t_1+t_2+t_3\\
	& = m_1+m_2+m_3-2\lambda_{123}-\lambda_{12}-\lambda_{13}-\lambda_{23}-\lambda \label{eq:achm123}\\
	& = \rk(\bV_1|\bV_1')+\rk(\bV_2|\bV_2')+\rk(\bV_3|\bV_3')-2\lambda_{123}-\lambda_{12}-\lambda_{13}-\lambda_{23}-\lambda\label{eq:ach_v123}\\
	& \triangleq f(\lambda_{123},\lambda_{12},\lambda_{13},\lambda_{23},\lambda).\label{eq:achf}
\end{align}
This implies that $ \Delta^* \leq f(\lambda_{123},\lambda_{12},\lambda_{13},\lambda_{23},\lambda)$ if $\lambda_{123},\lambda_{12},\lambda_{13}$, $\lambda_{23}$ and $\lambda$ are non-negative integers subject to the constraints specified in Theorem \ref{thm:LP}. Next let us show that the arguments extend to rational  $\lambda_\bullet$ by a simple matrix extension.

\subsection{Matrix Extension}\label{sec:matex}
Technically, the choice of $z>1$ that  enables field extensions in the achievable scheme, already amounts to vector coding, because it requires joint coding of $L=z$ symbols. However, after the field extension, the solution presented above reduces to a scalar coding solution over the extended field $\mathbb{F}_{q^z}$. This formulation only allows integer values of $\lambda_\bullet$ parameters. However, it is quite straightforward to extend the scheme to all rational values of $\lambda_\bullet$ parameters (subject to the constraints specified in Theorem \ref{thm:LP}) by a typical vector coding extension, labeled here as a Matrix Extension to avoid confusion with field extensions that also require $L>1$. This is described as follows. Recall that we are allowed to choose any $L\in \mathbb{N}$ in the coding schemes, by letting $L = L'z$ (meaning that the computations are in $\mathbb{F}_{q^z}$ and we jointly code for $L'$ such computations), the ranks of all subspaces scale by $L'$  as the data dimension increases by a factor of $L'$. Essentially, this amounts to treating successive instances of the data vector as new data dimensions. For example, consider the $m=1$ dimensional computation of $A+B$ over $d=2$ dimensional data $(A,B)$, say over $\mathbb{F}_{q^z}$. Considering $L'=2$ instances, the data becomes $({\bf A}, {\bf B}) = ((A(1),A(2)), (B(1),B(2))$, and the desired computation is ${\bf A}+{\bf B}$, which can also be interpreted as $m_{\mbox{\tiny new}}=2$ dimensional computations $(A+C, B+D)$ over $d_{\mbox{\tiny new}}=2d=4$ dimensional data $(A, B, C,D)$ in $\mathbb{F}_{q^z}$, by mapping $((A(1),A(2)), (B(1),B(2))$ to $(A,B,C,D)$. A bit more formally, by  considering $L'$ data instances as one instance of $L'd$ dimensional data (both in ${\Bbb F}_{q^z}$), User $k\in[1:3]$ has side-information ${\bf X}^T{\bf V}_k'$, which is equivalent to $\mbox{vec}^T({\bf X}){\bf I}^{L'\times L'}\otimes {\bf V}_k'$.  User $k$ wants to compute $\mbox{vec}^T({\bf X}){\bf I}^{L'\times L'}\otimes {\bf V}_k$. The problem is then equivalent to that with data ${\bf X} \in \mathbb{F}_{q^z}^{L'd\times 1}$, with coefficient matrices now changed to ${\bf I}\otimes {\bf V}_k', {\bf I}\otimes {\bf V}_k, k=[1:3]$. The signal spaces ${\bf U}_k$ are also changed to ${\bf I}\otimes {\bf U}_k$.  
Note that this is essentially different from the field size extension presented in Section \ref{subsec:field_size}, where the dimensions of the coefficient matrices are not changed after the extension, only the field size is changed. We refer to this as the \emph{matrix extension}, since the dimensions (sizes) of the coefficient matrices  scale by a factor of $L'$ (but the field size remains unchanged). The ranks of ${\bf U}_k$ and ${\bf V}_k'$ also scale by $L'$, as do the ranks of all subspaces considered in \eqref{eq:constraintfirst}-\eqref{eq:constraintlast}. Thus, the RHS of all constraints in \eqref{eq:constraintfirst}-\eqref{eq:constraintlast} scale by $L'$, implying a similar scaling of the $\lambda_\bullet$ parameters. Thus, all rational values of $\lambda_\bullet$ parameters can be transformed into integer values by considering a  matrix extension by a factor $L'$ where $L'$ is the common denominator of the rational values.

\subsection{Completing the Proof of Achievability}\label{sec:achievablef}
At this point we have the bound that $\Delta^*\leq f(\lambda_{123}, \lambda_{12}, \lambda_{13}, \lambda_{23}, \lambda)$ if $\lambda_{123},\lambda_{12},\lambda_{13},\lambda_{23},\lambda$ are non-negative rational numbers subject to the constraints specified in \eqref{eq:constraintfirst}-\eqref{eq:constraintlast}. The final step of the achievability proof is to recall \cite{keller2017applied,Siong_proof_rational} that 
for any linear programming problem, say $\max~~ {\bf c}^T{\bf x},$ s.t. ${\bf A}{\bf x}\leq {\bf b}$, ${\bf x}\geq {\bf 0}$,  if all the elements of  ${\bf A}, {\bf b}, {\bf c}$ are rational, and the optimal exists, then there exists an optimizing ${\bf x}$ whose elements are also rational, and so is the optimal value of the objective function. Note that in the linear program in Theorem \ref{thm:LP} all coefficients are indeed rational, in fact the coefficients of  $\lambda_\bullet$ parameters in the constraints and the objective are all either $0,1$ or $2$, and the constants on the RHS of the constraints \eqref{eq:constraintfirst}-\eqref{eq:constraintlast} are conditional-ranks, so they are integers as well, by definition. The feasible region is a rational polytope, so all vertices are rational, and one of the vertices must be optimal for a linear program over a rational polytope. Therefore, there exist non-negative rational values $\lambda_{123}^*,\lambda_{12}^*,\lambda_{13}^*,\lambda_{23}^*,\lambda^*$ that satisfy \eqref{eq:constraintfirst}-\eqref{eq:constraintlast}, for which we we have $f(\lambda_{123}^*,\lambda_{12}^*,\lambda_{13}^*,\lambda_{23}^*,\lambda^*)=F^*$. This gives us the desired bound, $\Delta^*\leq F^*$.
$\hfill\square$

\section{Matching Achievability with Converse: $F^* \leq \max\{\Delta_1, \Delta_2\}$} \label{sec:linear_prog}
The converse proof in Section \ref{sec: Converse} established the lower bound $\Delta^*\geq\max\{\Delta_1,\Delta_2\}$, whereas the achievability proof in Section \ref{sec: Achievability} established the upper bound $\Delta^*\leq F^*$. In this section we show that the bounds are tight. To do so, we will prove that $F^* \leq \max\{\Delta_1, \Delta_2\}$.

Recall that, subject to the constraints \eqref{eq:constraintfirst}-\eqref{eq:constraintlast}, the linear program in Theorem \ref{thm:LP} finds
\begin{align}
	F^*&=\min_{\lambda_{123},\lambda_{12},\lambda_{13},\lambda_{23}, \lambda\in\mathbb{R}_+}  m_1+m_2+m_3\notag \\
	&~~~~~~~~~~~~~~~~~~~~-2\lambda_{123}- \lambda_{12}-\lambda_{13}-\lambda_{23}-\lambda\\
	&=\min_{\lambda_{123},\lambda_{12},\lambda_{13},\lambda_{23}, \lambda\in\mathbb{R}_+}f(\lambda_{123},\lambda_{12},\lambda_{13},\lambda_{23}, \lambda).\label{eq:minF}
\end{align}
We will proceed with the proof  in two steps. First, in Subsection \ref{sec:manipulatedeltastar}, we  manipulate $\max\{\Delta_1,\Delta_2\}$ into an equivalent compact form. Then, in Subsection \ref{sec:solveLP} we show that  in all cases there exist feasible $(\lambda_{123},\lambda_{12},\lambda_{13},\lambda_{23}, \lambda)$ for which $f(\lambda_{123},\lambda_{12},\lambda_{13},\lambda_{23}, \lambda) \leq \max\{\Delta_1, \Delta_2\}$ and therefore by \eqref{eq:minF}, we have $F^* \leq \max\{\Delta_1, \Delta_2\}$.

\subsection{Equivalent Expression for $\Delta_1,\Delta_2$ with Compact Notation}\label{sec:manipulatedeltastar}
To avoid lengthy notation due to the repetitive use of conditional ranks, let us introduce the following compact forms.
\begin{align} \label{eq:def_c}
	 r_{123} &~\triangleq~ \min_{k\in [3]}\rk(\bU_{123}\mid \bV_k'), \notag \\
	 r_{12} &~\triangleq~ \min_{k\in\{1,2\}} \rk(\bU_{12}\mid\bV_k'), \notag \\
	 r_{13} &~\triangleq~ \min_{k\in\{1,3\}}\rk(\bU_{13}\mid\bV_k'), \notag \\
	 r_{23} &~\triangleq~ \min_{k\in\{2,3\}} \rk(\bU_{23}\mid\bV_k'), \notag\\
	 r_{12,13} &~\triangleq~ \rk([\bU_{12},\bU_{13}]\mid \bV_1'),\notag \\
	  r_{12,23} &~\triangleq~ \rk([\bU_{12},\bU_{23}]\mid \bV_2'), \notag \\
	 r_{13,23} &~\triangleq~ \rk([\bU_{13},\bU_{23}]\mid \bV_3'), \notag\\
	 r_{1(2,3)} &~\triangleq~ \rk(\bU_{1(2,3)}\mid \bV_1'),\notag \\
	 r_{2(1,3)} &~\triangleq~ \rk(\bU_{2(1,3)}\mid \bV_2'), \notag \\
	 r_{3(1,2)} &~\triangleq~ \rk(\bU_{3(1,2)}\mid \bV_3'). 
\end{align}
It follows that,
\begin{align} \label{eq:r_con}
	&r_{12,13} \geq \max\{r_{12},r_{13}\},\notag \\
	&r_{12,23} \geq \max\{r_{12},r_{23}\},\notag \\
	&r_{13,23} \geq \max\{r_{13},r_{23}\}.
\end{align}
Note that by these notations, the constraints \eqref{eq:constraintfirst}-\eqref{eq:constraintlast} for $\lambda_{\bullet}$ can be equivalently posed as
\begin{align}
	&\eqref{eq:constraintfirst} \iff && \lambda_{123} \leq r_{123}, \label{eq:constraint_123} \\
	&\eqref{eq:constraintsecond}\iff &&  \lambda_{ij}+\lambda_{123} \leq r_{ij}, && \forall (i,j) \in \{(1,2), (1,3), (2,3)\}\label{eq:constraint_12} \\
	& \eqref{eq:constraintthird}\iff&& \lambda_{ij} +\lambda_{ik}+\lambda_{123}  \leq r_{ij,ik}, && \forall (i,j,k) \in \{(1,2,3), (2,1,3), (3,1,2)\} \label{eq:constraint1213} \\
	&\eqref{eq:constraintlast}\iff && \lambda+\lambda_{ij}+\lambda_{ik}+\lambda_{123} \leq  r_{i(j,k)}, && \forall (i,j,k) \in \{(1,2,3), (2,1,3), (3,1,2)\} \label{eq:constraint_1(23)}
\end{align}
With these notations, we are able to express the $\Delta_1, \Delta_2$ values defined in Theorem \ref{thm:main} in the following equivalent forms.
For $\Delta_1^{ijk}$, we have
\begin{align}
	\Delta_1^{123} &= \rk({\bf V}_1|{\bf V}_1') + \rk({\bf V}_2|{\bf V}_2')+\rk({\bf V}_3|{\bf V}_3') - \rk({\bf U}_{12}|{\bf V}_2')  - \rk({\bf U}_{3(1,2)}|{\bf V}_3') \\
	& = \rk({\bf V}_1) + \rk({\bf V}_2)+\rk({\bf V}_3) - \rk({\bf U}_{12}|{\bf V}_2')  - \rk({\bf U}_{3(1,2)}|{\bf V}_3') \label{eq:Delta_1_full_rank} \\
	& = m_1+m_2+m_3- \rk({\bf U}_{12}|{\bf V}_2')  - r_{3(1,2)}
\end{align}
where \eqref{eq:Delta_1_full_rank} is due to \eqref{eq:full_rank_Uk},
and similarly
\begin{align}
	\Delta_1^{213} &= m_1+m_2+m_3- \rk({\bf U}_{12}|{\bf V}_1')  - r_{3(1,2)}
\end{align}
which implies that,
\begin{align}
	\max\{\Delta_1^{123},\Delta_1^{213}\}& = m_1+m_2+m_3-\min\{\rk({\bf U}_{12}|{\bf V}_1'),\rk({\bf U}_{12}|{\bf V}_2')\}- r_{3(1,2)}\\
	& = m_1+m_2+m_3-r_{12}- r_{3(1,2)}\\
	& \triangleq \delta_3
\end{align}
By taking the pairwise maximum of $\{\Delta_1^{132},\Delta_1^{312}\}$ and $\{\Delta_1^{231},\Delta_1^{321}\}$ respectively, we similarly obtain $\delta_{1}$ and $\delta_{2}$ as follows.
\begin{align}
	&\delta_1 \triangleq \max\{\Delta_1^{132},\Delta_1^{312}\} =  m_1+m_2+m_3-r_{23}-r_{1(2,3)},\\
	&\delta_2 \triangleq \max\{\Delta_1^{231},\Delta_1^{321}\} =  m_1+m_2+m_3-r_{13}-r_{2(1,3)}.
\end{align}
For $\Delta_2^{ijk}$, first note that $\Delta_2^{123} = \Delta_2^{132}$. Thus, we have,
\begin{align}
	&\max\{\Delta_2^{123}, \Delta_2^{132}\} = \Delta_2^{123} \notag \\
	& =  \rk({\bf V}_1\mid \bV_1') + \rk({\bf V}_2\mid \bV_2') + \rk({\bf V}_3\mid \bV_3') \notag\\
&\hspace{1cm}- \frac{1}{2} \Bigg(\min_{\ell\in [3]} \big(\rk({\bf U}_{123}\mid {\bf V}_\ell')\big) +\rk([{\bf U}_{12},{\bf U}_{13}]\mid {\bf V}_1') +\rk({\bf U}_{2(1,3)}\mid {\bf V}_2') + \rk({\bf U}_{3(1,2)}\mid {\bf V}_3') \Bigg) \\
& =  \rk({\bf V}_1) + \rk({\bf V}_2) + \rk({\bf V}_3) \notag\\
&\hspace{1cm}- \frac{1}{2} \Bigg(\min_{\ell\in [3]} \big(\rk({\bf U}_{123}\mid {\bf V}_\ell')\big) +\rk([{\bf U}_{12},{\bf U}_{13}]\mid {\bf V}_1') +\rk({\bf U}_{2(1,3)}\mid {\bf V}_2) + \rk({\bf U}_{3(1,2)}\mid {\bf V}_3') \Bigg) \label{eq:Delta_2_full_rank} \\
	&= m_1+m_2+m_3 - \frac{1}{2}\Big(r_{123}+r_{12,13}+r_{2(1,3)}+r_{3(1,2)} \Big) \\
	&\triangleq \delta_{23}
\end{align}
where \eqref{eq:Delta_2_full_rank} is due to \eqref{eq:full_rank_Uk}.
By taking the pairwise maximum of $\{\Delta_2^{213},\Delta_2^{231}\}$ and $\{\Delta_2^{312},\Delta_2^{321}\}$ respectively, we similarly obtain,
\begin{align}
	&\delta_{13} \triangleq \max\{\Delta_2^{213},\Delta_2^{231}\} = m_1+m_2+m_3-\frac{1}{2}\Big( r_{123}+r_{12,23}+r_{1(2,3)}+r_{3(1,2)} \Big),\\
	&\delta_{12} \triangleq \max\{\Delta_2^{312},\Delta_2^{321}\} = m_1+m_2+m_3-\frac{1}{2}\Big( r_{123}+r_{13,23}+r_{1(2,3)}+r_{2(1,3)} \Big).
\end{align}
Thus, we have, 
\begin{align} \label{eq:converse_f}
	  \max\{\Delta_1, \Delta_2\} = \max\{\delta_{1},\delta_{2},\delta_{3},\delta_{12},\delta_{13},\delta_{23}\}.
\end{align}
The next step is to prove that $F^* \leq \max\{\delta_{1},\delta_{2},\delta_{3},\delta_{12},\delta_{13},\delta_{23}\}$.

\subsection{Proving $F^* \leq \max\{\delta_{1},\delta_{2},\delta_{3},\delta_{12},\delta_{13},\delta_{23}\}$: Constrained Waterfilling}\label{sec:solveLP}

By definition, $F^*\leq f(\lambda_{123},\lambda_{12},\lambda_{13},\lambda_{23}, \lambda)$ for any $(\lambda_{123},\lambda_{12},\lambda_{13},\lambda_{23}, \lambda)$ that satisfies \eqref{eq:constraint_123}-\eqref{eq:constraint_1(23)}. Therefore, it suffices to show that
\begin{align}
	f(r_{123},\lambda_{12},\lambda_{13},\lambda_{23}, \lambda') \leq \max\{\delta_{1},\delta_{2},\delta_{3},\delta_{12},\delta_{13},\delta_{23}\},
\end{align}
where $\lambda' \triangleq \min\{ r_{1(2,3)}-\lambda_{12}-\lambda_{13},r_{2(1,3)}-\lambda_{12}-\lambda_{23},r_{3(1,2)}-\lambda_{13}-\lambda_{23} \} - r_{123}$. In other words, we fix $\lambda_{123}$ to $r_{123}$ and $\lambda$ to $\lambda'$. It can be easily verified that $\lambda_{123} = r_{123}, \lambda = \lambda'$ are in the feasible region specified by \eqref{eq:constraint_123}-\eqref{eq:constraint_1(23)}. As will be shown in the end, fixing $\lambda_{123} = r_{123}, \lambda = \lambda'$ will not hurt the optimality. It is also intuitive because $\lambda_{123}$ corresponds to the amount of transmission that has the highest efficiency (3 birds, 1 stone) so it should be set as large as possible to $r_{123}$. Then, $\lambda = \lambda'$ is also the largest possible we can set after $\lambda_{123}$ is fixed to $r_{123}$.

Setting $\lambda_{123}$ and $\lambda$ to these values (note that both values are non-negative), the objective simplifies to the minimization of,
\begin{align}
	&f = m_1+m_2+m_3 - 2r_{123} - \lambda_{12}-\lambda_{13}-\lambda_{23} \notag \\
	&~~~~ - \min\{ r_{1(2,3)}-\lambda_{12}-\lambda_{13},\notag \\
	&~~~~~~~~~~~~~~r_{2(1,3)}-\lambda_{12}-\lambda_{23},\notag \\
	&~~~~~~~~~~~~~~r_{3(1,2)}-\lambda_{13}-\lambda_{23} \}+r_{123} \\
	& = \underbrace{(m_1+m_2+m_3-r_{123})}_{\mbox{\footnotesize constant}} - \min \{ r_{1(2,3)}+\lambda_{23},~r_{2(1,3)}+\lambda_{13},~r_{3(1,2)}+\lambda_{12} \}.\label{eq:mapthistowf}
\end{align}
We focus on the remaining three parameters, $\lambda_{12},\lambda_{13}$ and $\lambda_{23}$. Note that minimization of $f$ is equivalent to the maximization of the minimum of the three terms: $r_{1(2,3)}+\lambda_{23}$, $r_{2(1,3)}+\lambda_{13}$, and $r_{3(1,2)}+\lambda_{12}$. Intuitively, this optimization may be seen as a constrained waterfilling problem. To make the connection to waterfilling clear, let us further introduce the following notation.
\begin{align}
	&b_1 \triangleq r_{1(2,3)}, &&w_1^{max} \triangleq r_{23}-r_{123}, \\
	&b_2 \triangleq r_{2(1,3)}, &&w_2^{max} \triangleq r_{13}-r_{123}, \\
	&b_3 \triangleq r_{3(1,2)}, &&w_3^{max} \triangleq r_{12}-r_{123}, \\
	&w_{1} \triangleq \lambda_{23}, &&	w_{1,2}^{max} \triangleq r_{13,23} - r_{123}, \\
	&w_{2} \triangleq \lambda_{13}, &&w_{1,3}^{max} \triangleq r_{12,23} - r_{123}, \\
	&w_{3} \triangleq \lambda_{12}, &&w_{2,3}^{max} \triangleq r_{12,13} - r_{123}.
\end{align}
With this notation, the optimization problem becomes
\begin{align}
	& \mbox{maxmize} ~ h_{min}\triangleq \min \{ b_1+w_1,b_2+w_2,b_3+w_3 \},\\
	& ~~~~~~~~~~~~~~~~s.t. \begin{cases}
		w_1 \leq w_1^{max}, \\
		w_2 \leq w_2^{max}, \\
		w_3 \leq w_3^{max}, \\
		w_1+w_2 \leq w_{1,2}^{max}, \\
		w_1+w_3 \leq w_{1,3}^{max}, \\
		w_2+w_3 \leq w_{2,3}^{max}. \\
		w_1,w_2,w_3 \in \mathbb{R}_+ \\
	\end{cases} \label{eq:wf_constraints}
\end{align}
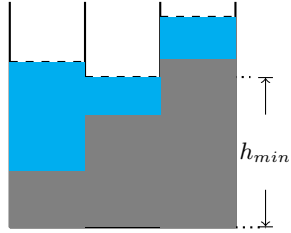
\begin{figure}[htbp]
	\centering
	\begin{tikzpicture}    
    \draw[-, thick](0,0) -- (0,3);
    \draw[-, thick](1,0) -- (1,3);
    \draw[-, thick](2,0) -- (2,3);
    \draw[-, thick](3,0) -- (3,3);
    \draw[-, thick](0,0) -- (3,0);
    \draw[-, thick](0,0.75) -- (1,0.75);
    \draw[-, thick](1,1.5) -- (2,1.5);
    \draw[-, thick](2,2.25) -- (3,2.25);      
    \node[] at (0.5,0.3) {\footnotesize $b_1$};
    \node[] at (1.5,0.7) {\footnotesize $b_2$};
    \node[] at (2.5,1.1) {\footnotesize $b_3$};
    \draw[dashed, thick](0,2.2) -- (1,2.2);
    \draw[dashed, thick](1,2) -- (2,2);
    \draw[dashed, thick](2,2.8) -- (3,2.8);
    \node[] at (0.5,1.5) {\footnotesize $w_1$};
    \node[] at (1.5,1.75) {\footnotesize $w_2$};
    \node[] at (2.5,2.5) {\footnotesize $w_3$};
    \draw[|<-](3.4,0) -- (3.4,0.5);
    \draw[->|](3.4,1.5) -- (3.4,2);
    \node[] at (3.4,1) {\footnotesize $h_{min}$};
    \draw[dotted, thick] (2,2) -- (3.2,2);
    \draw[dotted, thick] (3,0) -- (3.4,0);
    \fill [gray, opacity=0.35] (0,0) rectangle (1,0.75);
    \fill [gray, opacity=0.35] (1,0) rectangle (2,1.5);
    \fill [gray, opacity=0.35] (2,0) rectangle (3,2.25);
    \fill [cyan, opacity=0.35] (0,0.75) rectangle (1,2.2);
    \fill [cyan, opacity=0.35] (1,1.5) rectangle (2,2);
    \fill [cyan, opacity=0.35] (2,2.25) rectangle (3,2.8);
	\end{tikzpicture}
	\caption{Constrained Waterfilling.}\label{fig:wf}
\end{figure}
Let us explain the waterfilling analogy. There are three adjacent vessels as shown in Figure \ref{fig:wf}, labeled $1,2,3$ from left to right. Vessels $1,2,3$ have base levels (shown in gray) at heights $b_1,b_2,b_3$, respectively. We are allowed to add $w_1,w_2,w_3$ amounts of water to Vessel $1$, Vessel $2$ and Vessel $3$, respectively according to the constraints \eqref{eq:wf_constraints}, in order to maximize $h_{min}$, i.e., the minimum of the heights of water in the three vessels. The objective from \eqref{eq:mapthistowf} now maps to the waterfilling problem as,
\begin{align}
f=m_1+m_2+m_3-r_{123}-h_{min}.
\end{align}

Note that the first three constraints in \eqref{eq:wf_constraints} are constraints on the  capacity (for holding water) of individual vessels, and the next three constraints are for pairs of vessels. Furthermore, we have $\max\{w_1^{max},w_2^{max}\}\leq w_{1,2}^{max}$ by \eqref{eq:r_con}, which ensures that the pairwise capacity constraints do not dominate the individual capacity constraints. Since the only constraints are on individual vessel capacities and pairwise vessel capacities, the optimal value of $h_{min}$ must correspond to one of the following outcomes.
\begin{enumerate}
\item $h_{min}$ is limited by the individual capacity of Vessel $i$, $i\in\{1,2,3\}$, which holds  the maximum water it can, $w_i=w_i^{max}$. In this case, $h_{min}=b_i+w_i^{max}$ and $F^*\leq f =m_1+m_2+m_3-r_{123}-h_{min}=\delta_i$.
\item $h_{min}$ is limited by the pairwise capacity of Vessels $i,j$, $(i,j)\in\{(1,2), (1,3), (2,3)\}$,  which together hold the maximum water they can, i.e., $w_i+w_j=w_{i,j}^{max}$, and have the same final water level $h_{min}$. In this case, we have $(h_{min}-b_i)+(h_{min}-b_j)=w_{i,j}^{max}$ which gives us $h_{min}=\frac{b_i+b_j+w_{i,j}^{max}}{2}$ and $F^*\leq f=m_1+m_2+m_3-r_{123}-h_{min}=\delta_{ij}$.
\end{enumerate}
Thus, in every case we have $F^*\leq \max\{\delta_{1},\delta_{2},\delta_{3},\delta_{12},\delta_{13},\delta_{23}\}= \max\{\Delta_1,\Delta_2\}$, which completes the proof.
$\hfill\square$

\section{Conclusion} \label{sec:conc}
The exact capacity of the $3$ user LCBC is found for all cases, i.e., for arbitrary finite field $\mathbb{F}_q$, arbitrary data dimension $d$, and arbitrary specifications of the users' desired computations and side-information $\bV_k', \bV_k$. The $3$ user setting introduces several  intricacies that were not encountered in the $2$ user LCBC, such as the insufficiency of the entropic formulation for tight converse bounds,  the need for functional submodularity,  the rich variety of subspaces involved, random coding arguments to resolve discrepancies between the users' differing views of the same subspaces, the tradeoffs between the communication efficiencies associated with these subspaces, and the inherent optimization that led us to a constrained waterfilling solution. The fact that the $3$ user LCBC capacity turns out to be fully tractable despite  these intricacies is surprising.  In particular, we note that even though the $3$ user LCBC involves at least $6$ key subspaces in $\bV_k, \bV_k'$, $k\in\{1,2,3\}$, the solution did not require the Ingleton inequality, nor were non-Shannon inequalities required for the converse. Instead, the main tools used were Steinitz Exchange lemma, the dimension counting of pairwise unions and intersections of subspaces, functional submodularity, and the random coding argument invoked through the Schwartz-Zippel lemma. The tractability of the $3$ user LCBC is indicative of the potential for further progress in understanding the fundamental limits of basic computation networks in future efforts. Indeed, there are many promising directions for future work. Building on the $K=2$ and $K=3$ cases, the $K=4$ user LCBC in particular is an important next step, as it might either reveal a consistent pattern that holds for arbitrary $K$ users or present obstacles that are indicative of the difficulty of the large $K$ setting. Also of interest are asymptotic LCBC settings with large number of users. An intriguing generalization of the LCBC problem is the LCBC with partially informed server, LCBC-PIS in short, where the central server has only limited knowledge of the data in the form of some linear functions of the data. The LCBC-PIS setting has been introduced and solved recently for $K=2$ users in \cite{LCBC_PIS}. The capacity remains open for $K\geq 3$ users. Studies of linear computation multiple access settings (LCMAC) represent another promising research avenue, partially explored in \cite{Fei_Chen_Wang_Jafar} from a coding perspective. Approximate linear computations over real or complex numbers, as well as non-linear computations that connect to coded distributed computing represent other challenging and important research directions for future work. From a practical perspective, studies of computational and communication tradeoffs of AR/VR applications that take advantage of the coding schemes discovered through the studies of LCBC/LCMAC settings would be valuable complements to the theoretical efforts.

\appendix

\section{Field Extension} \label{app:field_extension}
To clarify the notation and illustrate the utility of field extensions, let us present an example. Consider $q = 2, K=3, d=2, m=m'=1$ and the following coefficient matrices ${\bf U}_k = [{\bf V}_k',{\bf V}_k], k\in[3]$ as
{\small
\begin{align} \label{eq:app_field_1}
	{\bf U}_1^{2\times 2} = \begin{bmatrix}
		1 & 0 \\ 0 & 1
	\end{bmatrix}, ~
	{\bf U}_2^{2\times 2} = \begin{bmatrix}
		0 & 1 \\ 1 & 1
	\end{bmatrix}, ~ 
	{\bf U}_3^{2\times 2} = \begin{bmatrix}
		1 & 1 \\ 1 & 0
	\end{bmatrix}.
\end{align}
}By the problem formulation, $\mathbf{x} \in \mathbb{F}_2^{2\times 1}$ denotes the data for each computation, and ${\bf X} \in \mathbb{F}_2^{2\times L}$ denotes the data for $L$ computations. Let us first try to design a coding scheme with $L=1$. Denote ${\bf X}^T = [x_1(1),x_2(1)]$ and then ${\bf W}_1' = {\bf X}^T{\bf V}_1' = x_1(1)$, ${\bf W}_1 = {\bf X}^T{\bf V}_1 = x_2(1)$, ${\bf W}_2' = {\bf X}^T{\bf V}_2' = x_2(1)$, ${\bf W}_2 = {\bf X}^T{\bf V}_2 = x_1(1)+x_2(1)$, ${\bf W}_3' = {\bf X}^T{\bf V}_3' = x_1(1)+x_2(1)$, ${\bf W}_3 = {\bf X}^T{\bf V}_3 = x_1(1)$. The following table shows all possible outcomes of $({\bf W}_k',{\bf W}_k)_{k\in [3]}$.
{\small
\begin{center}
\begin{tabular}{c|c||c|c|c|c|c|c}
\toprule
	$x_1(1)$ & $x_2(1)$ & ${\bf W}_1'$ & ${\bf W}_1$ & ${\bf W}_2'$ & ${\bf W}_2$ & ${\bf W}_3'$ & ${\bf W}_3$  \\ \hline
 	0 & 0 & 0 & 0 & 0 & 0 & 0 & 0   \\\hline
 	0 & 1 & 0 & 1 & 1 & 1 & 1 & 0   \\\hline
 	1 & 0 & 1 & 0 & 0 & 1 & 1 & 1   \\\hline
 	1 & 1 & 1 & 1 & 1 & 0 & 0 & 1   \\\bottomrule
\end{tabular}
\end{center}
}A coding scheme must satisfy the property that for any two outcomes, the broadcast information ${\bf S}$ corresponding to these outcomes has to be different if $\exists k\in [3]$ such that ${\bf W}_k'$ is the same but ${\bf W}_k$ is different for these two outcomes. This is necessary to ensure that User $k$ will not be confused when decoding under these two outcomes. Following this rule it is easy to verify from the table that the realization of ${\bf S}$ has to be different for any two outcomes in this example, which implies ${\bf S}$ has to be different in all outcomes. Thus, $|\mathcal{S}| \geq 4$, which implies $N\geq 2$ and thus $\Delta=N/L \geq 2$ for $L=1$. In other words, scalar coding schemes cannot achieve $\Delta<2$.

Let us now consider field extension. Let $z=2$ and consider $L=z=2$. Denote $\bar{{\bf V}}_k' = {\bf V}_k' \otimes {\bf I}^{2\times 2}$, $\bar{{\bf V}}_k = {\bf V}_k \otimes {\bf I}^{2\times 2}$ and $\bar{{\bf U}}_k = {\bf U}_k \otimes {\bf I}^{2\times 2}$ as the {\it $2$-extension} of the coefficient matrices, where $\otimes$ denotes the Kronecker product. We have
{\small
\begin{align}
	&\bar{{\bf U}}_1^{4\times 4} = [\bar{{\bf V}}_1',\bar{{\bf V}}_1]=\begin{bmatrix}
		1 & 0 & 0 & 0 \\
		0 & 1 & 0 & 0 \\
		0 & 0 & 1 & 0 \\
		0 & 0 & 0 & 1
	\end{bmatrix}, \notag \\
	&\bar{{\bf U}}_2^{4\times 4} = [\bar{{\bf V}}_2',\bar{{\bf V}}_2]=\begin{bmatrix}
		0 & 0 & 1 & 0 \\
		0 & 0 & 0 & 1 \\
		1 & 0 & 1 & 0 \\
		0 & 1 & 0 & 1
	\end{bmatrix}, \notag \\
	&\bar{{\bf U}}_3^{4\times 4} = [\bar{{\bf V}}_3',\bar{{\bf V}}_3]=\begin{bmatrix}
		1 & 0 & 1 & 0 \\
		0 & 1 & 0 & 1 \\
		1 & 0 & 0 & 0 \\
		0 & 1 & 0 & 0
	\end{bmatrix}.
\end{align}
}Then denote ${\bf X} = [x_{1}(1),x_{1}(2);x_{2}(1),x_{2}(2)]$ as the data matrix for $L=2$, and denote $\bar{{\bf X}} = \mbox{vec}({\bf X}^T)$, where $\mbox{vec}(\cdot)$ is the vectorization function. We have $\bar{{\bf X}}^T=[x_{1}(1),x_{1}(2),x_{2}(1),x_{2}(2)]\in \mathbb{F}_2^{1\times 4}$. We can see that User $k$ has side-information $\bar{{\bf X}}^T\bar{{\bf V}}_k'$, and wants to compute $\bar{{\bf X}}^T\bar{{\bf V}}_k$. Then by the property of finite field extensions, we can regard $[x_{1}(1),x_{1}(2)]$ as $\bar{x}_1 \in \mathbb{F}_4$ and similarly $[x_{2}(1),x_{2}(2)]$ as $\bar{x}_2 \in \mathbb{F}_4$. Accordingly, the extended coefficient matrices are regarded as $2\times 2$ matrices in ${\Bbb F}_4$ as
{\small
\begin{align}
	{\bf U}_1^{2\times 2} = \begin{bmatrix}
		1 & 0 \\ 0 & 1
	\end{bmatrix}, ~
	{\bf U}_2^{2\times 2} = \begin{bmatrix}
		0 & 1 \\ 1 & 1
	\end{bmatrix}, ~ 
	{\bf U}_3^{2\times 2} = \begin{bmatrix}
		1 & 1 \\ 1 & 0
	\end{bmatrix}.
\end{align}
}Note that the matrices are exactly the same as the matrices in (\ref{eq:app_field_1}) but considered in the extended field $\mathbb{F}_4$. To avoid complex notations, we redefine the data matrix as ${\bf X}=[\bar{x}_1;\bar{x}_2] \in \mathbb{F}_4^{2\times 1}$. Thus, by considering $L=2$ computations, we have an equivalent problem where $q=4,d=2,m=m'=1$ and the same coefficient matrices ${\bf U}_k = [{\bf V}_k',{\bf V}_k], k\in[3]$, but now all elements are from ${\Bbb F}_4$. As a coding scheme with $L=2$, it suffices to send ${\bf S} = {\bf X}^T[1;\alpha]$, where $\alpha \not \in \{0,1\}$ and $\alpha \in \mathbb{F}_4$. Since the column vector $[1;\alpha]$ is linearly independent of each of ${\bf V}_1', {\bf V}_2', {\bf V}_3'$, each user has two independent equations in ${\bf X}^T$ from which it can decode all of ${\bf X}$, and recover the desired ${\bf W}_k$. Since ${\bf S}$ is chosen as $1$ symbol from ${\Bbb F}_4$, we have $N=2$ ($1$ symbol in ${\Bbb F}_4$ corresponds to $2$ symbols in ${\Bbb F}_2$) and thus $\Delta = N/L = 1$, thus a better $\Delta = N/L$ is achieved by considering $L>1$.

In general, by considering $L=z$ computations, the original problem is equivalent to the problem with all the same parameters including the coefficient matrices but in the extended field ${\Bbb F}_{q^z}$. By considering  $z$ computations in the original problem as $1$ computation in the extended field, the original problem over $\mathbb{F}_q$ for $L=zL'$ computations,  is equivalent to the new problem with the same parameters in the extended field ${\Bbb F}_{q^z}$ for $L'$ computations. 

\section{Some Discussion on Lemma \ref{lemma:decomposition}} \label{app:3space_discuss}
As we go to $3$ spaces, $\langle{\bf U}_1\rangle, \langle{\bf U}_2\rangle, \langle{\bf U}_3\rangle$, generalizing the decomposition for $\langle{\bf U}_1\rangle$ and $\langle{\bf U}_2\rangle$ as in Lemma \ref{lem:2_decomp} is not so straightforward. Analogies to set-theoretic ideas such as  inclusion-exclusion principle and Venn's diagrams  do not  quite work for $3$ vector spaces. For example, if $\langle{\bf U}_1\rangle, \langle{\bf U}_2\rangle, \langle{\bf U}_3\rangle$ are three independent lines in a plane, i.e., pairwise independent one-dimensional subspaces of a $2$ dimensional vector space, then $\langle {\bf U}_1\rangle$ has no non-trivial intersection with either of $\langle {\bf U}_2\rangle$ or $\langle {\bf U}_3\rangle$ individually, yet $\langle {\bf U}_1\rangle$ is contained  in $\langle[{\bf U}_2,{\bf U}_3]\rangle$,  a situation for which there is no direct set-theoretic analogy. 
This is why we need the subspace decomposition for $\langle{\bf U}_1\rangle, \langle{\bf U}_2\rangle, \langle{\bf U}_3\rangle$, as illustrated in Figure \ref{fig:venn1} and formalized in Lemma \ref{lemma:decomposition}.
As noted, the decomposition parallels a corresponding decomposition in the  DoF studies of the $3$ user MIMO BC by Wang in \cite{wang2017degrees}, highlighting its fundamental conceptual significance. 
	
Following the idea of \emph{growing} the basis to cover larger and larger subspaces, similar to the constructive proof for Lemma \ref{lem:2_decomp}, let us interpret Figure \ref{fig:venn1}, so that Lemma \ref{lemma:decomposition} will be intuitively transparent. Consider the space $\langle{\bf U}_1\rangle$, i.e., the column space of ${\bf U}_1$. This space is decomposed into $5$ subspaces as follows. First we have the space within $\langle{\bf U}_1\rangle$ which overlaps with \emph{both} $\langle{\bf U}_2\rangle$ and $\langle{\bf U}_3\rangle$. This is the space $\langle{\bf U}_{123}\rangle\triangleq \langle{\bf U}_1\rangle\cap\langle{\bf U}_2\rangle\cap\langle{\bf U}_3\rangle$. The basis for this space is labeled in the figure as the matrix ${\bf B}_{123}$. Now consider the space within $\langle{\bf U}_1\rangle$ which overlaps with $\langle{\bf U}_2\rangle$. This is the space $\langle{\bf U}_{12}\rangle\triangleq\langle{\bf U}_1\rangle\cap\langle{\bf U}_2\rangle$. The basis for this space is $[{\bf B}_{123}, {\bf B}_{12}]$. Note that $\langle{\bf U}_{123}\rangle \subset\langle{\bf U}_{12}\rangle$, which is also reflected in the fact that the basis for $\langle{\bf U}_{12}\rangle$ explicitly contains the basis for $\langle{\bf U}_{123}\rangle$. It is  important to recall that \emph{the columns of a basis matrix must be linearly independent by definition}. Therefore, not only do we have a basis $[{\bf B}_{123}, {\bf B}_{12}]$ for $\langle{\bf U}_{12}\rangle$, but also by the linear independence of the basis vectors, it follows that $\langle{\bf U}_{12}\rangle$ is decomposed into two \emph{independent} subspaces, namely the subspaces $\langle{\bf B}_{123}\rangle$ and $\langle{\bf B}_{12}\rangle$. This can also be expressed as\footnote{For vector spaces $\mathcal{V}, \mathcal{V}_1,\cdots, \mathcal{V}_K$, we have $\mathcal{V}=\mathcal{V}_1\oplus \mathcal{V}_2\oplus\cdots\oplus \mathcal{V}_K$ iff for every $v\in \mathcal{V}$, there exist \emph{unique} $v_k\in \mathcal{V}_k$ for all $k\in[K]$ such that $v=v_1+v_2+\cdots+v_K$.} a direct sum, i.e., $\langle{\bf U}_{12}\rangle=\langle{\bf B}_{123}\rangle\oplus \langle{\bf B}_{12}\rangle$. Similarly, $\langle{\bf U}_{13}\rangle$, i.e., the intersection of $\langle{\bf U}_{1}\rangle$ and $\langle{\bf U}_3\rangle$ is decomposed into \emph{independent} subspaces $\langle{\bf B}_{123}\rangle$ and $\langle{\bf B}_{23}\rangle$, i.e., $\langle{\bf U}_{13}\rangle=\langle{\bf B}_{123}\rangle\oplus \langle{\bf B}_{13}\rangle$.  Continuing the process further, now consider the space within $\langle{\bf U}_1\rangle$ which overlaps with $\langle[{\bf U}_2,{\bf U}_3]\rangle$, i.e., the space denoted as $\langle{\bf U}_{1(2,3)}\rangle$. As indicated in the figure, the basis for this space is $[{\bf B}_{123}, {\bf B}_{12}, {\bf B}_{13}, {\bf B}_{1(2,3)}]$, which immediately decomposes $\langle{\bf U}_{1(2,3)}\rangle$ into $4$ \emph{independent} subspaces, i.e., $\langle{\bf U}_{1(2,3)}\rangle=\langle{\bf B}_{123}\rangle\oplus\langle{\bf B}_{12}\rangle\oplus\langle{\bf B}_{23}\rangle\oplus\langle{\bf B}_{1(2,3)}\rangle$. Finally, consider all of $\langle{\bf U}_1\rangle$, for which Figure \ref{fig:venn1} identifies the basis as the matrix $[{\bf B}_{123}, {\bf B}_{12}, {\bf B}_{13}, {\bf B}_{1(2,3)}, {\bf B}_{1c}]$, thus completing the decomposition of $\langle{\bf U}_1\rangle$ into $5$ disjoint subspaces, $\langle{\bf U}_{1}\rangle=\langle{\bf B}_{123}\rangle\oplus\langle{\bf B}_{12}\rangle\oplus\langle{\bf B}_{23}\rangle\oplus\langle{\bf B}_{1(2,3)}\rangle\oplus\langle{\bf B}_{1c}\rangle$. Similar decompositions apply to $\langle{\bf U}_2\rangle$ and $\langle{\bf U}_3\rangle$ as well.

The description thus far is similar to set-theoretic decompositions into disjoint sets, as one might represent through disjoint regions in a Venn's diagram. This brings us to the most interesting aspect of the $3$-subspace decomposition, highlighted as the yellow regions with dashed boundaries in Figure \ref{fig:venn1}. The  subspaces corresponding to these three regions, namely $\langle {\bf B}_{1(2,3)}\rangle$, $\langle {\bf B}_{2(1,3)}\rangle$, and $\langle {\bf B}_{3(1,2)}\rangle$ are only \emph{pairwise} independent, and the span of the union of any two of them contains the third. In fact, it is always possible to choose the basis matrices  such that ${\bf B}_{1(2,3)}+{\bf B}_{2(1,3)} = {\bf B}_{3(1,2)}$, which will simplify the construction of the coding scheme. Thus,  Figure \ref{fig:venn1} shows $10$ subspaces, including the $3$ subspaces highlighted in yellow, and if we exclude any one of the $3$ yellow subspaces, the remaining $9$ are independent spaces. Mathematically,
\begin{align}
&\langle[{\bf U}_1, {\bf U}_2, {\bf U}_3]\rangle\notag\\
&=\langle{\bf B}_{1(2,3)}\rangle\oplus \langle{\bf B}_{2(1,3)}\rangle\oplus\langle{\bf B}_{1c}\rangle\oplus\langle{\bf B}_{2c}\rangle\oplus\langle{\bf B}_{3c}\rangle\oplus\langle{\bf B}_{12}\rangle\oplus\langle{\bf B}_{23}\rangle\oplus\langle{\bf B}_{13}\rangle\oplus\langle{\bf B}_{123}\rangle\\
&=\langle{\bf B}_{2(1,3)}\rangle\oplus\langle {\bf B}_{3(1,2)}\rangle\oplus\langle{\bf B}_{1c}\rangle\oplus\langle{\bf B}_{2c}\rangle\oplus\langle{\bf B}_{3c}\rangle\oplus\langle{\bf B}_{12}\rangle\oplus\langle{\bf B}_{23}\rangle\oplus\langle{\bf B}_{13}\rangle\oplus\langle{\bf B}_{123}\rangle\\
&=\langle{\bf B}_{3(1,2)}\rangle\oplus\langle {\bf B}_{1(2,3)}\rangle\oplus\langle{\bf B}_{1c}\rangle\oplus\langle{\bf B}_{2c}\rangle\oplus\langle{\bf B}_{3c}\rangle\oplus\langle{\bf B}_{12}\rangle\oplus\langle{\bf B}_{23}\rangle\oplus\langle{\bf B}_{13}\rangle\oplus\langle{\bf B}_{123}\rangle
\end{align}
 
\section{Proof of Lemma \ref{lemma:decomposition}: Decomposition of $\langle {\bf U}_1\rangle, \langle {\bf U}_2\rangle, \langle {\bf U}_3\rangle$} \label{proof:decomposition}

Let us begin by informally summarizing the key facts that are used extensively  in this section. 
\begin{enumerate}
\item A matrix $M$ forms a \emph{basis} of the column-space of a matrix $U$, if and only if $\langle U\rangle \subset \langle M\rangle$ and the number of columns of $M$ is equal to $\rk(U)$. Note that a basis matrix must have full column-rank, i.e., all its columns are linearly independent, and it has only as many columns as needed to span $\langle U\rangle$, i.e., $\rk(U)$ columns.
\item If $A\in \mathbb{F}_q^{d\times a}$ and $B\in \mathbb{F}_q^{d\times b}$ are basis matrices (i.e., they each have full column-rank) and $\langle B \rangle \subset \langle A \rangle$, then there exists a matrix $C \in \mathbb{F}_q^{d\times (a-b)}$ such that $[B,C]$ is a basis of $\langle A \rangle$. It follows that $\langle C \rangle \subset \langle A \rangle$. Let us call $C$ the complement of $B$ in $A$ and denote it as $C = A\backslash B$. Note that such $C$ is not unique, and one feasible choice of such $C$ follows from the Steinitz Exchange Lemma, which produces a $C$ that is a submatrix of $A$. Other feasible choices of $C$ can be constructed as follows. Denote $C^{\footnotesize{\mbox{\footnotesize old}}}$ as the choice from the Steinitz Exchange Lemma. Other feasible choices can be constructed as $C^{\footnotesize{\mbox{\footnotesize new}}} = C^{\footnotesize{\mbox{\footnotesize old}}}R+BR'$, where $R\in\mathbb{F}_q^{(a-b)\times (a-b)}$, $R'\in\mathbb{F}_q^{b\times (a-b)}$ and $R$ is invertible. To see this, first note that $C^{\footnotesize{\mbox{\footnotesize new}}}$ has the same size as $C^{\footnotesize{\mbox{\footnotesize old}}}$. Then note that any $v\in \langle A \rangle$ can be represented as $v = Br_b+C^{\footnotesize{\mbox{\footnotesize old}}}r_c$ because $[B, C^{\footnotesize{\mbox{\footnotesize old}}}]$ is a basis of $\langle A \rangle$. It follows that $v = B(r_b-R'R^{-1}r_c)+C^{\footnotesize{\mbox{\footnotesize new}}}R^{-1}r_c$, which implies that $v\in \langle [B, C^{\footnotesize{\mbox{\footnotesize new}}}] \rangle$. 
\item Also recall the dimension formula  \eqref{eq:dimension_formula}, i.e., $\rk(M_1)+\rk(M_2)=\rk(M_1\cap M_2)+\rk([M_1,M_2])$.
\end{enumerate}
We will now construct the  $10$ bases that are mentioned in Lemma \ref{lemma:decomposition}, collectively referred to as $\mathcal{B}$.
\begin{align}
	&\mathcal{B} = \{ {\bf B}_{123},{\bf B}_{12},{\bf B}_{13},{\bf B}_{23},{\bf B}_{1(2,3)},{\bf B}_{2(1,3)},{\bf B}_{3(1,2)}, {\bf B}_{1c},{\bf B}_{2c},{\bf B}_{3c}\} \notag
\end{align}
First, let us define the compact notation, $b_* \triangleq \rk({\bf B}_*)$ where
\begin{align*}
	* \in \{1,2,3,12,13,23,1(2,3),2(1,3),3(1,2),1c,2c,3c\}.
\end{align*}
For example, $b_{1(2,3)}\triangleq \rk({\bf B}_{1(2,3)})$. Since the vector space is $d$-dimensional, and ${\bf B}_*$ are basis matrices, it follows that  the size of ${\bf B}_*$ is $d\times b_*$. The construction now proceeds as follows.

\begin{enumerate}[wide, labelindent=1em, labelwidth=!, leftmargin =3em,  label={Step \arabic*}: ]
	\item ${\bf B}_{123} = {\bf U}_{123}$.
	\item ${\bf B}_{12} = {\bf U}_{12} \backslash {\bf B}_{123}$.
	\item ${\bf B}_{13} = {\bf U}_{13} \backslash {\bf B}_{123}$.
	\item ${\bf B}_{23} = {\bf U}_{23} \backslash {\bf B}_{123}$.
\end{enumerate}
These four steps are direct applications of the Steinitz Exchange Lemma, which also guarantees that properties {\it (P1)} - {\it (P4)} are satisfied.  Next let us prove that {\it (P5)} - {\it (P7)} are also satisfied. Consider $(P5)$. It follows from the construction that $[{\bf B}_{123},{\bf B}_{12},{\bf B}_{13}]$ spans $\langle [{\bf U}_{12},{\bf U}_{13}] \rangle$ because it explicitly contains the bases for both spaces,  but we wish to show that it is itself a \emph{basis}, i.e., it has full column-rank. Now, since $[{\bf B}_{123},{\bf B}_{12},{\bf B}_{13}]$ has $b_{123}+b_{12}+b_{13}$ columns and $\rk([{\bf U}_{12},{\bf U}_{13}]) = \rk({\bf U}_{12})+\rk({\bf U}_{13})-\rk({\bf U}_{123}) = (b_{123}+b_{12})+(b_{123}+b_{13})-b_{123} = b_{123}+b_{12}+b_{13}$, it follows that $[{\bf B}_{123},{\bf B}_{12},{\bf B}_{13}]$ has full column rank. Thus, {\it (P5)} is satisfied. {\it (P6)} and {\it (P7)} are similarly proved by symmetry. We continue the construction of $\mathcal{B}$.
\begin{enumerate}[wide, labelindent=1em, labelwidth=!, leftmargin =3em,  label={Step \arabic*}: ]
\setcounter{enumi}{4}
	\item ${\bf B}_{1(2,3)} = {\bf U}_{1(2,3)} \backslash [{\bf B}_{123},{\bf B}_{12},{\bf B}_{13}]$.
	\item ${\bf B}_{2(1,3)} = {\bf U}_{2(1,3)} \backslash [{\bf B}_{123},{\bf B}_{12},{\bf B}_{23}]$.
	\item ${\bf B}_{3(1,2)} = {\bf U}_{3(1,2)} \backslash [{\bf B}_{123},{\bf B}_{13},{\bf B}_{23}]$.
	\item ${\bf B}_{1c} = {\bf U}_1 \backslash [{\bf B}_{123},{\bf B}_{12},{\bf B}_{13},{\bf B}_{1(2,3)}]$.
	\item ${\bf B}_{2c} = {\bf U}_2 \backslash [{\bf B}_{123},{\bf B}_{12},{\bf B}_{23},{\bf B}_{2(1,3)}]$.
	\item ${\bf B}_{3c} = {\bf U}_3 \backslash [{\bf B}_{123},{\bf B}_{13},{\bf B}_{23},{\bf B}_{3(1,2)}]$.
\end{enumerate}
Again, Steps 5-10 are applications of the Steinitz Exchange Lemma, which implies that $[{\bf B}_{123},{\bf B}_{12}$, ${\bf B}_{13},{\bf B}_{1(2,3)},$ ${\bf B}_{1c}]$ is a basis of $\langle {\bf U}_1 \rangle$, $[{\bf B}_{123},{\bf B}_{12},{\bf B}_{23}$, ${\bf B}_{2(1,3)},{\bf B}_{2c}]$ is a basis of $\langle {\bf U}_2 \rangle$, and $[{\bf B}_{123},{\bf B}_{13}$, ${\bf B}_{23},{\bf B}_{3(1,2)},{\bf B}_{3c}]$ is a basis of $\langle {\bf U}_3 \rangle$. Furthermore, $[{\bf B}_{123},{\bf B}_{12},{\bf B}_{13},{\bf B}_{23},{\bf B}_{2(1,3)},{\bf B}_{3(1,2)},{\bf B}_{2c},{\bf B}_{3c}]$ is a basis of $\langle[{\bf U}_2, {\bf U}_3]\rangle$ because it spans $\langle [{\bf U}_2, {\bf U}_3] \rangle$ by construction, and has full column-rank because
\begin{align}
&\rk([{\bf U}_2,{\bf U}_3]) \\
&= \rk({\bf U}_2)+\rk({\bf U}_3) - \rk({\bf U}_{23})\\
&=(b_{123}+b_{12}+b_{23}+b_{2(1,3)}+b_{2c}) +(b_{123}+b_{13}+b_{23}+b_{3(1,2)}+b_{3c}) -(b_{123}+b_{23})\\
& = b_{123}+b_{12}+b_{13}+b_{23}+b_{2(1,3)}+b_{3(1,2)}+b_{2c}+b_{3c}
\end{align}
 which happens to be the number of columns of $[{\bf B}_{123},{\bf B}_{12},{\bf B}_{13},{\bf B}_{23},{\bf B}_{2(1,3)},{\bf B}_{3(1,2)},{\bf B}_{2c},{\bf B}_{3c}]$. Thus,  {\it (P14)} - {\it (P16)} are satisfied. 

Next let us show that {\it (P17)} - {\it (P19)} are satisfied. Consider {\it (P17)}, i.e., we wish to show that ${\bf B}_{17}\triangleq [{\bf B}_{123},{\bf B}_{12},{\bf B}_{13},{\bf B}_{23}$, ${\bf B}_{1(2,3)},{\bf B}_{2(1,3)},{\bf B}_{1c},{\bf B}_{2c},{\bf B}_{3c}]$ is a basis for $\langle [{\bf U}_1, {\bf U}_2,{\bf U}_3]\rangle$. First let us show that $\langle [{\bf U}_1, {\bf U}_2,{\bf U}_3]\rangle$ is contained in the span of ${\bf B}_{17}$. From {\it (P11), (P12)} note that the basis for $\langle{\bf U}_1\rangle$ is explicitly contained in ${\bf B}_{17}$, and so is the basis for $\langle{\bf U}_2\rangle$. Then, noting that $\langle {\bf B}_{3(1,2)} \rangle \subset \langle [{\bf U}_1,{\bf U}_2] \rangle$ by its construction in Step 7, it follows from {\it (P13)} that $\langle {\bf U}_3\rangle$ is also contained in the span of ${\bf B}_{17}$. Thus, $\langle [{\bf U}_1, {\bf U}_2,{\bf U}_3]\rangle$ is contained in the column-span of ${\bf B}_{17}$. Next let us show that ${\bf B}_{17}$ has only as many columns as $\rk([{\bf U}_1,{\bf U}_2, {\bf U}_3])$, so it must be a basis.
\begin{align}
&\rk([{\bf U}_1,{\bf U}_2,{\bf U}_3])\\
&= \rk([{\bf U}_1,{\bf U}_2]) + \rk({\bf U}_3) - \rk({\bf U}_{3(1,2)})\\
&=(b_{123}+b_{12}+b_{13}+b_{23}+b_{1(2,3)}+b_{2(1,3)}+b_{1c}+b_{2c}) +(b_{123}+b_{13}+b_{23}+b_{3(1,2)}+b_{3c}) \notag \\
&~~~~-(b_{123}+b_{13}+b_{23}+b_{3(1,2)})\\
&=b_{123}+b_{12}+b_{13}+b_{23} +b_{1(2,3)}+b_{2(1,3)} +b_{1c}+b_{2c}+b_{3c}
\end{align}
which is the number of columns of ${\bf B}_{17}$. 
Thus, {\it (P17)} is satisfied, and by symmetry {\it (P18)} and {\it (P19)} are satisfied as well.

At this point, only {\it (P20)} remains to be shown. It is worthwhile to note that we always have, 
 \begin{align}\label{eq:3b_equal}
 	b_{1(2,3)} = b_{2(1,3)} = b_{3(1,2)}.
 \end{align}
This is because properties {\it (P17)}-{\it (P19)} together imply that $b_{1(2,3)}+b_{2(1,3)} = b_{1(2,3)}+b_{3(1,2)} = b_{2(1,3)}+b_{3(1,2)}$ and thus the three components must be equal. If $b_{1(2,3)} = b_{2(1,3)} = b_{3(1,2)}=0$, then {\it (P20)} can be neglected. Otherwise, we continue the process from Step 11.
\begin{enumerate}[wide, labelindent=1em, labelwidth=!, leftmargin =3em,  label={Step \arabic*}: ]
\setcounter{enumi}{10}
	\item Since $\langle {\bf B}_{3(1,2)} \rangle \subset \langle [{\bf U}_1,{\bf U}_2] \rangle$ as noted above, let us uniquely represent ${\bf B}_{3(1,2)}$ in the basis of $\langle [{\bf U}_1,{\bf U}_2] \rangle$ according to {\it (P14)} as,
\end{enumerate}
	\begin{align}
		&{\bf B}_{3(1,2)} = {\bf B}_{123}{\bf R}_1 + {\bf B}_{12}{\bf R}_2 + {\bf B}_{13}{\bf R}_3 + {\bf B}_{23}{\bf R}_4 + {\bf B}_{1(2,3)}{\bf R}_5+{\bf B}_{1c}{\bf R}_6+{\bf B}_{2(1,3)}{\bf R}_7+{\bf B}_{2c}{\bf R}_8
	\end{align}
where ${\bf R}_1$ to ${\bf R}_8$ are ${\Bbb F}_q$ matrices with appropriate sizes. In particular, from \eqref{eq:3b_equal} it follows  that ${\bf R}_5$ and ${\bf R}_7$ are square matrices. A key goal in the remainder of the proof will be to show that ${\bf R}_5$ and ${\bf R}_7$ are invertible.

First we claim that ${\bf R}_6$ and ${\bf R}_8$ must be zero matrices. We prove this by contradiction. Suppose ${\bf R}_6$ is not the zero matrix,  say its first column is a non-zero vector ${\bf r}$, then ${\bf B}_{1c}{\bf r} \not={\bf 0}$ will lie in $\langle [{\bf U}_2, {\bf U}_3]\rangle$. However, by construction, $\rk({\bf B}_{1c}\cap[{\bf U}_2, {\bf U}_3]) = \rk({\bf B}_{1c}\cap{\bf U}_1\cap[{\bf U}_2, {\bf U}_3]) = \rk({\bf B}_{1c}\cap{\bf U}_{1(2,3)}) = 0$, meaning that $\langle {\bf B}_{1c}\rangle$ and $\langle [{\bf U}_2,{\bf U}_3] \rangle$ are independent spaces. This completes the proof by contradiction, confirming that ${\bf R}_6$ is a zero matrix. Similar argument is true for ${\bf R}_8$ due to symmetry. Thus, we have
	\begin{align}
		&{\bf B}_{3(1,2)} = {\bf B}_{123}{\bf R}_1 + {\bf B}_{12}{\bf R}_2 + {\bf B}_{13}{\bf R}_3 + {\bf B}_{23}{\bf R}_4 + {\bf B}_{1(2,3)}{\bf R}_5+{\bf B}_{2(1,3)}{\bf R}_7.\label{eq:check}
	\end{align}
We now claim that ${\bf R}_5$ and ${\bf R}_7$ have full column rank, i.e., they are invertible square matrices. The proof is by contradiction as well. Suppose ${\bf R}_5$ does not have full column rank, then there exists a non-zero vector `${\bf a}$' such that ${\bf R}_5{\bf a} = {\bf 0}$, which then implies that 
\begin{align} \label{eq:R5R7}
	&\underbrace{({\bf B}_{3(1,2)}-{\bf B}_{123}{\bf R}_1 - {\bf B}_{13}{\bf R}_3 -{\bf B}_{23}{\bf R}_4)}_{\in \langle {\bf U}_3 \rangle} {\bf a}  = \underbrace{({\bf B}_{12}{\bf R}_2 +{\bf B}_{2(1,3)}{\bf R}_7)}_{\in \langle {\bf U}_2 \rangle}{\bf a} \triangleq {\bf b}.
\end{align}
Since ${\bf B}_{3(1,2)},{\bf B}_{123},{\bf B}_{13},{\bf B}_{23}$ are disjoint submatrices of the basis matrix for $\langle{\bf U}_{3(1,2)}\rangle$ according to {\it (P10)}, they are linearly independent by construction. It follows that,
\begin{enumerate}
\item  $({\bf B}_{3(1,2)}-{\bf B}_{123}{\bf R}_1 - {\bf B}_{13}{\bf R}_3 -{\bf B}_{23}{\bf R}_4)$ has full column-rank equal to $b_{3(1,2)}$. This is because if on the contrary, there exists a non-zero vector ${\bf z}$ such that $({\bf B}_{3(1,2)}-{\bf B}_{123}{\bf R}_1 - {\bf B}_{13}{\bf R}_3 -{\bf B}_{23}{\bf R}_4){\bf z}={\bf 0}$, then ${\bf B}_{3(1,2)}{\bf z}\in\langle[{\bf B}_{123},{\bf B}_{13},{\bf B}_{23}]\rangle$. Since ${\bf B}_{3(1,2)}$ has full column rank and ${\bf z}\neq {\bf 0}$, this means $\langle{\bf B}_{3(1,2)}\rangle$ has non-trivial intersection with $\langle[{\bf B}_{123},{\bf B}_{13},{\bf B}_{23}]\rangle$, contradicting their linear independence.
\item $({\bf B}_{3(1,2)}-{\bf B}_{123}{\bf R}_1 - {\bf B}_{13}{\bf R}_3 -{\bf B}_{23}{\bf R}_4){\bf a}\triangleq {\bf b}\neq {\bf 0}$, because of the previous observation and because ${\bf a}$ is a non-zero vector.
\item ${\bf b} \not \in \langle [{\bf B}_{123},{\bf B}_{23}] \rangle = \langle {\bf U}_{23} \rangle$, because if ${\bf b}  \in \langle [{\bf B}_{123},{\bf B}_{23}] \rangle$ then the non-zero vector ${\bf B}_{3(1,2)}{\bf a}={\bf b}+{\bf B}_{123}{\bf R}_1{\bf a}+{\bf B}_{13}{\bf R}_3{\bf a}+{\bf B}_{23}{\bf R}_4{\bf a}\in\langle[{\bf B}_{123},{\bf B}_{23}, {\bf B}_{13}]\rangle$, which is a contradiction because ${\bf B}_{3(1,2)},{\bf B}_{123},{\bf B}_{13},{\bf B}_{23}$ are linearly independent.
\item ${\bf b}\in\langle{\bf U}_3\rangle$. This follows from \eqref{eq:R5R7}. Specifically, since ${\bf B}_{3(1,2)}, {\bf B}_{123}, {\bf B}_{13}, {\bf B}_{23}$ are all submatrices of the basis matrix for $\langle{\bf U}_3\rangle$ according to {\it (P13)}, and ${\bf b}$ is their linear combination, this implies that ${\bf b}\in\langle{\bf U}_3\rangle$. 
\item   ${\bf b}\in\langle{\bf U}_2\rangle$. This also follows from \eqref{eq:R5R7} by similar reasoning.
\item From enumerated items 4 and 5 above, we have, ${\bf b}\in \langle {\bf U}_2 \rangle$ and ${\bf b}\in \langle{\bf U}_3 \rangle$, and thus ${\bf b} \in \langle {\bf U}_2 \rangle\cap\langle {\bf U}_3 \rangle=\langle {\bf U}_{23} \rangle$, which contradicts item 3.
\end{enumerate}

The contradiction proves the desired result that ${\bf R}_5$ has full column rank, i.e., it is an invertible square matrix. Similarly we can prove that ${\bf R}_7$ has full column rank, also an invertible square matrix. The last three steps hinge on this property.
\begin{enumerate}[wide, labelindent=1em, labelwidth=!, leftmargin =3em,  label={Step \arabic*}: ]
\setcounter{enumi}{11}
	\item Redefine ${\bf B}_{3(1,2)}$ as
	\begin{align}
		{\bf B}_{3(1,2)}^{\mbox{\footnotesize new}} = {\bf B}_{3(1,2)}^{\mbox{\footnotesize old}}-{\bf B}_{123}{\bf R}_1-{\bf B}_{13}{\bf R}_3 - {\bf B}_{23}{\bf R}_4.
	\end{align}
	\item Redefine ${\bf B}_{2(1,3)}$ as
	\begin{align}
		{\bf B}_{2(1,3)}^{\mbox{\footnotesize new}} = {\bf B}_{2(1,3)}^{\mbox{\footnotesize old}}{\bf R}_7+{\bf B}_{12}{\bf R}_2.
	\end{align}
	\item Redefine ${\bf B}_{1(2,3)}$ as
	\begin{align}
		{\bf B}_{1(2,3)}^{\mbox{\footnotesize new}} = {\bf B}_{1(2,3)}^{\mbox{\footnotesize old}}{\bf R}_5.
	\end{align}
\end{enumerate}
Since ${\bf R}_5$ and ${\bf R}_7$ are invertible square matrices, it follows by {\it (P6)} and {\it (P7)} that ${\bf B}_{1(2,3)}^{\mbox{\footnotesize new}}$, ${\bf B}_{2(1,3)}^{\mbox{\footnotesize new}}$ and ${\bf B}_{3(1,2)}^{\mbox{\footnotesize new}}$ are also feasible choices in Steps 5,6, and 7. Thus, {\it (P8)}-{\it (P19)} are still satisfied after ${\bf B}_{1(2,3)}$, ${\bf B}_{2(1,3)}$ and ${\bf B}_{3(1,2)}$ are replaced with ${\bf B}_{1(2,3)}^{\mbox{\footnotesize new}},{\bf B}_{2(1,3)}^{\mbox{\footnotesize new}},{\bf B}_{3(1,2)}^{\mbox{\footnotesize new}}$, respectively. However, because of the last three steps  and \eqref{eq:check}, {\it (P20)} is now satisfied as well with   ${\bf B}_{1(2,3)}^{\mbox{\footnotesize new}},{\bf B}_{2(1,3)}^{\mbox{\footnotesize new}},{\bf B}_{3(1,2)}^{\mbox{\footnotesize new}}$, i.e., 
\begin{align}
{\bf B}_{3(1,2)}^{\mbox{\footnotesize new}}={\bf B}_{2(1,3)}^{\mbox{\footnotesize new}}+{\bf B}_{1(2,3)}^{\mbox{\footnotesize new}}.
\end{align}
This concludes the proof of Lemma \ref{lemma:decomposition}.
$\hfill\square$

\section{Comparison of Lemma \ref{lemma:decomposition} to the Channel Decomposition of \cite{wang2017degrees}} \label{app:comparison}
Reference \cite[Chapter 3]{wang2017degrees} explores the DoF of a $3$ user MIMO broadcast channel where the transmitter has $m$ antennas, and the $k^{th}$ receiver has $n_k$ antennas, $k\in[3]$. The channel is specified by $Y_1 = H_1X+Z_1$, $Y_2 = H_2X+Z_2$ and $Y_3 = H_3X+Z_3$, where $X\in \mathbb{C}^{m\times 1}$ denotes the input of the channel and $Y_k \in \mathbb{C}^{n_k \times 1},k\in[3]$ denotes the output of the broadcast channel at the $k^{th}$ receiver. $H_k\in \mathbb{C}^{n_k\times m}, k\in[3]$ denotes the channel matrix between the transmitter and the $k^{th}$ receiver. $Z_1, Z_2$ and $Z_3$ are independent Gaussian noise vectors with zero mean and identity covariance matrix. There are independent messages desired by various subsets of receivers. As apparent from the high level description, the overall $3$ user MIMO BC DoF question does not allow any direct mapping to our $3$ user LCBC capacity question, e.g., the LCBC formulation has no notion of channel matrices, all users receive the same broadcast symbols,  whereas the MIMO BC problem has no notion of side-information or linear computations.

What connects the two problems is that they both involve a  decomposition of $3$ subspaces. In the LCBC, the $3$ subspaces of interest are $\langle{\bf U}_1\rangle, \langle{\bf U}_2\rangle, \langle{\bf U}_3\rangle$ as in Lemma \ref{lemma:decomposition}. In the MIMO BC the corresponding subspaces are $\langle N_1\rangle, \langle N_2\rangle, \langle N_3\rangle$, defined as the null spaces of the channel matrices $H_1^T$, $H_2^T$, $H_3^T$, respectively. The decompositions parallel each other very closely. Intuitively, even though the context surrounding these subspaces is quite different in each problem, the subspaces are similar mathematical objects, so  it makes sense that they should have similar properties, e.g., similar decompositions. The following table establishes a one-to-one correspondence of the subspace decompositions in the two settings. The  notation  in the right column of the table follows the definitions in \cite{wang2017degrees}.

\begin{table}[h]
\begin{center}
\begin{tabular}{c|c|c}
\toprule
& Lemma 1 & \cite{wang2017degrees} \\ \hline
Field & $\mathbb{F}_q$ & $\mathbb{C}$ \\ \hline
Dimension of universe & $d$ & $m$ \\ \hline
3 subspaces to decompose & $\langle {\bf U}_1 \rangle , \langle {\bf U}_2 \rangle , \langle {\bf U}_3 \rangle $ & $\langle N_1 \rangle , \langle N_2 \rangle , \langle N_3 \rangle $ \\ \hline
\multirow{11}{*}{Corresponding Bases} & ${\bf B}_{123}$ & $\emptyset$ \\ 
{} & ${\bf B}_{12}$ & $V_3$ \\ 
{} & ${\bf B}_{13}$ & $V_2$ \\ 
{} & ${\bf B}_{23}$ & $V_1$ \\ 
{} & ${\bf B}_{1(2,3)}$ & $V_{23X}$ \\ 
{} & ${\bf B}_{2(1,3)}$ & $V_{13X}$ \\ 
{} & ${\bf B}_{3(1,2)}$ & $V_{12X}$ \\ 
{} & ${\bf B}_{1c}$ & $V_{23R}$ \\ 
{} & ${\bf B}_{2c}$ & $V_{13R}$ \\ 
{} & ${\bf B}_{3c}$ & $V_{12R}$ \\ 
{} & $\emptyset$ & $V_{123}$ \\ \bottomrule
\end{tabular}
\end{center}
\caption{Correspondence between Lemma \ref{lemma:decomposition} and the channel decomposition in \cite{wang2017degrees} }
\label{tab:comparison}
\end{table}
A distinction is apparent in the first and last rows of the table. According to the first row, \cite{wang2017degrees} assumes that the three subspaces have empty intersection, whereas Lemma \ref{lemma:decomposition} accounts for this space with the basis representation ${\bf B}_{123}$. On the other hand, according to the last row, Lemma \ref{lemma:decomposition} assumes the complement of the span of three subspaces is empty, whereas \cite{wang2017degrees} accounts for this space with the basis representation $V_{123}$. This distinction arises mainly because in the LCBC setting the complement of span of the three spaces is uninteresting (data dimensions that are neither known nor desired by any user), whereas in the MIMO BC the intersection of the three nullspaces is similarly uninteresting (transmit dimensions that are nulled at every receiver), and each problem naturally eliminates the uninteresting spaces for simplicity. The distinction is not  significant, however, since the omitted spaces can be trivially included. A bit more significant distinction that is not apparent from the table is that as an additional feature, Lemma \ref{lemma:decomposition} chooses the basis matrices ${\bf B}_{1(2,3)}, {\bf B}_{2(1,3)}, {\bf B}_{3(1,2)}$ carefully to satisfy the matrix-sum condition ${\bf B}_{1(2,3)}+{\bf B}_{2(1,3)}={\bf B}_{3(1,2)}$, which facilitates code design in the LCBC. Such an explicit specialization of bases is not considered by the corresponding construction in \cite{wang2017degrees}. Setting aside these finer  distinctions, it is quite remarkable that the subspace decompositions in \cite{wang2017degrees} and Lemma \ref{lemma:decomposition}, obtained independently in different contexts, turn out to be in one-to-one correspondence. The  one-to-one correspondence constitutes strong evidence of the fundamental significance of the decomposition, as well as a verification of the same concept from two perspectives.

While the two decompositions are intuitively similar, there are several underlying technical details that prevent the direct application of the decomposition in \cite{wang2017degrees} to the LCBC problem. Note that the proof in \cite[Sec. 3]{wang2017degrees} relies on the existence of an orthogonal complement, i.e., a linear subspace that is both orthogonal and complementary to a given linear subspace. The existence of an orthogonal complement is guaranteed over $\mathbb{C}$, but not over $\mathbb{F}_q$. For instance, self-orthogonality is a prominent theme in error correction code design over finite fields. Removing the requirement of orthogonality and just using  any complement space instead does not automatically resolve the issue, because the complement space needs to be chosen carefully to achieve the correct alignment of spaces. The orthogonality of the complement space helps to achieve the desired alignment in the proof of \cite[Sec. 3]{wang2017degrees}. Over $\mathbb{F}_q$, since we are not guaranteed orthogonal complements, this choice is non-trivial). This is important because the alignment aspect of subspaces (any two subspaces contain the third) is  what makes the subspace decomposition non-trivial. Furthermore, the proof of correctness of the subspace decomposition in \cite[Sec. 3]{wang2017degrees} applies to \emph{almost all} spaces, since the argument relies on the values taken by ranks \emph{almost surely}. Over arbitrary $\mathbb{F}_q$, an `almost-surely' guarantee is not meaningful. Indeed, the proof of correctness of the decomposition is shown for \emph{all} realizations in Lemma \ref{lemma:decomposition}.\footnote{It is noteworthy that the proof of  Lemma \ref{lemma:decomposition} also extends to the field of complex numbers. The proof  is based on linear independence/dependence of subspaces which holds over both $\mathbb{F}_q$ and $\mathbb{C}$.}

Let us introduce a simple setup to further illustrate these points. Consider the following three (complex) channel matrices $H_1,H_2,H_3$.
\begin{align}
	H_1 = \begin{bmatrix}
		1&0&0\\0&1&0
	\end{bmatrix},~
	H_2 = \begin{bmatrix}
		0&0&1
	\end{bmatrix},~
	H_3 = \begin{bmatrix}
		1&1&0\\0&0&1
	\end{bmatrix}.
\end{align}
Denote $\mathcal{N}(A)$ as the nullspace of $A$, i.e., the set of $X$ such that $AX=0$. Let $N_k, k\in \{1,2,3\}$ be a basis (written in columns vectors) of $\mathcal{N}(H_k)$, i.e.,
\begin{align}
	N_1 =  \begin{bmatrix}
		0\\0\\1
	\end{bmatrix},~~
	N_2 =  \begin{bmatrix}1&0\\0&1\\0&0\end{bmatrix},~~
	N_3 =  \begin{bmatrix}1\\1\\0\end{bmatrix}.
\end{align}
Let $\langle A \rangle$ denote the linear subspace spanned by the columns of $A$. For example, $\langle N_1 \rangle = \mathcal{N}(H_1)$. $A\cap B$ denotes a basis that spans the subspace  $\langle A \rangle \cap \langle B \rangle$. In addition, if $\langle A \rangle$ is a subspace of $\langle U \rangle$, then  let $A^{\bot}_U$ denote a basis of the intersection of $\mathcal{N}(A^T)$ with $\langle U \rangle$. It follows that $A^T A^{\bot}_U = {\bf 0}$ and $\rk(A)+\rk(A^{\bot}_U) = \rk(U)$. Note that for $A$ defined in $\mathbb{C}$, $[A, A^{\bot}_U]$  spans $\langle U \rangle$. Thus, $A^{\bot}_U$ is the `orthogonal complement' of $A$ (within the subspace $\langle U\rangle$). In particular, $\langle A \rangle$ and $\langle A^\bot_U \rangle$ have no non-trivial intersection. However, this is not always true in finite fields, e.g, $A = [1,1]^T \in \mathbb{F}_2^{2\times 1}$, and $A^\bot = [1,1]^T = A$. 

Using the same notation as in \cite[Sec. 3]{wang2017degrees}, let us define
\begin{align}
	&H_{123} \triangleq \begin{bmatrix}
		H_1\\H_2\\H_3
	\end{bmatrix}=
	\begin{bmatrix}
		1&0&0\\0&1&0 \\0&0&1\\1&1&0\\0&0&1
	\end{bmatrix}, ~ H_{12} \triangleq \begin{bmatrix}
		H_1\\H_2
	\end{bmatrix},~
	H_{13} \triangleq \begin{bmatrix}
		H_1\\H_3
	\end{bmatrix},~
	H_{23} \triangleq \begin{bmatrix}
		H_2\\H_3
	\end{bmatrix}.
\end{align}

With these definitions, let us apply the decomposition method in \cite[Sec. 3.4.6]{wang2017degrees} on $N_1, N_2, N_3$. A summary of the steps is given below.
\begin{enumerate}
	\item Find $V_1 = N_2 \cap N_3$ as a basis of $\mathcal{N}(H_{23})$.
	\item Find $V_2 = N_1\cap N_3$ as a basis of $\mathcal{N}(H_{13})$.
	\item Find $V_3 = N_1\cap N_2$ as a basis of $\mathcal{N}(H_{12})$.
	\item Find $V_{13}$ as a basis of the orthogonal complement of $\langle[V_1,V_3]\rangle$ within $\langle N_2 \rangle$.
	\item Find $V_{23}$ as a basis of the orthogonal complement of $\langle[V_2,V_3]\rangle$ within $\langle N_1 \rangle$.
	\item Find $V_{12}$ as a basis of the orthogonal complement of $\langle[V_1,V_2]\rangle$ within $\langle N_3 \rangle$.
	\item Find independent bases $V_{13X},V_{12X}, V_{12R}, V_{13R}, V_{23R}$ by (3.43) - (3.47) of \cite{wang2017degrees}.
\end{enumerate}
 
\noindent By Steps 1-3,
\begin{align}
	&V_1 = N_2\cap N_3 = \begin{bmatrix}1&1&0\end{bmatrix}^T, \\
	&V_2 = N_1\cap N_3 = [~],\\
	&V_3 = N_1\cap N_2 = [~].
\end{align}
At this point, note that the next step, which is to construct $V_{13}$ does not work. According to \cite{wang2017degrees} (3.38),
\begin{align}
V_{13} = (H_{123}^T H_{123})^{-1} H_{123}^T (H_{123}[V_1,V_3])^\bot_{H_{123}N_2},
\end{align}
but $(H_{123}^T H_{123})$ is not invertible although $H_{123}$ has full column rank $3$. This is because the orthogonal complement of $H_{123}$ has nontrivial intersection with itself. This can happen in $\mathbb{F}_q$ but not in $\mathbb{C}$.

Alternatively, if we use the implicit definition, \cite[(3.37)]{wang2017degrees}, we may avoid the inversion of $H_{123}^TH_{123}$, but a similar problem will again emerge. \cite[(3.37)]{wang2017degrees}  requires that
\begin{align}
	&H_{123}V_{13} = (H_{123}[V_1,V_3])^\bot_{H_{123}N_2},
\end{align}
which is a basis of the orthogonal complement of $\langle H_{123}[V_1,V_3] \rangle$ within the subspace $\langle H_{123}N_2 \rangle$.
Note that 
\begin{align}
	H_{123}[V_1,V_3] = \begin{bmatrix}
		1&1&0&0&0
	\end{bmatrix}^T
\end{align}
and 
\begin{align}
	H_{123}N_2 = \begin{bmatrix}
		1&1&0&0&0\\1&0&0&1&0
	\end{bmatrix}^T.
\end{align}
It is readily verified that the only solution to $(H_{123}[V_1,V_3])^\bot_{H_{123}N_2} = [1,1,0,0,0]^T$. Therefore, we obtain that $H_{123}V_{13} = [1,1,0,0,0]^T$. This gives us $V_{13} = [1,1,0]^T = V_1$, which is linearly \emph{dependent} on $V_1$. However, $V_{13}$ is required to be linearly independent of $V_1$ in Step 4.

Next let us consider Step 7. At a high level, the motivation of Step 7 is that the three spaces $V_{13}, V_{12}$ and $V_{23}$ are not independent in general and therefore a finer decomposition is needed. In \cite{wang2017degrees},  $6$ subspaces are introduced, namely $V_{13X}, V_{13R}, V_{12X}, V_{12R}, V_{23X}, V_{23R}$, so that $[V_{13X}, V_{13R}]$ spans $\langle V_{13} \rangle$, $[V_{12X}, V_{12R}]$ spans $\langle V_{12} \rangle$, and $\langle [V_{23X}, V_{23R}] \rangle$ spans $\langle V_{12} \rangle$. Note that these subspaces identified by the algorithm have such properties that 
\begin{itemize}
	\item $[V_{13X}, V_{13R}, V_{12X}, V_{12R}, V_{23R}]$ are independent and span $\langle [V_{13}, V_{12}, V_{23}] \rangle$. Besides, $[V_{**X}, V_{**R}]$ spans $\langle V_{**} \rangle$, for $** \in \{13,12,23\}$.
	\item  In addition, $V_{23X}$ is linearly representable by $[V_{13X}, V_{12X}]$, i.e., $\langle V_{23X} \rangle \subset \langle [V_{13X}, V_{12X}] \rangle$. Also, $\langle V_{12X} \rangle \subset \langle [V_{13X}, V_{23X}] \rangle$ and $\langle V_{13X} \rangle \subset \langle [V_{12X}, V_{23X}] \rangle$.
	\item $V_{13X}, V_{12X}, V_{23X}$ are aligned in a way such that $\langle H_{1}V_{13X}\rangle = \langle H_{1}V_{12X} \rangle$, $\langle H_{2}V_{12X}\rangle = \langle H_{2}V_{23X} \rangle$ and $\langle H_{3}V_{13X}\rangle = \langle H_{3}V_{23X} \rangle$.
\end{itemize}
The critical alignment is the second one, i.e., we need $V_{13X}, V_{12X}, V_{23X}$ such that each one is contained in the span of the other two. Let us see what happens if we replace the `orthogonal complement' space (which may not exist over $\mathbb{F}_q$) with any `complement' space (which do exist over $\mathbb{F}_q$). The following toy example shows that simply replacing `orthogonal complement' with `any complement' may not work. Suppose we are given,
\begin{align}
	N_1 = \begin{bmatrix}
		1 & 0 \\ 0 & 1 \\ 0 & 0 
	\end{bmatrix},~~
	N_2 = \begin{bmatrix}1&0\\1&1\\0&1\end{bmatrix},~~
	N_3 = \begin{bmatrix}1\\1\\1\end{bmatrix},
\end{align}
with entries all defined in $\mathbb{F}_2$.
It follows by definition that,
\begin{align}
	&V_1 = N_2 \cap N_3 = [~], ~ V_2 = N_1\cap N_3 = [~],~ V_3 = N_1\cap N_2 = \begin{bmatrix}
		1 \\ 1 \\0
	\end{bmatrix}.
\end{align}
Next, say we choose the complements (not necessarily orthogonal) as,
\begin{align} \label{eq:V131223}
	V_{13} = \begin{bmatrix}
		0 \\ 1 \\ 1
	\end{bmatrix}, ~~
	V_{12} = \begin{bmatrix}
		1 \\ 1 \\ 1
	\end{bmatrix}, ~~
	V_{23} = \begin{bmatrix}
		0 \\ 1 \\ 0
	\end{bmatrix},
\end{align}
so that $[V_1, V_3, V_{13}]$ span $\langle N_2 \rangle$, $[V_1, V_2, V_{12}]$ span $\langle N_3 \rangle$, and $[V_2, V_3, V_{23}]$ span $\langle N_1 \rangle$. Such an attempt to translate Steps $1$ -- $6$  to the finite field case does not work because now we see that $[V_1, V_2, V_3, V_{13}, V_{12}, V_{23}]$ are not linearly independent. In particular, with this choice there is no non-trivial $V_{13X}, V_{23X}, V_{12X}$ so that any one is contained in the span of the other two. What this shows is that the complement spaces $V_{13},V_{12},V_{23}$ need to be chosen carefully. The proof in \cite{wang2017degrees} does not face this problem, because orthogonal spaces and complement spaces are compatible for complex numbers and therefore calculations of orthogonal complements help to identify the appropriate $V_{13}, V_{12}, V_{23}$. Over $\mathbb{F}_q$ this does not work. Fortunately, there does exist a solution to $[V_{13}, V_{12}, V_{23}]$ so that such non-trivial $V_{13X}, V_{23X}, V_{12X}$ can be found (see Steps $11$--$14$ in Appendix \ref{proof:decomposition} for the  details of this key aspect of the proof in the general case). Indeed, with a more careful choice of subspaces we have,
\begin{align} \label{eq:V131223_refined}
	V_{13} = \begin{bmatrix}
		0 \\ 1 \\ 1
	\end{bmatrix}, ~~
	V_{12} = \begin{bmatrix}
		1 \\ 1 \\ 1
	\end{bmatrix}, ~~
	V_{23} = \begin{bmatrix}
		1 \\ 0 \\ 0
	\end{bmatrix},
\end{align}
\begin{align} 
	V_{13X} = \begin{bmatrix}
		0 \\ 1 \\ 1
	\end{bmatrix}, ~~
	V_{12X} = \begin{bmatrix}
		1 \\ 1 \\ 1
	\end{bmatrix}, ~~
	V_{23X} = \begin{bmatrix}
		1 \\ 0 \\ 0
	\end{bmatrix},
\end{align}
and then
\begin{align}
	V_{13R} = V_{12R} = V_{23R} = [~].
\end{align}
Thus, following our proof we successfully found the $V_{13X}, V_{13R}, V_{12X}, V_{12R}, V_{23X}, V_{23R}$ that satisfy the desired (first two) properties described in Step 7.

\bibliographystyle{IEEEtran}
\bibliography{3CBC}

\end{document}